\documentclass{lmcs}
\pdfoutput=1
\usepackage[utf8]{inputenc}

\usepackage{lastpage}
\lmcsdoi{22}{4}{20}
\lmcsheading{}{\pageref{LastPage}}{}{}%
{Oct.~18,~2024}{Nov.~17,~2026}{}


\usepackage{hyperref}

\usepackage{subcaption}
\usepackage{amssymb}
\usepackage{listings}
\definecolor{dkgreen}{rgb}{0,0.6,0}
\definecolor{ltblue}{rgb}{0,0.4,0.4}
\definecolor{dkviolet}{rgb}{0.3,0,0.5}

\newcommand{\coqcode}[1]{{\textsf{\textup{#1}}}}

\newcommand{\myexists}[2]{
{\color{dkgreen}\mathtt{exists}}
\ #1 ,\ 
#2
}
\newcommand{\myforall}[2]{
{\color{dkgreen}\mathtt{forall}}
\ #1 ,\ 
#2
}

\usepackage{pgfplots}
\pgfplotsset{compat=1.18}
\usepgfplotslibrary{groupplots}

\newcommand{\IN}{\mathbb{N}}
\newcommand{\IR}{\mathbb{R}}
\newcommand{\bR}{\mathbf{R}}

\newcommand{\IQ}{\mathbb{Q}}
\newcommand{\Casm}{\mathbf{Asm}(\mathcal{K}_2,\mathcal{K}_2^{\mathrm{rec}})}

\newcommand{\Baire}{\IN^\IN}

\newcommand{\cff}{\mathit{ff}}
\newcommand{\ctt}{\mathit{tt}}
\newcommand{\abs}[1]{\left|#1\right|}

\newcommand{\pic}{\mathrm{pic}}
\newcommand{\pto}{\rightharpoonup}
\newcommand{\dom}{\mathrm{dom}}

\newcommand{\bN}{\mathbf{N}}

\newcommand{\bK}{\mathbf{K}}

\newcommand{\dR}{\mathsf{R}}

\newcommand{\dK}{\mathsf{K}}
\newcommand{\dS}{\mathsf{S}}
\newcommand{\dN}{\mathsf{N}}
\newcommand{\dZ}{\mathsf{Z}}
\newcommand{\dD}{\mathsf{D}}
\newcommand{\dB}{\mathsf{B}}
\newcommand{\dX}{\mathsf{X}}

\newcommand{\kleene}{\mathsf{K}}
\newcommand{\Type}{{\color{ltblue}\mathtt{Type}}}
\newcommand{\Prop}{{\color{ltblue}\mathtt{Prop}}}
\newcommand{\ttrue}{\mathtt{True}}
\newcommand{\tfalse}{\mathtt{False}}

\newcommand{\all}[1]{\mathrm{\Pi}( #1).\ }
\newcommand{\lall}[1]{{\forall}( #1).\ }
\newcommand{\allx}{\mathrm{\Pi}}
\newcommand{\some}[1]{\mathrm{\Sigma}( #1).\ }
\newcommand{\somex}{\mathrm{\Sigma}}
\newcommand{\csome}[1]{\exists(#1 ).\ }
\newcommand{\usome}[1]{\exists!(#1 ).\ }

\newcommand{\lam}[1]{\lambda ( #1 ).\ }

\newcommand{\ctrue}{\textsf{true}}
\newcommand{\cfalse}{\mathsf{false}}
\newcommand{\cundef}{\mathsf{bot}}

\newcommand{\Csubset}{\mathrm{subset}}

\newcommand{\sopen}{\mathrm{open}}
\newcommand{\sclosed}{\mathrm{closed}}
\newcommand{\scompact}{\mathrm{compact}}
\newcommand{\sovert}{\mathrm{overt}}

\newcommand{\isopen}{\mathrm{is\_open}}
\newcommand{\isclosed}{\mathrm{is\_closed}}
\newcommand{\iscompact}{\mathrm{is\_compact}}
\newcommand{\isovert}{\mathrm{is\_overt}}
\newcommand{\islocated}{\mathrm{is\_located}}
\newcommand{\istotallybounded}{\mathrm{is\_totally\_bounded}}
\newcommand{\iscomplete}{\mathrm{is\_complete}}

\newcommand{\ball}{\mathrm{B}}
\newcommand{\mval}{\mathsf{M}}

\newcommand{\upc}[1]{\lceil #1 \rceil}
\newcommand{\downc}[1]{\lfloor #1 \rfloor}
\newcommand{\dist}[2]{\| #1 - #2 \|}

\newcommand{\distxop}{\mathrm{d}_X}
\newcommand{\distx}[2]{\distxop \  #1\ #2}
\newcommand{\defineds}[1]{#1 \!\downarrow}
\newcommand{\definedx}{\downarrow}
\newcommand{\undefinedx}{\uparrow}
\newcommand{\undefineds}[1]{#1 \!\uparrow}
\newcommand{\skinv}{\mathsf{KtoS}}
\newcommand{\fproj}{\mathsf{proj}_1}
\newcommand{\creal}{\mathtt{CReal}}
\newcommand{\ckleenean}{\mathtt{CKleenean}}
\newcommand{\softwarename}[1]{\textrm{#1}}

\newcommand{\elim}{\underline{\lim}}
\newcommand{\restr}[2]{#1\restriction_{#2}}
\newcommand{\listof}[1]{[#1]}
\newcommand{\equivto}{ \mathrel{\coloneqq} }
\newcommand{\definedto}[2]{#1 \searrow #2 }

\newcommand{\lext}[1]{\overline{#1}}

\newcommand{\fname}[2]{\href{https://github.com/holgerthies/coq-aern/blob/bc11353f450cf866b47c3985eee6150a5f99cf00#2}{\textsf{#1}}}
\newcommand{\fnamefull}[2]{\href{#2}{\textsf{#1}}}

\newcommand{\coqref}[3]{\href{https://github.com/holgerthies/coq-aern/blob/bc11353f450cf866b47c3985eee6150a5f99cf00#2\#L#3}{\textsf{#1}}}

\newcommand{\inlinedef}[1]{\textbf{\emph{#1}}}

\newcommand{\liff}{\mathrel{\leftrightarrow}}

\theoremstyle{plain}
 \newtheorem{axiom}[thm]{Axiom}

\begin{document}
\title[Hyperspaces over Abstract Exact Real Numbers]
{Formalizing Hyperspaces and Operations on Subsets of Polish Spaces over Abstract Exact Real Numbers}

\author[M.~Konečný]{Michal Konečný\lmcsorcid{0000-0003-2374-9017}}[a]
\author[S.~Park]{Sewon Park\lmcsorcid{0000-0002-6443-2617}}[b,c]
\author[H.~Thies]{Holger Thies\lmcsorcid{0000-0003-3959-0741}}[c]

\address{Aston University, Birmingham, UK}	%
\email{m.konecny@aston.ac.uk}  %

\address{Institute of Mathematics, Physics and Mechanics, Ljubljana, Slovenia}	
\email{sewon@sewonpark.com}
           
\address{Kyoto University, Kyoto, Japan}	%
\email{thies.holger.5c@kyoto-u.ac.jp}  %

\begin{abstract}
Building on our prior work on axiomatization of exact real computation by formalizing nondeterministic first-order partial computations over real and complex numbers  in a constructive dependent type theory, we present a framework for certified computation on hyperspaces of subsets by formalizing various higher-order data types and operations.

We first define open, closed, compact and overt subsets for generic spaces in an abstract topological way that allows short and elegant proofs with computational content coinciding with standard definitions in computable analysis and constructive mathematics.
From these proofs we can extract programs for testing inclusion, overlapping of sets, et cetera. 

To enhance the efficiency of the extracted programs, we then focus on Polish spaces, where we give more efficient encodings based on metric properties of the space.
As various computational properties depend on the continuity of the encoding functions, we introduce a nondeterministic version of a continuity principle which is natural in our formalization and valid under the standard type-2 realizability interpretation.
Using this principle we further derive the computational equivalence between the generic and the metric encodings.

We  prove that many operations on subsets preserve an efficient encoding inspired by standard definitions from constructive mathematics.
This leads to the development of a small calculus to build new subsets by operating on given ones, including computing limits of sequences that converge with respect to the Hausdorff metric.

Our theory is fully implemented in the Coq proof assistant.
From proofs in this Coq formalization, we can extract certified programs for error-free operations on subsets. 
As an application,  we provide a function that constructs fractals in Euclidean space, such as the Sierpinski triangle, from iterated function systems using the limit operation.
The resulting programs can be used to draw such fractals up to any desired resolution.

\end{abstract}
\maketitle
\tableofcontents
\section{Introduction}

In exact real computation, based on the theoretical framework of computable analysis \cite{kreitz1985theory,w00}, real numbers are expressed by infinite representations that are manipulated exactly without rounding errors. 
In most software implementations of exact real computation \cite{Ariadne,irram,konecny2008AERN,clerical_ocaml}, real numbers are presented as an abstract data type, hiding tedious computations related to the infinite representation from users. 
Hence, to users, real numbers closely resemble familiar abstract real numbers \cite{DBLP:journals/mlq/Hertling99,escardoSimpson2001}.
This also makes the approach particularly well-suited for formal verification \cite{park2016foundation,clerical}.

In recent work \cite{konevcny2022extracting}, the authors have been working on an axiomatic formalization of exact real computation in constructive type theory and its implementation as the Coq library \softwarename{cAERN}. 
\softwarename{cAERN} aims to model real numbers similarly to how they are observed by users of exact real computation software. That is, instead of being constructed explicitly, e.g., as a setoid of sequences of approximations, real numbers are axiomatized such that two reals are propositionally equal, not just equivalent, when they represent the same numbers; cf. \cite{BLM16,cruz2004c,steinberg2019quantitative,DBLP:journals/apal/BergerT21,lmcs:11550}. 
This can be regarded as a modification of classical real numbers by replacing uncomputable operations, e.g., rounding to the nearest integer, and computably invalid axioms, e.g., the law of trichotomy, with computable variants. 
As the axiomatization is representation-irrelevant, reals appear similar to classical real numbers, allowing many classical results to be transported into this setting \cite[Section~6]{konevcny2022extracting}. At the same time, by allowing only computably valid axioms, it admits a program extraction mechanism built on top of an existing exact real computation framework.

The extraction is achieved by mapping the axiomatic real type $\dR$ to the abstract data type of real numbers $\creal$ in \softwarename{AERN}, a Haskell library for exact real number computation developed by one of the authors \cite{konecny2008AERN}. Thus, we do not need to focus on the concrete representation of real numbers and its efficient implementation, which is a challenging task on its own (see, e.g., \cite{MR2137733}). 
More concretely, in our system, a proof of a theorem of the form
\[
\all{x : \dR}\mval\some{y: \dR}P\ x\ y 
\]
yields a nondeterministic \softwarename{AERN} function of type $\creal \to \creal$. Here, $\mval :\Type \to \Type$ formalizes nondeterminism, which is an inevitable effect of identifying real numbers extensionally \cite{LUCKHARDT1977321}.\footnote{Nondeterminism in this context is also called multivaluedness \cite{w00} or non-extensionality \cite{10.1007/978-3-031-14788-3_5}.} 
Although many exact-real computation libraries are available, \softwarename{AERN} stands out as a suitable target library for this project.
This is supported by the fact that Haskell is one of the languages supported natively by Coq's program extraction mechanism, and \softwarename{AERN} has scored well in a benchmark of Haskell exact-real computation libraries \cite{haskell-reals-comparison-2022}.
Additionally, the abstract types and operations provided by \softwarename{AERN} are generic and flexible enough to be represented well by the axioms in the \softwarename{cAERN} library. 

In \cite{konevcny2022extracting} we provide various examples of how to build interesting first-order computations within the theory.
Many applications in exact real computation, e.g., reachability in dynamical systems, integrals, and solving differential equations, deal, however, with higher-order data types such as hyperspaces of functions and subsets \cite{collins2020}. 
In practice, such higher-order data types are defined based on abstract real numbers, relying on the underlying programming language's higher-order functionality \cite{DBLP:conf/ershov/BrausseKM15,kawamura2018parameterized,selivanova2021exact}. 
To argue the correctness of the defined data types, careful comparison of the implementations with the findings in constructive and computable mathematics \cite{BRATTKA200343,DBLP:conf/csl/GassnerP021,DBLP:journals/aml/IljazovicJ24,DBLP:journals/jla/CoquandS10,diener2008compactness} is necessary, which seems challenging to do for each implementation. Therefore, it is desirable to have a set of data types and their operations that are proven correct and can be used across various applications.

In this paper, we extend the previous axiomatic formalization to hyperspaces. 
For an arbitrary type, we define a type of classical subsets to which we attach computational meaning by introducing new types of open, closed, compact, and overt subsets from constructive mathematics. 
For more practical applications, we then restrict the arbitrary type to be a \inlinedef{Polish space}, i.e., a \emph{complete, separable metric space}, where we can define specialized encodings enabling more efficient computations. 
To ensure the correctness of these metric versions of the topological notions, we prove their (computational) equivalence to the general encodings.
The equivalence proofs depend on the fact that every computable function is continuous and thus requires a continuity principle in our theory.
Hence, we include a new axiom stating that any function \emph{nondeterministically} admits a modulus of continuity functional. We further show how this operation can be naturally realized in exact real computation software, specifically how to extract our continuity principle into a simple operation in \softwarename{AERN}.

In Euclidean space, a rather intuitive notion of computing a subset is to draw the set with arbitrary precision using arbitrarily small boxes. 
This idea can be extended to metric spaces, corresponding to the space being \emph{totally bounded}.
A space being totally bounded and complete is often referred to as \inlinedef{Bishop-compactness} \cite{bishop1967foundations}. Similar notions have been considered in computable analysis as representations of compact sets \cite{BRATTKA199965}. 
We formalize the proof that in our setting for a Polish space, Bishop-compactness constructively coincides with compact-overtness, and we take the type of compact-overt subsets as an important data type on which we prove various efficient operations that benefit from the total boundedness property, including limit operations based on the Hausdorff distance. 
As an application, we define a certain type of fractals  in Euclidean space given by iterated function systems as the limit of a recurrence relation on compact-overt sets. 
From the proof we extract a program that draws the fractal exactly up to any desired resolution.

\subsection{Overview of the Project}\label{ssec:overview} 
All results in this paper have been implemented in the Coq proof assistant as an extension of the \softwarename{cAERN} library. 
The Coq formalization itself is considered a significant contribution to this work and is therefore an integral part of the publication. 
The version associated with the paper is marked as release \cite{caern}, and most lemmas and definitions in the paper provide references to the corresponding Coq statements.
In the digital version, these references are hyperlinks that direct to the corresponding code segments in the GitHub repository.
The GitHub repository also includes installation instructions for the library and some additional examples.

The proofs in the paper are meant to convey the essence of the formal proofs in the formalization. 
Thus, although there might be more elegant proofs, the presentation is meant to reflect how we proved them with the tools available in our Coq library.
On the other hand, we do not aim to rephrase the full formal proofs. 
Some technical details have been simplified and a few simpler proofs even omitted completely.
The interested reader can check the details of these proofs by following the hyperlinks to the code segment in the GitHub repository.

The new formalization presented here consists of approximately 9000 lines of code out of approximately 30000 lines of code for the complete \softwarename{cAERN} library.

The file structure of the repository closely resembles the structure of the paper.
The paper is organized as follows. 
In \autoref{sec:preliminaries} we summarize the background system of our formalization and introduce notations used throughout the paper. 
In \autoref{sec:continuity} we define the continuity principle and prove some of its consequences.
We define open, closed, compact and overt subsets of an arbitrary type in \autoref{sec:subsets}.
We consider the case when the arbitrary type is a Polish space in \autoref{sec:polish}, outline some details of the code extraction in \autoref{s:extraction}, and finally show some applications generating certified drawings of subsets of two-dimensional Euclidean spaces in \autoref{sec:examples}.

The location of the corresponding files is as follows.
\begin{description}[style=nextline,font=\normalfont]
    \item[\fname{Hyperspace/Sierpinski.v}{/formalization/Hyperspace/Sierpinski.v}] for the properties of Sierpinski space described in \autoref{sec:preliminaries}
    \item[\fname{Hyperspace/Continuity.v}{/formalization/Hyperspace/Continuity.v}] for the results on Continuity in \autoref{sec:continuity}
    \item[\fname{Classical/Subsets.v}{/formalization/Classical/Subsets.v}] defines classical subsets and their operations from the beginning of \autoref{sec:subsets}
    \item[\fname{Hyperspace/Subsets.v}{/formalization/Hyperspace/Subsets.v}] for the definitions and properties of open, closed, compact and overt sets described in \autoref{sec:subsets}
    \item[\fname{Hyperspace/MetricSubsets.v}{/formalization/Hyperspace/MetricSubsets.v}]
    for the definition of metric spaces, the subset types on them and the proofs of equivalences outlined in \autoref{sec:polish}. 
    \item[\fname{Hyperspace/EuclideanSubsets.v}{/formalization/Hyperspace/EuclideanSubsets.v}, \fname{Hyperspace/SimpleTriangle.v}{/formalization/Hyperspace/Simpletriangle.v}, \fname{Hyperspace/SierpinskiTriangle.v}{/formalization/Hyperspace/SierpinskiTriangle.v}] and \fname{Hyperspace/SierpinskiTriangleLimit.v}{/formalization/Hyperspace/SierpinskiTriangleLimit.v} contain the examples of \autoref{sec:examples}.  Haskell code for the extracted programs can be found in the \fnamefull{extracted-examples}{https://github.com/holgerthies/coq-aern/tree/bc11353f450cf866b47c3985eee6150a5f99cf00/extracted-examples} folder.
\end{description}
Apart from the list of files above, a few changes have also been made to other parts of the library, e.g., to include the new axioms in the appropriate files, some auxiliary statements about real numbers that were needed in the formalization, etc.

Although the paper uses type-theoretical notations instead of the Coq syntax, it should be straightforward to translate the statements in the paper to the corresponding Coq code.
For example, $\csome{x:A} B\ x$ corresponds to $(\myexists{x : A}B\ x) : \Prop$ and $\some{x:A}B\ x$ corresponds to either
$\{x : A \ \&\  B\ x\} : \Type$ or $\{x : A \mid  B\ x\} : \Type$.
The $\allx$-type, $\all{x: A}B\ x$, is expressed by $\myforall{x  : A}{B \ x}$ in Coq. 
Further, the notations used in the paper are slightly simplified, e.g.,  in the Coq source code we use $\hat{}\ \coqcode{Real}$ instead of $\dR$, etc.
These minimal modifications should enhance readability of the paper while still being close enough to the source code to easily map the statements in the paper to the Coq formalization.

\subsection{Related Work}
Computations with subsets are central to many areas of constructive mathematics and computable analysis, and already appear in Bishop's Foundations of Constructive Analysis \cite{bishop1967foundations}. 
A crucial notion in Bishop's framework is that of a compact space,
defined as a metric space that is \emph{complete} and \emph{totally bounded}.
The relationship between various classically equivalent notions of compactness has been a central theme in constructive mathematics \cite{bridges1999sequential,10.1007/3-540-45793-3_7}.
An overview of compactness in constructive mathematics can be found, for example, in Diener's PhD thesis \cite{diener2008compactness}.

In computable analysis, computations on subsets of various spaces are modeled through effective representations.
Representations for subsets of Euclidean space have first been studied by Brattka and Weihrauch \cite{BRATTKA199965} and later extended by Brattka and Presser \cite{BRATTKA200343} to subsets of computable metric spaces.
In the meantime, a rich theory of computability of subsets of metric spaces has been established and many deep results have been shown.
Iljazovi\'{c} and Kihara \cite{Iljazovi2021} provide a modern overview on subsets of metric spaces in the context of computable analysis. 

These concepts have been extended to generic represented spaces \cite{DBLP:phd/dnb/Schroder03a}, a detailed overview of which can be found in  \cite{Pauly2016OnTT} and \cite{collins2020}. 
It should be noted that there is some discrepancy in nomenclature for the same concepts across different works in the literature. In this paper, we primarily follow the terminology established in \cite{Pauly2016OnTT} (cf. also the discussion in op. cit.).

The idea of deriving topological notions from computational properties  has been introduced in Escardó's synthetic topology \cite{escardo2004synthetic} and similar ideas have been explored in Taylor's abstract Stone duality \cite{TaylorP:lamcra,TaylorP:fofct}.

Constructive dependent type theory is well-suited for the implementation of these concepts and thus naturally some implementations in proof assistants already exist.
One notable example is the \softwarename{CoRN} library \cite{cruz2004c}, a substantial formalization of constructive reals in the Coq proof assistant.
Subsequent work deals with the formalization of metric spaces and compact subsets \cite{DBLP:journals/corr/abs-0806-3209, conorPhD}.
As already mentioned in the introduction,  our library follows a different approach from \softwarename{CoRN}, with the most significant difference being the treatment of equality of real numbers, avoiding the rather complicated setoid constructions.
We think that this approach is more accessible to typical mathematicians and computer scientists who are less familiar with these constructions.

Most results in this paper are not mathematically novel, but established facts in constructive analysis and formal topology \cite{FOURMAN1982107,Sambin1987}.
For example, the equivalence of Bishop-compactness and compact-overtness has (although in a slightly different setting) already been shown in \cite{DBLP:journals/jucs/CoquandS05} (cf. also \cite{spitters2010locatedness}).
However, our proofs are fully formalized in Coq and the integration within the cAERN library enables us to extract efficient programs from proofs in our theory.
We therefore think that the implementation provides new perspectives on computational applications and can inspire more practical use cases, bridging the gap between abstract theory and concrete computation.
Furthermore, the tools used in the cAERN library are rather elementary, building only on basic facts from computability theory.
Thus, we believe that the proofs are more accessible to non-experts in constructive mathematics.

\subsection{Preliminary Version}
This paper is an extended version of the conference paper \cite{DBLP:conf/mfcs/Konecny0T23} which appeared in the proceedings of the 48th
International Symposium on Mathematical Foundations of Computer Science (MFCS 2023).
While some of the fundamental theory has already been laid out in the conference paper, in the meantime we have substantially extended the work and thus most of the paper and the Coq formalization have been fully rewritten.
Most importantly, in  \cite{DBLP:conf/mfcs/Konecny0T23}  we only consider Euclidean spaces, while in this paper we extend the theory to generic Polish spaces as presented in \autoref{sec:polish}.
This extension  required considering a more general form of the continuity axiom (see \autoref{sec:continuity} for a comparison of the principles).
Thus, \autoref{sec:continuity} and \autoref{sec:polish} have been rewritten from scratch and only have minor overlaps with the conference version.
The biggest overlap is in \autoref{sec:examples}, as we kept the examples from the conference version and only made minor changes to the text to adapt to the generalized theory.
We also present some details of the Coq formalization in \autoref{ssec:overview} and the code extraction procedure in \autoref{s:extraction} which was missing in the conference version.
Further, the equivalence of compact-overt and our definition of totally bounded and complete was only stated as a meta-theorem with an informal proof as an appendix in \cite{DBLP:conf/mfcs/Konecny0T23}. 
Along with the generalization from Euclidean space to general Polish spaces, in this paper we present a fully formalized proof of this theorem in the Coq proof assistant, which is now (together with the theory on metric spaces) a significant part of the \softwarename{cAERN} library.

\section{Preliminaries}\label{sec:preliminaries}
We briefly introduce the setting where we build our formalization of hyperspaces, which is a constructive dependent type theory extended with computationally valid axioms. While this section is mostly an excerpt of   \cite{konevcny2022extracting}, where the full set of axioms can be found, we suggest some new axioms as well.

\subsection{Type Theory and Realizability}
We work in a simple constructive dependent type theory with basic types $\mathsf{0}, \mathsf{1}, \dB, \dN, \dZ$, (empty, unit, boolean, natural number, integer type, respectively), a universe of classical propositions $\Prop$, and a universe of types $\Type$.
We assume that the identity types $=$ are in $\Prop$ and that $\Prop$ is a type universe closed under $\to, \times, \lor, \exists, \allx$, containing two
type constants $\ttrue, \tfalse : \Prop$ which are the unit and the empty type, respectively. 

We denote by $P \lor Q : \Prop$ the classical fact that $P$ or $Q$ holds, whereas the usual sum type $P + Q : \Type$ denotes constructive evidence of which of $P$ or $Q$ holds.
Similarly, when we have a family of classical propositions $P : X \to \Prop$, the type $\csome{x : X}P\, x$ belonging to $\Prop$ denotes the classical existence of $x : X$ satisfying $P\, x$, while the ordinary dependent pair type (also called $\somex$-type), $\some{x : X}P\, x$ belonging to $\Type$ denotes the constructive existence.
Note that the sum types and $\Sigma$-types are defined for any types and type families, not only for those in $\Prop$; in particular, in $P + Q$ and $\some{x:X}P\, x$ the components $P, Q$ and $P\, x$ need not be in $\Prop$.

We write $\lall{x:X}P\, x$ instead of
$\all{x:X}P\, x$ when the intended meaning is a logical formula rather than a dependent function, for better readability; the two denote precisely the same $\Pi$-type, and in neither case is $P\, x$ required to be in $\Prop$.

We assume enough axioms that make $\Prop$ indeed a universe of classical propositions including the law of excluded middle 
\[\lall{P:\Prop} P \lor \neg P\] 
where $\neg P$ is defined by 
$P \to \tfalse$
and
propositional extensionality 
\[\lall{P, Q : \Prop} (P \to Q) \to (Q \to P) \to P = Q.\]
These axioms allow us to reason classically when we are dealing with
classical statements such as $\csome{x : X}P\, x : \Prop$ without harming the constructivity elsewhere in $\Type$.
We further assume a general function extensionality:
\[\lall{P: X \to \Type}\lall{f, g : \all{x: X}P\ x}
 (\lall{x : X} f\ x = g \ x) \to f=g.\]

To validate the axiomatic formalization,
we consider the category $\Casm$ of assemblies over Kleene's second algebra
and computable maps between them as the model of our type theory.

An assembly $\mathbf{X}$ is a pair of a set $|\mathbf{X}|$ and a binary relation $\Vdash_\mathbf{X} \subseteq \Baire \times |\mathbf{X}|$ called a (multi-) representation of $|\mathbf{X}|$ such that \[\lall{x \in |\mathbf{X}|}\csome{\varphi \in \Baire}\varphi \Vdash_\mathbf{X} x\] 
holds. When $\varphi \Vdash_\mathbf{X} x$, we say $\varphi$ is a realizer of $x$. Given two assemblies $\mathbf{X}, \mathbf{Y}$, a mapping $f : |\mathbf{X}| \to |\mathbf{Y}|$ is computable 
when there is a type-2 machine that for all $x\in |\mathbf{X}|$, transforms each realizer of $x$ to a realizer of $f(x)$. A type in our type theory is interpreted as an assembly and a term in our type theory is interpreted as a computable mapping. Therefore, the soundness of an axiom can be validated by the existence of a computable point in the underlying set of the corresponding assembly.
See \cite{van2008realizability} for details about assemblies and \cite{10.1007/BFb0022273} and \cite[\S~3.6.2]{BIRKEDAL20002} for the interpretation of type theories in $\Casm$.

\subsection{Kleenean and Sierpinski Types}
 As real numbers are expressed exactly using infinite representations, 
 comparing real numbers becomes a partial operation. Testing $x < y$ does not terminate when $x = y$ regardless of the specific representations that are used \cite[Theorem~4.1.16]{w00}.
To deal with such partiality, instead of making computations freeze when the same numbers are compared, \softwarename{AERN} and other exact real computation software provide a data type for lazy Booleans
$\{\cff, \ctt, \bot\}$, generalizing Booleans by adding a third element $\bot$ as an explicit state of divergence.
The third value $\bot$ is an admissible value that a variable can store whose nontermination is delayed until it is required to test whether the value is $\ctt$ or $\cff$.

Observing that the space of lazy Booleans satisfies the algebraic structure of Kleene's three-valued logic where $\bot$ is the third truth value for indeterminacy, the data type is also often called \emph{Kleenean}.
(See \cite{park2016foundation} for more of this discussion.)
The Kleeneans can be interpreted by the Kleenean assembly $\bK$  where the infinite sequence of zeros realizes $\bot$, any infinite sequence that starts with $1$ after a finite prefix of zeros realizes $\cff$, and any infinite sequence that starts with $2$ after a finite prefix of zeros realizes $\ctt$.

Following this approach, we introduce the type $\dK$ for Kleeneans as a primitive type in our formalization and axioms for its characterization. 
In particular, we assume that there are distinct constants $\ctrue, \cfalse : \kleene$, and if a Kleenean is \emph{defined}, meaning that $\defineds{k} \equivto k = \ctrue \lor k = \cfalse$, we assume that we can do branching on its value:
\[
\all{k:\dK}\defineds{k} \to \upc{k} + \downc{k}
\]
where $\upc{k} \equivto k = \ctrue$ and $\downc{k} \equivto k = \cfalse$.
We further assume that every Kleenean is observationally determined by whether it is true or false:
\[
(\upc{k_1}\leftrightarrow\upc{k_2})
\;\to\;
(\downc{k_1}\leftrightarrow\downc{k_2})
\;\to\;
k_1=k_2.
\]
Equivalently, two Kleeneans are equal whenever they have the same definedness behaviour and, if defined, the same Boolean value. Thus, although different realizers may represent the same Kleenean, the type \(\dK\) has exactly three extensional elements corresponding to \(\ctrue\), \(\cfalse\), and the undefined value.
We also introduce standard logical operators (negation, conjunction and disjunction) on Kleeneans following Kleene's three-valued logic. 

Since we often deal with infinite sequences of partial tests in this paper,
we suggest an extension to \cite{konevcny2022extracting} with countable disjunctions:
\begin{axiom}[Countable disjunction, \coqref{lazy\_bool\_countable\_or}{/formalization/Base/Kleene.v}{24}]\label{ax:disj}
\[
\all{f : \dN \to \dK} \some{k : \dK} k \neq \cfalse \land \left ( \upc{k} \liff \csome{n : \dN} \upc{f\ n} \right ).
\]
\end{axiom}
\noindent
To see that the axiom is valid, it suffices to show that such a $k$ is computable which can be done by iterating over all the outputs of the realizers of $f\ n$ in parallel, or, more precisely, by dovetailing the computations of all realizers, i.e., by interleaving their computation such that every realizer is simulated for arbitrarily many steps. Any witness produced by one of the realizers will therefore eventually be found.

Sierpinski space, the topological space $\{\definedx, \undefinedx\}$ whose only nontrivial open is $\{\definedx\}$ plays a central role in synthetic topology as it classifies opens, and in computability theory (e.g.\ \cite{escardo2004synthetic}), as it classifies semi-decidable sets, i.e., subsets for which there exists an algorithm that terminates whenever membership holds, but is not required to terminate otherwise.

We introduce Sierpinski space $\dS$ by restricting $\dK$:
\[
\dS \equivto \some{b :\dK} b \neq \cfalse.
\]
We name $\definedx \equivto (\ctrue, t_\definedx)$
and $\undefinedx \equivto (\cundef, t_\undefinedx)$
where $t_\definedx$ is the unique proof term for
$\ctrue \neq \cfalse$ and $t_\undefinedx$ is the unique proof term for $\bot \neq \cfalse$.
The uniqueness is due to the classical axioms in $\Prop$.
For a Sierpinski term $s : \dS$,
we write $\defineds{s}$ for $s$ being defined, i.e., $s = \definedx$ and $\undefineds{s}$ for $s$ being undefined, i.e., $s = \undefinedx$.

We further assume that the first projection $\fproj : \dS \to \dK$ admits a retraction $\skinv : \dK \to \dS$.
The assumption can be validated by a program on the Kleenean that on its input $b : \dK$ checks 
if $b = \ctrue$ or $\cfalse$ and simply diverges when $b = \cfalse$.

\subsection{Nondeterminism}
The nondeterminism considered in computable analysis is axiomatized 
as a monad $\mval : \Type \to \Type$ in our type theory.
The type $\mval X : \Type$ denotes the type of nondeterministic computations in $X : \Type$.
For example, $f : X \to \mval Y$ denotes a nondeterministic function from $X$ to $Y$.

Besides the monad axioms that make $\mval$ indeed a monad, 
we assume a \emph{nondeterministic dependent choice} principle.
Suppose a sequence of types  $P : \dN \to \Type$ and a sequence of classical binary predicates $R : \all{n:\dN}(P\ n) \to (P\ (n+1)) \to \Prop$. 
Suppose we have a mapping $f$ that for each $n  : \dN$, given $x : P\ n$, selects nondeterministically $y : P\ (n+1)$ such that $R_n\ x\ y$ holds:
\[
f : \all{n:\dN}\all{x:P\ n}
\mval \some{y : P\ (n+1)}R_n\ x\ y
\]
The nondeterministic dependent choice axiom states that there nondeterministically exists a trace $g : \all{n : \dN}P\ n$ such that $\lall{n:\dN}R_n\ (g\ n)\ (g\ (n+1))$ holds. 
In other words, we can compose the local non-deterministic choices into a global trace. 
The trace $g$ further satisfies a coherence condition, saying that each transition from $g\ n$ to $g\ (n+1)$ must be one of the possible values of $f\ n\ (g\ n)$. The details can be found in the original work \cite{konevcny2022extracting}. 
Under the realizability interpretation, the principle is realized by iteratively applying the realizer for $f$, which, given $n : \dN$ and a realizer for $x : P\ n$ produces a realizer for some $f\ n\ x$. 
Thus, starting with the realizer of some $x_0 : P\ 0$, applying the procedure $n$ times we get the $n$-th element of the trace.

The nondeterministic version of countable choice can be stated as
\[
\all{Y : \Type}\all{f : \dN \to \mval Y}\mval \some{g : \dN \to Y}
\lall{n:\dN} g\ n \in_\mval f\ n
\]
where $x \in_\mval y $, for $x : X$ and $y : \mval X$, 
denotes a classical proposition saying $x$ is a possible value in $y$.
Nondeterministic countable choice is derivable from nondeterministic dependent choice.

Generalizing the notion of continuous choice, for any type $X$, we define 
\[
X\mathrm{\_choice} \equivto
\all{Y: \Type} \all{f : X \to \mval Y}\mval \some{g : X \to Y}
\lall{x:X} g\ x \in_\mval f\ x.
\]
In the model, $X\mathrm{\_choice}$ is equivalent to saying that the interpretation of $X$ is a \emph{projective} object, meaning that it is \emph{isomorphic} to an assembly where each element has exactly one realizer. Intuitively, it means that when each element in $X$ has exactly one realizer, there will not be a true nondeterministic function out of $X$, considering that nondeterminism comes from computations being performed differently on different realizers of the same element.
Besides the natural number object $\bN$, another projective object in $\Casm$ is the Baire space assembly $(\IN^\IN, \{(\varphi, \varphi) \mid \varphi \in \IN^\IN\})$, where an infinite sequence $\varphi$ is realized only by itself, as is immediate from the realizability relation.
The interpretation of $\dN \to \dN$ is the exponential object $\bN^\bN$, which is isomorphic to this assembly; hence it is projective.
We can therefore assume the Baire choice for $X \equivto \dN \to \dN$:
\begin{axiom}[Baire choice, \coqref{M\_baire\_choice}{/formalization/Base/MultivalueMonad.v}{114}]\label{ax:baire-choice}
\[
\all{Y: \Type} \all{f : (\dN \to \dN) \to \mval Y}\mval \some{g : (\dN \to \dN) \to Y}
\lall{\varphi:\dN \to \dN} g\ \varphi \in_\mval f\ \varphi.
\]
\end{axiom}
\noindent
A common source of nondeterminism in practice in computable analysis is semi-decidability. Given multiple semi-decidable predicates, we choose one that terminates and evaluates to true, nondeterministically. An important example is the soft comparison (a variant from \cite{BRATTKA1998490}):
\[x <_n y \equivto \begin{cases}
\ctrue &\text{if } x \leq y - 2^{-n},\\
\cfalse &\text{if } y \leq x - 2^{-n},\\
\ctrue\text{ or } \cfalse&\text{otherwise when } |x - y| < 2^{-n},\\
\end{cases}\]
that nondeterministically approximates the partial $x < y$.
It can be realized by evaluating the two semi-decidable predicates $x< y+ 2^{-n}$ and $y < x + 2^{-n}$ in parallel provided that at least one of the two evaluations terminates and results in $\ctrue$. 
This is reflected in our type-theoretic framework by an axiom that states, given a countable sequence of Kleeneans $x : \dN \to \dK$, 
when we know that there (classically) is at least one index $n : \dN$ such that $x\ n = \ctrue$, we can nondeterministically find such an index:
\begin{axiom}[Nondeterministic countable selection, \coqref{ M\_countable\_selection}{/formalization/Base/MultivalueMonad.v}{100}]\label{ax:cchoice}
\begin{equation*}
\all{x : \dN \to \dK} \big( \csome{n : \dN}\upc{x\ n}\big) \to \mval \some{n : \dN} \upc{x\ n}.
\end{equation*}
\end{axiom}
\noindent 
Observe that this axiom, which can also be viewed as a version of Markov's principle for partially decidable predicates, is new compared to the original axiomatization in \cite{konevcny2022extracting}, which contains only binary nondeterminism.
This axiom can be realized by iterating over all pairs of natural numbers $(n,m)$ and checking if the $m$-th sequence element of the realizer of $x(n)$ equals two, which indicates that $x(n) = \ctrue$.

Another primitive way to construct nondeterminism is from a single Kleenean $k$ by testing whether $k = \ctrue$. We can consider a procedure that for each natural number $n$ evaluates $k$ using $n$ steps of computation and returns $1$ as an indicator that $k = \ctrue$ if within the $n$ steps $k$ evaluates to $\ctrue$. Otherwise, if $n$ steps were not sufficient to make $k$ evaluate to anything or $k$ happens to be $\cfalse$, it returns any different number, e.g., $0$.
This procedure of constructing a sequence of natural numbers $\dN \to \dN$ from $k$ is nondeterministic because the precise $n$ needed to judge $k = \ctrue$ is dependent on the specific realizer of $k$. That is, given two realizers of $\ctrue$, the sequences produced by the above procedure can be different. Hence, what this procedure realizes is the following axiom:
\begin{axiom}[\coqref{kleenean\_to\_nat\_sequence}{/formalization/Base/MultivalueMonad.v}{110}]\label{ax:k-to-seq}
\begin{equation*}
\all{k : \dK}  \mval\some{f : \dN \to \dN}
\upc{k} \liff \csome{n : \dN}f\ n = 1.
\end{equation*}
\end{axiom}

\subsection{Real Numbers and Euclidean Space}
We consider the assembly of reals $\bR$ with the standard Cauchy representation \cite[Section~4.1]{w00}
\[
\varphi \Vdash_\bR x \Leftrightarrow
\varphi\ \text{encodes}\ (q_n)_{n\in\IN}\subset \IQ\ \text{such that}\  |x - q_n| \leq 2^{-n}\text{ for all }n
\]
as the interpretation of our type $\dR$ of real numbers. We further assume that $\dR$ has two distinct constants $0$ and $1$, 
that there is a $\Prop$-valued classical order relation $< : \dR \to \dR \to \Prop$, that $<$ is semi-decidable 
(i.e. $\all{x,y:\dR}\some{k : \dK} \upc{k} \leftrightarrow x < y$), and the standard arithmetical operators. To formalize the completeness of reals, we define a sequence $f: \dN \to \dR$ being \inlinedef{fast Cauchy} by
\[
\mathrm{is\_cauchy}\ f \equivto
\lall{n,m : \dN}\abs{f\ n - f\ m} \leq 2^{-n}+2^{-m},
\]
and $x : \dR$ being a limit of a sequence $f$ by 
\[
\mathrm{is\_limit}\ x\ f \equivto
\lall{n : \dN}  \abs{x - f\ n} \leq 2^{-n}
\]
where $|x - y| \leq z$ is a shorthand for $-z \leq x - y \leq z$.
The completeness axiom states that for any fast Cauchy sequence, there exists its limit point:
\[
  \all{f: \dN \to \dR} \mathrm{is\_cauchy}\,f \to \some{x : \dR} \mathrm{is\_limit}\,x\,f\ .
\]

Using the axiomatization, we can define several standard operations such as $\max(x,y)$, $\abs{x}$, etc. Furthermore, as suggested in \cite{DBLP:journals/mlq/Hertling99}, the set of axioms is enough to prove that any two real number types sharing the set of axioms are equivalent. Again, the details can be found at \cite{konevcny2022extracting}.

Having axiomatized real numbers, extending the theory to complex numbers and general Euclidean space is simple.
We define a complex number as a pair of real numbers, i.e.\ $\mathrm{C} \equivto \dR \times \dR$, and for any $m \colon \dN$ we define the type $\dR^m$ simply as the type of length-$m$ lists over $\dR$, with maximum norm $\lvert \cdot \rvert : \dR^m \to \dR$ and a distance $d\ z_1\ z_2 \equivto \lvert z_1 - z_2 \rvert$.
Vector space operations are defined pointwise and the limit operators can be extended in the obvious way.

\section{Continuity}\label{sec:continuity}
The continuity principle on real numbers  states that every real function is continuous.
Formally, it can be written as
\[
\all{f : \dR \to \dR}\all{x : \dR}\all{n:\dN}
\mval\some{m:\dN}\lall{y :\dR}|x - y|  \leq  2^{-m} \to |f(x) - f(y)| \leq 2^{-n}.
\]

Under the realizability interpretation, this corresponds to an operation which, given a realizer $\varphi_f$ of $f$, a realizer $\varphi_x$ of $x$, and a natural number $n$, \emph{nondeterministically} computes a natural number $m$ such that $\lall{y : \dR} |x - y| \leq 2^{-m}\to |f(x) - f(y)| \leq 2^{-n}$. The computability of this operation, as well as the impossibility of removing $\mval$ from the statement has been well-studied in constructive and computable analysis \cite{xuthesis}.
The necessity of the presence of $\mval$ is due to the representation of functions:
For a fixed function $f$, different realizers may yield different valid moduli $m$. 

The continuity principle is counter-classical, and at first glance it might appear restrictive to work in a system where discontinuous functions cannot be discussed.
However, in the intended model of our formalization, the continuity principle is valid: a term $f: X \to Y$ is realized by a computable function between represented spaces, and such functions are necessarily continuous.
Thus, exploiting the computational content of this fact can be useful.

At  the same time, it is easy to see that the continuity principle is independent of the other axioms in our system, as all previous axioms could still be validated in classical mathematics, i.e.\ the category of sets, while the continuity principle is counter-classical.
Consequently, the continuity principle cannot be derived from the other axioms alone, and to introduce it in our system it must be assumed as an additional axiom.

Let us next discuss the concrete formulation of the axiom.
Rather than restricting to real functions, we want a formulation that is as general and abstract as possible. 
A common approach in constructive mathematics is to formulate continuity for functions on Baire space, $(\IN \to \IN) \to \IN$, as in the Kreisel–Lacombe–Shoenfield theorem~\cite{KreiselLacombeShoenfield1959}, and then encode more general spaces into Baire space.
However, this is precisely the kind of encodings that we try to avoid in our approach whenever possible, by axiomatizing properties of abstract spaces directly.
At the same time, we do not want to formulate a separate continuity principle for each combination of abstract types occurring in the development.
We therefore propose the following continuity principle for sequences, which turns out to be general enough to cover the applications in this paper:
\begin{axiom}[Continuity, \coqref{seq\_subset\_continuity}{/formalization/Base/MultivalueMonad.v}{105}]\label{ax:continuity}
We assume a term of the following type:
\begin{multline*}
\all{f : (\dN \to X) \to \dS} \all{x : \dN \to X} 
\defineds{(f\ x)}\\  \qquad\qquad\rightarrow \mval \some{m : \dN} 
\lall{y : \dN \to X}  \restr{x}{m} = \restr{y}{m} \rightarrow \defineds{(f\ y)}.
\end{multline*}
\end{axiom}
\noindent
Here, for $x : \dN \to X$ and $n : \dN$, the term $\restr{x}{n} : \listof{X}$ denotes the prefix of $x$ of length $n$, where $\listof{X}$ is the type of lists over $X$.
Essentially, the axiom states that any semi-decidable subset of sequences is continuous, i.e.\ if a sequence $x: \dN \to X$ belongs to a subset defined by its characteristic function $f: (\dN \to X) \to \dS$, then there nondeterministically is a finite prefix of $x$ such that every sequence sharing this prefix also belongs to the subset.

In the remainder of this section, we demonstrate how this formulation of the continuity principle gives rise to continuity principles on various other types. 
One immediate application is to derive a continuity principle for functions into any type $Y$ with semi-decidable equality:
\begin{lem}[\coqref{continuity\_semidec\_eq}{/formalization/Hyperspace/Continuity.v}{239}]\label{lem:semi-cont}
Let $X, Y: \Type$ and assume that equality on $Y$ is semi-decidable, i.e.,
\[ 
\all{y, y' : Y} \some{s : \dS} \defineds{s} \liff y = y'.\]
Then the following continuity principle holds: 
\[
\begin{aligned}
&\all{f : (\dN \to X) \to Y}\all{x : \dN \to X} \mval \some{m : \dN}\\
&\qquad \lall{x' : \dN \to X} \restr{x}{m} = \restr{x'}{m} \rightarrow f\ x=f\ x'.
\end{aligned}
\]
\end{lem}
\begin{proof}
The idea is to reduce $f$ to a family of Sierpinski-valued functions.
That is, for each $y : Y$, we define a function $f_y : (\dN \to X) \to \dS$ such that $\lall{x : \dN \to X} \defineds{(f_y\ x)} \liff (f\ x = y)$, which is possible since equality on $Y$ is semi-decidable.

Then trivially $\defineds{(f_{(f\ x)}\ x)}$ and thus applying \autoref{ax:continuity} on $f_{(f\ x)}$ and $x$ yields
\[
\mval \some{m : \dN} \lall{y : \dN \to X} 
\restr{x}{m} = \restr{y}{m} \to \defineds{(f_{(f\ x)}\ y)}.
\]
As $\defineds{(f_{(f\ x)}\ y)} \liff (f\ y = f\ x)$, the claim follows.
\end{proof}
\noindent 
When considering represented spaces, the spaces where equality is semi-decidable are precisely the discrete spaces.
Thus, we can apply this lemma to get the standard continuity principles for discrete types such as $Y = \dN$ or $Y = \dB$. 
However, computable analysis is mostly concerned with non-discrete spaces.
To treat such cases, we extend the previous lemma to functions where the codomain is a sequence type as well, i.e., replacing $Y$ by $\dN \to Y$.
\begin{lem}[\coqref{baire\_continuity}{/formalization/Hyperspace/Continuity.v}{15}]\label{lem:baire-continuity}
For any types  $X, Y$ where equality for $Y$ is semi-decidable, the following continuity principle holds: 
\begin{multline*}
\all{f : (\dN \to X) \to (\dN \to Y)} \all{x : \dN \to X}\all{n : \dN}\\ 
\qquad\qquad\mval \some{m : \dN} \lall{y : \dN \to X} \restr{x}{m} = \restr{y}{m} \to  \restr{(f\ x)}{n} = \restr{(f\ y)}{n}.
\end{multline*}
\end{lem}
\begin{proof}
The proof is by induction on $n$.
For $n=0$, we can choose  $m=0$ and the statement is trivially true.
Next, assume we can (nondeterministically) choose some $m_0$ such that 
\[
\lall{y : \dN \to X} \restr{x}{m_0} = \restr{y}{m_0} \to \restr{(f\ x)}{n} =\restr{(f\ y)}{n}.
\]
We need to show how to extend this to prefixes of length $n+1$.
It suffices to show that there is some $m_1$ such that 
\begin{equation}\label{eqn:baire-cont1}
\lall{y : \dN \to X} \restr{x}{m_1} = \restr{y}{m_1} \to (f\ x\ n = f\ y\ n)
\end{equation}
and choose $m = \max(m_0,m_1)$.

To show \autoref{eqn:baire-cont1}, we define a Sierpinski-valued function 
$f' : \dN \to Y  \to (\dN \to X) \to \dS$ such that 
\[
    \lall{k : \dN}\lall{f_k :  Y}\lall{x : \dN \to X} \defineds{(f'\ k\ f_k\ x)} \liff (f\ x\ k) = f_k
\]
which is possible since equality on $Y$ is semi-decidable.
Applying \autoref{ax:continuity} to $f'\ n\ (f\ x\ n)$ and $x$, we get
\[
\mval \some{m_1 : \dN} \lall{y : \dN \to X} \restr{x}{m_1} = \restr{y}{m_1} \to \defineds{(f'\ n\ (f\ x\ n)\ y)}.
\]
Since $\defineds{(f'\ n \ (f\ x\ n)\ y)} \liff f\ y\ n = f\ x\ n$ this completes the proof.
\end{proof}
\noindent
Note that the formal development only proves the special case of Baire space, i.e.\ $f : (\dN \to X) \to (\dN \to \dN)$, as it is the only instance required in the applications. As equality on $\dN$ is decidable and not just semi-decidable, the formal proof is slightly simpler than the one stated here, but follows the same basic idea.

\subsection{Continuity of Continuity}
Given a semi-decidable subset of the Baire space $p : (\dN \to \dN) \to \dS$, applying the continuity principle to $p$ yields a partial nondeterministic mapping
\[
\mu_p : (\dN \to \dN) \pto \mval \dN\ .
\]
Here, $X \pto Y$ denotes the type of partial mappings from $X$ to $Y$ defined by
\[X \pto Y \equivto \some{D:X\to\Prop} \all{x:X} D\ x \to Y.\]
We call the first entry $D  : X \to \Prop$ of a partial mapping $f : X \pto Y$ the \emph{domain of $f$}, denoted by  $\dom(f) : X \to \Prop$. The domain of $\mu_p$ corresponds to the semi-decidable subset $p$, namely
\[\dom(\mu_p) = \lam{x:\dN \to \dN}\defineds{(p\ x)}.\]
For a partial function $f : X \pto Y$, $x : X$, and $y : Y$, we further write
$\definedto{f\ x}{y}$ to mean that $f\ x$ is defined with value $y$, formally
\[
\definedto{f\ x}{y} \quad \equivto \quad \csome{p : (\dom(f)\ x)}(f\ x\ p = y).
\]

The partial mapping $\mu_p$ is a nondeterministic modulus of continuity for $p$:
for every $x : \dN \to \dN$ and $t : \defineds{(p\ x)}$, we have
\[\lall{n : \dN} n \in_\mval (\mu_p\ x\ t) \to
\lall{y : \dN \to \dN} \restr{x}{n} = \restr{y}{n} \to \defineds{(p\ y)}.\]
In other words, $\mu_p$ nondeterministically computes the length $n$ of a prefix of $x$ such that any extension of this prefix also satisfies $p$.
In the standard topology on Baire space, such a prefix determines the basic open set of all sequences with the same first $n$ values. Thus, the property says that this entire basic open set is contained in the set defined by $p$.

Later in the paper, we use $\mu_p$ to construct covers of $p$ by basic opens with centers chosen from a countable dense subset of Baire space. 
For this construction, it is essential that nearby centers must be assigned compatible neighborhoods. In other words, we require some form of continuity of the partial map $\mu_p$ itself.

However, as $\mu_p$ is multivalued, the continuity principle does not apply directly.
Still, if $\mu_p$ were a total mapping, we could use the Baire choice axiom (\autoref{ax:baire-choice}) to nondeterministically obtain a deterministic section:
\[
\mval\some{\mu_p':(\dN \to \dN)\to\dN}
\lall{x : \dN \to \dN} \mu_p'\ x \in_\mval \mu_p\ x
\]
and apply the continuity principle (\autoref{lem:semi-cont}), to obtain that each section is continuous:
\[
\mval\some{\mu_p':(\dN \to \dN)\to\dN}
(\mathrm{is\_cont}\ \mu_p') \land
\lall{x : \dN \to \dN} \mu_p'\ x \in_\mval \mu_p\ x.
\]
where $\mathrm{is\_cont}\ \mu_p'$ stands for
\[\all{x : \dN \to \dN} \mval \some{m : \dN} \lall{x' : \dN \to \dN} \restr{x}{m} = \restr{x'}{m} \rightarrow \mu_p'\ x=\mu_p'\ x'.\]
Thus, from a total nondeterministic $\mu_p : (\dN \to \dN)\to\mval\dN$, we can nondeterministically choose a continuous section $\mu_p' : (\dN \to \dN)\to\dN$ which can be used consistently in subsequent computations.
It should be noted that the continuous dependence here is only on $x$. Continuous dependence on $p$ would be invalid in the model.

It turns out that analogous results still hold when $\mu_p$ is partial, provided that the domain of $\mu_p$ is semi-decidable. Namely, we prove \autoref{lem:partial-baire-choice} extending \autoref{ax:baire-choice} and \autoref{lem:partial-continuity} extending \autoref{lem:semi-cont}. As their direct corollary, we obtain \autoref{lem:continuous-modulus} stating that even in the general case, when $\mu_p$ is partial on a semi-decidable subset, we can nondeterministically choose a nondeterministic section of it and this section is a continuous function.

\begin{lem}[\coqref{partial\_baire\_choice}{/formalization/Hyperspace/Continuity.v}{279}]\label{lem:partial-baire-choice}
Let $p : (\dN \to \dN) \to \dS$ define a semi-decidable subset of Baire space and  $Q : (\dN \to \dN) \to \dN \to \Prop$ a relation.
Assume further we have a nondeterministic partial function whose outputs satisfy $Q$, formally a term of type
\[ \left( \all{x : \dN \to \dN} \defineds{(p\ x)} \rightarrow \mval \some{n : \dN} (Q\ x\ n)\right). \]
Then, we can nondeterministically choose a partial function $F : (\dN \to \dN) \pto \dN$ whose domain is exactly the subset defined by $p$, and whose values satisfy $Q$.
Formally, we construct a term of type
\[
\mval \some{F : (\dN \to \dN) \pto \dN}  \dom(F) = \lam{x : \dN \to \dN} \defineds{(p\ x)}
\land \ \lall{x \in \dom(F)} (Q\ x\ (F\ x)).
\]
Here, $\lall{x\in\dom(f)}Q(x, f\ x)$, is an abbreviation for $\lall{x: X}\lall{p : \dom(f)\ x}Q(x, f\ x\ p)$.
\end{lem}
\begin{proof}
By replacing $(p\ x)$ with an infinite sequence of integers (\autoref{ax:k-to-seq}) we get
\[
\all{x : \dN \to \dN} \mval \some{f : \dN \to \dN} \defineds{(p\ x)} \liff \csome{n : \dN} (f\ n) = 1
\]
Thus, for any $x : \dN \to \dN$ and $n : \dN$, we can check if $f\ n = 1$ and if yes find an $m : \dN$ with $Q\ x\ m$.
To apply \autoref{ax:baire-choice}, we show 
\begin{multline*}
\all{x : \dN \to \dN}  \mval \some{g: \dN \to \dN}  \\
 (g\ n \neq 0 \rightarrow Q\ x\ ((g\ n)-1)) \  \land \  (\defineds{(p\ x)} \rightarrow \csome{n : \dN} g\ n \neq 0)
\end{multline*}
from which we get a function $h : (\dN \to \dN) \to \dN \to \dN$.
Now we can define the partial function $F$ at $x : \dN \to \dN$ with $\defineds{(p\ x)}$ by searching for the first index $n : \dN$ such that $h\ x\ n \neq 0$.
\end{proof}
\begin{lem}[\coqref{continuity\_partial}{/formalization/Hyperspace/Continuity.v}{316}]\label{lem:partial-continuity}
Any partial mapping $f : (\dN \to \dN) \pto \dN$ with semi-decidable domain is continuous. That is,
\[
 \all{x \in \dom(f)} \mval \some{m : \dN} \lall{y : \dN \to \dN} \restr{x}{m} = \restr{y}{m} \to \definedto{f\ y}{f\ x}.
\]
\end{lem}

\begin{proof}
Similar to the proof of the previous lemma, let $p: (\dN \to \dN) \to \dS$ define the domain of $f$ and define a total function 
$g : (\dN \to \dN) \to \dN \to \dN$ such that 
\[
(\all{n : \dN} g\ x\ n \neq 0 \rightarrow f\ x = (g\ x\ n) - 1) \  \land \defineds{(p\ x)} \rightarrow \csome{n : \dN} g\ x\ n \neq 0.
\]
Then by continuity of Baire-valued functions (\autoref{lem:baire-continuity}) for each $x: \dN \to \dN$ and $n : \dN$, we have
\begin{equation}\label{eq:cont-ext-fun}
\mval\some{m : \dN} \all{y : \dN \to \dN} \restr{x}{m} = \restr{y}{m} \rightarrow \restr{(g\ x)}{n} = \restr{(g\ y)}
{n}.
\end{equation}
Now assume $\defineds{(p\ x)}$, then $\mval \some{n : \dN} g\ x\ n = (f\ x)+1$, thus we can use \autoref{eq:cont-ext-fun} to find an $m: \dN$ such that
\[
\lall{y: \dN \to \dN} \restr{x}{m} = \restr{y}{m} \rightarrow \restr{(g\ x)}{n+1} = \restr{(g\ y)}{n+1} 
\]
and since $g\ x\ n \neq 0$, for any $y: \dN \to \dN$, $g\ y\ n = g\ x\ n \rightarrow f\ x = f\ y$.
\end{proof}

\begin{thm}[\coqref{baire\_continuous\_modulus}{/formalization/Hyperspace/Continuity.v}{417}, \coqref{continuity\_partial}{/formalization/Hyperspace/Continuity.v}{316}]\label{lem:continuous-modulus}
Let $p : (\dN \to \dN) \to \dS$, then there nondeterministically is a partial function $\mu : (\dN \to \dN) \pto \dN$ with $\dom(\mu) = \lam{x : \dN \to \dN} \defineds{(p\ x)}$ such that
\begin{enumerate}
\item $\mu$ is a modulus of continuity for $p$, i.e., \[\lall{x : \dN \to \dN} \defineds{(p\ x)} \rightarrow \lall{y : \dN \to \dN} \restr{x}{\mu\ x} = \restr{y}{\mu\ x} \rightarrow \defineds{(p\ y)},\] 
\item and $\mu$ is continuous in the sense that 
\[
\all{x \in \dom(\mu)}\mval \some{k : \dN} 
\lall{y : \dN \to \dN} 
\restr{x}{k} = \restr{y}{k} \rightarrow
 \definedto{\mu\ y}{\mu\ x}.
\]  
\end{enumerate}
\end{thm}
\begin{proof}
By applying \autoref{ax:continuity} to $p$, we get
\[
\all{x : \dN \to \dN} \defineds{(p\ x)} \to \mval \some{m : \dN}
\all{y : \dN \to \dN} \restr{x}{m} = \restr{y}{m}\to
\defineds{(p\ y)}.
\]
Applying further \autoref{lem:partial-baire-choice} to this term, we nondeterministically get a partial modulus function with the required domain:
\[
\mval \some{\mu : (\dN \to \dN) \pto \dN}
\begin{aligned}[t]
\bigl(&\dom(\mu)
=
\{x : \dN \to \dN \mid \defineds{(p\ x)}\}
\\
&{}\land
\lall{x \in \dom(\mu)}
\all{y : \dN \to \dN}\,
\restr{x}{\mu\,x}
=\restr{y}{\mu\,x}
\to
\defineds{(p\ y)}
\bigr).
\end{aligned}
\]
By \autoref{lem:partial-continuity},  any such $\mu$ is continuous.
\end{proof}

\subsection{Continuity on Euclidean Spaces}
In the preliminary version \cite{DBLP:conf/mfcs/Konecny0T23} of this extended paper we only consider the special case of the Euclidean space.
To this end, we formalize continuity only for functions $f : \dR \to \dS$.
As the formalization is quite different from what is presented in this paper, let us briefly restate the definitions and show how they relate to the new principles.

The idea of the formalization in \cite{DBLP:conf/mfcs/Konecny0T23} is based on the fact that efficient implementations of exact real computation often implement real numbers as infinite sequences of rational intervals.
Typically, the whole computation is done by replacing real numbers by fixed precision intervals.
If the precision turns out to be insufficient at some point, the computation is restarted with intervals of higher precision.
To gain some intuition for the following definition, let us consider an implementation of $f: \dR \to \dS$ and the computation of $f\,x$ for some concrete $x : \dR$.
The program starts with an interval containing $x$ and tries to evaluate the test on this interval. If the test succeeds, it returns $\definedx$, otherwise the computation is repeated with a smaller interval. Thus, whenever $\defineds{(f\,x)}$, this process must eventually terminate after finitely many refinements. Equivalently, there is some rational interval containing $x$ on which the internal implementation can already certify that $\defineds{(f\,x)}$ holds. 

To formalize this intuition, define a type $\dD$ of dyadic rational numbers, for example as pairs of integers, and let $(d, n) : \dD \times \dN$ represent the interval $(d - 2^{-n}, d + 2^{-n}) \subset \IR$.
We can then view the test performed during refinement as a Boolean-valued function $\hat{f} : \dD \times \dN \to \dB$, where $\ctrue$ means success, while \(\cfalse\) means that the current interval is not sufficient (i.e.\ undecided).
That is, for a map $f : \dR \to \dS$, we say that a function $\hat f : \dD \times \dN \to \dB$ is an \inlinedef{interval extension} of $f$ whenever it satisfies the following inclusion property:
\[
\mathrm{incl}\ f\ \hat{f} \equivto \lall{d : \dD} \lall{n : \dN} \hat{f} \ (d, n) = \ctrue \rightarrow \left ( \lall{x : \dR} \abs{x-d} \leq 2^{-n} \rightarrow \defineds{(f\  x)} \right ).  
\]
In other words, whenever the interval extension $\hat{f}$ of $f$ returns $\ctrue$ on the pair $(d,n)$, $(f\;x)$ is defined for all $x \in (d-2^{-n}, d+2^{-n})$.
The term interval extension comes from the fact that we can regard the boolean values as encodings for the two nonempty open subsets (intervals) of Sierpinski space, i.e.\ $\ctrue : \dB$ is an encoding for  $\{\definedx\}$ and 
$\cfalse :\dB$ is an encoding for $\{\definedx, \undefinedx\}$. Thus $\hat{f}$ extends the function $f$ to a function that maps nonempty real (dyadic) intervals to nonempty intervals in Sierpinski space.

However, not all interval extensions are interesting as the inclusion property holds also for the trivial map $\hat{f}$ that returns $\cfalse$ for all inputs.
We say that an interval extension $\hat{f} : \dD \times \dN \to \dB$ of a function $f : \dR \to \dS$ is \inlinedef{tight} if whenever $\defineds{(f\ x)}$, there classically exists an interval containing $x$ for which $\hat{f}$ indicates that $\definedx$ is the only valid answer on that interval:
\[
\mathrm{tight}\ f\ \hat{f} \equivto \lall{x : \dR} \defineds{(f\ x)} \rightarrow \csome{d : \dD} \csome{n : \dN} \abs{x-d} \leq 2^{-n} \land \hat{f}\ (d, n) = \ctrue
\]

We then state continuity by saying that
for any map $f : \dR \to \dS$, there nondeterministically exists an interval extension of it that is tight, i.e.
\begin{equation}\label{eq:old-continutiy}
\all{f: \dR \to \dS} \mval \some{\hat{f} : \dD \times \dN \to \dB} 
\mathrm{incl}\ f\ \hat{f} \land \mathrm{tight}\ f\ \hat{f}.
\end{equation}
Such a formulation is rather natural in our setting as assuming that the underlying exact real computation is done via expressing real numbers by converging sequences of intervals with dyadic rational endpoints, \autoref{eq:old-continutiy} can be realized simply by the trivial operation extracting the realizer $\hat f$ from $f$.
Note that since a single function admits several different realizers, the procedure of extracting realizers is nondeterministic.

The continuity defined in \autoref{ax:continuity} by itself is weaker than the one defined in \autoref{eq:old-continutiy}.
The tightness property in \autoref{eq:old-continutiy} implies that for each $x : \dR$ we can find a dyadic number $d$ and an integer $n$ such that $f$ is defined on the whole ball with radius $2^{-n}$ centered at $d$, and the ball contains $x$.
Thus, for a function $f: \dR \to \dS$, we can cover all real numbers with $\defineds{(f\ x)}$ simply by testing $\hat{f}$ on dyadic balls.
On the other hand, for pointwise continuity as defined in \autoref{ax:continuity}, the modulus at some $x : \dR$ and dyadic rationals close to $x$ can a priori be completely unrelated, and thus covering with balls at dyadic centers might miss some reals.

However, using \autoref{lem:continuous-modulus}, i.e., the fact that for a function on Baire space the modulus function is also continuous (which additionally required the Baire choice axiom), it is possible to derive \autoref{eq:old-continutiy}.
As the old continuity principles become obsolete in the new formalization, we do not provide a formal proof of this derivation in the Coq library, but let us nonetheless briefly sketch the idea here informally.
First, we can replace a real number $x : \dR$ by an approximation function $\varphi : \dN \to \dN$, encoding a dyadic sequence rapidly converging to $x$.
Thus, we can replace the function $f: \dR \to \dS$ with a function $f': (\dN \to \dN) \to \dS$.
Let $d : \dD$ and $\varphi_d$ encode the trivial approximation of $d$ by a constant sequence.
Assume $f'$ is defined on $\varphi_d$, then by continuity, choosing a large enough $m : \dN$ we can make sure that $\defineds{(f'\ \psi)}$ for all $\psi$ coinciding with $\varphi_d$ on the first $m$ elements. 
Further, since the modulus function $\mu_{f'}: (\dN \to \dN) \pto \dN$ of $f'$ is continuous at $\varphi_d$ it follows that again choosing $m : \dN$ large enough,   
for any sequence $\psi$ coinciding with $\varphi_d$ on the first $m$ elements, $\mu_{f'}\ \varphi = \mu_{f'}\ \psi$.
Thus, we define $(\hat{f}\; d\; n)$ to be true iff $n$ is large enough to make both the above statements hold.

Recalling that $f'$ simulates $f$ on rapid sequences of dyadics converging to a real number, it follows that $(f\;x)$ is defined whenever $\abs{x-d} \leq 2^{-n}$, i.e. the inclusion property of $\hat{f}$.
Now assume $(f\;x)$ is defined for some $x : \dR$, then $(f'\;\varphi)$ is defined for all approximations $\varphi_{x_n}$ where the first $n$ entries coincide with a constant dyadic $2^{-n}$ approximation $d_n : \dD$ of $x$.
Now by continuity of $\mu_{f'}$ at some point $(\mu_{f'}\ \varphi_{x_n}) = (\mu_{f'}\ \varphi_{d_n}) < n$ and therefore, assuming large enough $n$, $(\hat{f}\ d_n\ n) = \ctrue$ and $\abs{x-d_n} \leq 2^{-n}$.
Thus, the tightness property also holds.
\section{Subsets}\label{sec:subsets}
Giving a classical description of an object is much simpler than showing that it exists constructively.
For example, the square root of a nonnegative real number $x : \dR$ can simply be described as a nonnegative number $y$ such that $y \cdot y = x$ without concern for which algorithm to choose.
Our formalization allows us to conveniently separate such classical definitions from their computational implementation.

Similarly, we can define classical subsets pointwise as the set of points satisfying a classical predicate.
That is, for any type $X$, we express subsets of $X$ classically as
\[
\Csubset\ X \equivto X \to \Prop.
\]
For example, the set $\{x : \dR \mid x^3 \leq 2\}$ can simply be defined by $\lam{x : \dR} x^3 \leq 2$.

We define the standard set operations and relations in the obvious way.
That is, for $S, T : \Csubset\ X$, we define 
\begin{align*}
    \emptyset &\equivto \lam{x:X}\tfalse,\\
    x \in S &\equivto S\ x,\\
    S^C &\equivto \lam{x : X} \neg (S\ x),\\
    S \cup T &\equivto \lam{x : X}x \in S \lor x \in T,\\
    S \cap T &\equivto \lam{x : X}x \in S \land x \in T,\text{ and}\\
    S \subseteq T &\equivto \lall{x : X} S\ x \rightarrow T\ x.
\end{align*}
Moreover, for indexed subsets $f : I \to \Csubset\ X$ for some type $I$,
we define 
\[
\bigcup\ f \equivto \lam{x : X}\csome{i : I}x \in f\ i\quad\text{and}\quad
\bigcap\ f \equivto \lam{x : X}\lall{i : I}x \in f\ i.
\]
We often write $\bigcup_{x :X} f$ when $f$ is a term with $x$ appearing free to denote $\bigcup \lam{x:X} f$ and $\bigcap_{x : X} f$ to denote $\bigcap \lam{x:X}f$.

When we later define a type for structured subsets e.g. 
\[
P\!\textrm{\_subset}\ X\equivto \some{S : \Csubset\ X}P\ S
\]
for some $P : \Csubset\ X \to \Type$, if it is not explicitly defined otherwise, the set operations and relations on $P\!\textrm{\_subset}\ X$ are meant to act on the first projection; e.g., \[\left((S : P\!\textrm{\_subset}\ X) \subseteq (T : \Csubset\ X) \right)\equivto 
\fproj S \subseteq T.\]
We further use shorthand notations such as $S :\subseteq X$ for $S : \Csubset\ X$, $\all{x \in S}P\ x$ for $\all{x : X} S\ x \rightarrow P\ x$, $\some{x \in S}P\ x$ for $\some{x : X} S\ x \land P\ x$, and so on. The notation $S :\subseteq X$ of type judgement should not be confused with $(S \subseteq X) : \Prop$.

Our goal is to assign computational content to classically defined subsets which in turn allows us to extract programs that perform operations on these sets. 
In constructive mathematics, topological notions such as open, closed, compact, and so on, are closely linked to the kind of information that can be computed about the set.
For instance, an open set in constructive mathematics is not just a collection of points but is typically defined by a rule or algorithm that can semi-decide whether any given point belongs to the set.
Similarly, in computable analysis, computational content is often provided by defining representations for specific spaces of subsets, such as spaces of open or closed sets and performing computations on these representations (see e.g.\ \cite{Pauly2016OnTT,collins2020} for an overview).

We follow a similar approach by enriching classical subsets with computational content using topological notions such as open, compact and overt. 
Note that although we use the same notions as in point-set topology, one should not think about them as properties of the point set, but as the computational information additionally attached to the set.

The definitions are possible without restricting to specific spaces.
However, the proofs in the general setting often rely on unbounded search or other procedures with unbounded complexity and thus do not encode useful programs.
Later, we therefore consider the case of metric spaces, and in particular subsets of Euclidean space, more specifically and specialize some of the theorems to get more efficient algorithms.

\subsection{Open and Closed Sets}
Recall that Sierpinski space is the topological space on the two-point set $\{\definedx, \undefinedx\}$ with open subsets $\emptyset$, $\{\definedx\}$, and $\{\definedx, \undefinedx\}$. 
Thus, for a subset $T \subseteq X$ of a topological space $X$ we can consider its characteristic function $\chi_{T} : X \to \dS$, defined by
\[
\chi_T(x) =
\begin{cases}
\definedx & \text{if } x \in T,\\
\undefinedx & \text{otherwise.}
\end{cases}
\]
As $T$ is open iff $\chi_{T}$ is continuous, we can use its characteristic function as a witness of the set being open.
This leads to the following definition.
\begin{defi}[\coqref{open}{/formalization/Hyperspace/Subsets.v}{13},\coqref{closed}{/formalization/Hyperspace/Subsets.v}{134}]\label{def:open-closed}
We identify \emph{open sets} with their characteristic function:
\begin{align*}
\isopen\  (U :\subseteq X) & \equivto \some{f : X \to \dS} \lall{x : X} \defineds{(f\ x)} \,\liff \, x \in U, \\
\sopen\ (X : \Type) &\equivto \some{U :\subseteq X}\isopen\ U\ .
\end{align*}
Similarly, we define \emph{closed sets} as their complement:
\begin{align*}
\isclosed\  (A :\subseteq X) &\equivto \some{f : X \to \dS} \lall{x : X} \defineds{(f\ x)} \,\liff\, x \not \in A, \\
\sclosed\ (X:\Type) &\equivto \some{A :\subseteq X }\isclosed\ A\ .
\end{align*}
\end{defi}
\noindent
By considering the function $f: X \to \dS$ to be computable, we automatically  
get the computational information that is assigned to open and closed subsets.
Namely, a subset $T$ of $X$ is open iff $x \in T$ is semi-decidable and closed if $x \not \in T$ is semi-decidable for $x : X$.
The definition allows simple, straightforward proofs of many basic properties of open and closed subsets. We list a few examples without proof.
A more exhaustive treatment can be found in \cite{Pauly2016OnTT}.
\begin{lem}\label{lem:open-facts}
The following properties hold for open and closed subsets (\coqref{open\_countable\_union}{/formalization/Hyperspace/Subsets.v}{59}, \coqref{open\_intersection}{/formalization/Hyperspace/Subsets.v}{99}, \coqref{closed\_union}{/formalization/Hyperspace/Subsets.v}{136}, \coqref{closed\_countable\_intersection}{/formalization/Hyperspace/Subsets.v}{151}).
\begin{enumerate}
    \item Given any countable open subsets $U : \dN \to \sopen\ X$, the countable union $\bigcup  U$ is also open.
    \item If $U_1$ and $U_2$ are open then so is their intersection $U_1 \cap U_2$.
    \item If $A$ and $B$ are closed then so is their union $A \cup B$.
    \item Given any countable closed subsets $A : \dN \to \sclosed\ X$,  the countable intersection $\bigcap A$ is again closed.
\end{enumerate}
\end{lem}
\subsection{Compact and Overt Subsets}
Compactness is one of the most fundamental concepts in topology and analysis.
Compactness in classical topology is typically defined by the property that every open cover has a finite subcover. 
Defining compactness constructively is more subtle (see e.g.\ \cite{diener2008compactness}).
A generic definition for compactness on represented spaces requires that we can effectively check if an open subset is a cover.
Topologically, this corresponds to the upper Vietoris topology \cite{vietoris1922bereiche}.
That is, compactness is defined as follows: 
\begin{defi}[\coqref{compact}{/formalization/Hyperspace/Subsets.v}{201}]
We define a classical subset $K :\subseteq X$ being \emph{compact} by
\[
    \iscompact\ K \equivto
    \some{f : \sopen\  X \to \dS} \lall{U : \sopen\  X} \defineds{(f\ U)} \,\liff\, K \subseteq U.
\]
Consequently, the type of compact subsets is defined as
\[
    \scompact\ X \equivto
    \some{K :\subseteq X} \iscompact\ K.
\]
\end{defi}
\noindent
That is, a subset $K :\subseteq X$ is compact if and only if, for any open set $U$, the statement $K \subseteq U$ is semi-decidable. Equivalently, for any semi-decidable predicate $P : X \to \Prop$, the statement  $\lall{x \in K} P\  x$ is semi-decidable.

Let us demonstrate the definition by stating a few simple properties of compact sets.
\begin{lem}[\coqref{is\_compact\_union}{/formalization/Hyperspace/Subsets.v}{203}, \coqref{is\_compact\_intersection}{/formalization/Hyperspace/Subsets.v}{229}, \coqref{singleton\_compact}{/formalization/Hyperspace/Subsets.v}{239}]\label{lem:compact-facts}
\leavevmode
\begin{enumerate}
\item If $K_1,K_2 :\subseteq X$ are compact, then so is their union $K_1 \cup K_2$.
\item If $K :\subseteq X$ is compact and $A :\subseteq X$ is closed, then their intersection $K \cap A$ is compact.
\item For any $x : X$,  the singleton set $\{x\}$, 
defined by $\lam{y:X}y = x$, is compact. 
\end{enumerate}
\end{lem}

\noindent
For a space $X$ where inequality is semi-decidable, i.e.,
\[
\all{x,y : X} \some{s : \dS} (\defineds{s} \liff x \neq y)
\]
which for represented spaces is equivalent to the space being computably T2 \cite[Proposition 6.2]{Pauly2016OnTT}, and includes the Polish spaces in \autoref{sec:polish}, the following lemma becomes useful.
\begin{lem}[\coqref{compact\_closed\_ineq\_semidec}{/formalization/Hyperspace/Subsets.v}{251}, \coqref{compact\_closed}{/formalization/Hyperspace/Subsets.v}{258}]\label{lem:compact-closed}
The following are equivalent.
\begin{enumerate}
 \item Every compact set $K :\subseteq X$ is closed.
 \item Inequality on $X$ is semi-decidable.
\end{enumerate}
\end{lem}
\begin{proof}
Formally, we prove the following statement:
\[
\left (\all{K : \scompact\ X} \isclosed\ K \right) \liff \left(\all{x,y : X} \some{s : \dS} (\defineds{s} \liff x \neq y)\right).
\]
If every compact set is closed then singletons are closed and thus inequality is semi-decidable.
Now assume $K :\subseteq X$ is compact. 
We need to show that $x \notin K$ is semi-decidable.
As $\{x\}$ is closed, $U \equivto X \setminus \{x\}$ is open.
Thus, to test $x \notin K$ it suffices to test if $K \subseteq U$.
This shows that $K$ is closed.
\end{proof}
\begin{cor}[\coqref{is\_compact\_countable\_intersection}{/formalization/Hyperspace/Subsets.v}{277}]
Assume inequality on $X$ is semi-decidable and $K: \dN \to \scompact\ X$ is a sequence of compact subsets of $X$. Then their countable intersection $\bigcap K$ is compact as well.
\end{cor}
\begin{proof}
By \autoref{lem:compact-closed} each $K_i$ is also closed and thus by \autoref{lem:open-facts}, $\bigcap  K$ is closed as well.
As $\bigcap K = K_0 \cap \bigcap_i K_{i+1}$, the statement follows from \autoref{lem:compact-facts}.
\end{proof}
\noindent
Using continuity of the test-function, we can relate our definition of compactness to the constructive version of the classical finite subcover definition of compactness.
That is, we show the following lemma.
\begin{lem}[\coqref{compact\_fin\_cover}{/formalization/Hyperspace/Subsets.v}{302}]\label{lem:compact-fin-cover}
Let $K :\subseteq X$ be compact. 
Then every open cover of $K$ has a finite subcover.
Formally,
\[
\all{U : \dN \to \sopen\ X} K \subseteq \bigcup U \rightarrow \mval \some{m : \dN} K \subseteq \bigcup_{n=0}^m (U\ n).
\]
Here, the finite union $\bigcup_{n=0}^m (U\ n)$ is defined by primitive recursion on $m : \dN$.
\end{lem}
\begin{proof}
As $K$ is compact, we can test if an open set covers $K$ and as the countable union of open sets is open (\autoref{lem:open-facts}), we can also test $K \subseteq \bigcup U$.
Thus, there is a function $f: (\dN \to \sopen\ X) \to \dS$ that tests if a countable union of open sets covers $K$, i.e., such that for any $V : \dN \to \sopen\ X$, 
$\defineds{(f\ V)} \liff K \subseteq \bigcup V$.
Furthermore, since we know that $U$ is a cover of $K$, it holds that $(f\ U)$ is defined.
Thus by applying continuity (\autoref{ax:continuity}) on $f$ and $U$, we get 
\[
\mval \some{m : \dN} \all{V : \dN \to \sopen\ X} \restr{V}{m} = \restr{U}{m} \rightarrow \defineds{(f\ V)}.
\]
Let us define $V : \dN \to \sopen\ X$ by $V\ i = U\ i$ for $i < m$ and $V\ i = U\ m$ otherwise.
Then $\bigcup V = \bigcup_{n=0}^m (U\ n)$ and $K \subseteq \bigcup V$ which shows that $\bigcup_{n=0}^m (U\ n)$ is a finite subcover of $U$.
\end{proof}

\noindent
The above definition of compactness has a dual notion which is known as overtness and topologically corresponds to the lower Vietoris topology \cite{vietoris1922bereiche}.
\begin{defi}[\coqref{overt}{/formalization/Hyperspace/Subsets.v}{383}]
A classical subset $V :\subseteq X$ is defined to be \emph{overt} if
\[
    \isovert\ V \equivto
    \some{f : \sopen\  X \to \dS} \lall{U : \sopen\  X} \defineds{(f\ U)} \,\liff\, V \cap U \neq \emptyset
\]
and the type of overt subsets is defined by 
\[
\sovert\ X \equivto \some{V :\subseteq X}\isovert\ V.
\]
\end{defi}
\noindent
Thus computationally, we can effectively verify if an open set intersects the set, or equivalently we can semi-decide $\csome{x \in A} P\  x $ for a semi-decidable property $P$.
Overtness usually does not appear in classical mathematics, as any function defined in the above way is continuous, i.e.\ classically every set is overt.
However, once computability is required, the notion becomes non-trivial and it has numerous applications in computable analysis and constructive mathematics \cite{Pauly2016OnTT,TaylorP:fofct,escardo2004synthetic,DBLP:journals/jla/CoquandS10}.

Similarly to compactness we can show some basic facts about overt sets.
\begin{lem}[\coqref{overt\_countable\_union}{/formalization/Hyperspace/Subsets.v}{409}, \coqref{overt\_open\_intersection}{/formalization/Hyperspace/Subsets.v}{424}, \coqref{singleton\_overt}{/formalization/Hyperspace/Subsets.v}{385}]
\leavevmode
\begin{enumerate}
    \item Given a countable family $V : \dN \to \sovert\ X$ of overt subsets,  
    the countable union $\bigcup V$ is also an overt subset.
\item Let $V :\subseteq X$ be overt and $U :\subseteq X$ open, then their intersection $V \cap U$ is overt.
\item For any $x : X$,  the singleton set $\{x\}$ is overt. 
\end{enumerate}
\end{lem}
\noindent
As equality on represented spaces is only semi-decidable for discrete spaces \cite[Theorem 8.2]{Pauly2016OnTT}, the following lemma shows that for most spaces we consider the intersection of overt sets is not overt in general.
\begin{lem}[\coqref{overt\_intersection\_eq\_semidec}{/formalization/Hyperspace/Subsets.v}{438}]
For a type $X$, assume that for all overt sets $V_1, V_2 : \sovert\ X$, their intersection $V_1\cap V_2$ is overt. Then equality on $X$ is semi-decidable.
Formally, 
\[
\left(\all{V_1, V_2 : \sovert\ X} \isovert\  (V_1 \cap V_2)\right) \rightarrow \left(\all{x,y : X} \some{s : \dS} (\defineds{s} \liff x = y)\right)
\]
\end{lem}
\noindent
The lemma emphasizes again that one should not consider the notions defined in this section as properties of the point sets: Both $\emptyset$ and singletons $\{x\} :\subseteq X$ are overt, and the set $\{x\} \cap \{y\}$ obviously coincides with one of them. However, overtness would provide us with a decision procedure to decide which of the cases we are in, which is not computable in general.

\begin{rem}
    As function extensionality holds, all types of subsets defined in this section, 
    $\Csubset\ X$, 
    $\sopen\ X$,
    $\sclosed\ X$, $\sovert\ X$, and 
    $\scompact\ X$
    are extensional in the sense that they are determined by their underlying point sets.
    To see this, consider, for example, $V_1, V_2 : \sovert\ X$ and write $V_i=(A_i,o_i)$, where $A_i \subseteq X$ and $o_i: \isovert\ A_i$ is an overt witness for $A_i$. Write further \(o_i=(f_i,p_i)\), where $f_i:\sopen\ X\to \dS$ is the semidecision procedure witnessing overtness and $p_i$ certifies its correctness.
    If $V_1 = V_2$, then trivially $A_1 = A_2$, i.e.\ the underlying point sets are equal.
    Conversely, if $A_1 = A_2$, then for every open subset $U:\sopen\ X$, 
    \[ \defineds{(f_1\ U)} \liff A_1\cap U\neq\emptyset \liff A_2 \cap U \neq \emptyset \liff \defineds{(f_2\ U)}. \]
    Thus, by function extensionality $f_1 = f_2$ and therefore $o_1 = o_2$, and it follows that
    \[V_1 = V_2 \liff (\lall{x : X}x \in V_1 \liff x \in V_2).\]
\end{rem}

\section{Subsets of Polish Spaces}\label{sec:polish}
Using basic properties of the Kleenean and Sierpinski spaces, we get short and elegant proofs for simple properties of open sets.
However, from the point of view of doing actual computations, using Sierpinski-valued functions does not allow exploiting any additional structure of the underlying space. Consequently, programs extracted from abstract proofs often fail to take advantage of more efficient algorithms that are available for concrete representations.
As we are interested in extracting exact real computation programs, in \cite{DBLP:conf/mfcs/Konecny0T23} we therefore focus on subsets of Euclidean spaces and how to encode optimized efficient algorithms by making use of particular properties of the representation.
On the other hand, many of the statements can at least be generalized to complete separable metric spaces, while still encoding almost the same algorithms.
As computations on general metric spaces have many applications in mathematics, we therefore decided to generalize our theory to this setting.
For completeness, we also prove many statements that are less relevant for using the extracted programs in actual computations, such as the equivalence of different representations of subsets. 
Still, those results are often useful in proofs and also help to build a unified theory in our formal development.
\subsection{Metric Spaces and Completeness}
A central notion in computable analysis is that of a computable metric space.
A computable metric space is a triple $(X, \distxop, \alpha)$ such that $(X, \distxop)$ is a metric space, $\alpha: \dN \to X$ is an effective enumeration of a dense subset of $X$, and such that the metric is computable on the dense subset.
Computable metric spaces are a central topic in computable analysis.
We introduce the following notions in our formalization.
\begin{defi}[\coqref{metric}{/formalization/Hyperspace/MetricSubsets.v}{14}]\label{def:metric}
For a type $X$, a \emph{metric} on $X$ is a mapping $\distxop: X \to X \to \dR$ such that 
\begin{enumerate}
\item$\lall{x, y : X} \distx{x}{y} = 0 \liff x = y$, 
\item $\lall{x, y : X} \distx{x}{y} = \distx{y}{x}$, and  
\item $\lall{x, y,  z: X} \distx{x}{y} \leq \distx{x}{z} + \distx{z}{y}$.
\end{enumerate}
We further say $X$ is a \emph{metric space} if there exists a metric on $X$, i.e., 
\[
\mathrm{metric}\;X \equivto \some{d : X \to X \to \dR} d \text { is a metric on } X 
\]
\end{defi}
\noindent
Although a metric space is formally a dependent pair of a type $X$ and a metric on it, we write both the metric space and the underlying type $X$ and its metric $\distxop$ when there is no possible confusion.

In exact real computation, the limit operation on sequences of real numbers is crucial for generating new real numbers from their approximating sequences and thereby leaving the realm of symbolic and algebraic computation.
We are mostly interested in complete metric spaces, i.e., spaces that contain the limits of sequences.
Similarly to the real number case, we introduce the following definitions.
\begin{defi}[\coqref{metric\_is\_fast\_cauchy}{/formalization/Hyperspace/MetricSubsets.v}{24},\coqref{metric\_is\_fast\_limit}{/formalization/Hyperspace/MetricSubsets.v}{25},\coqref{complete}{/formalization/Hyperspace/MetricSubsets.v}{1556}]
Let $X$ be a metric space.
\begin{enumerate}
\item A sequence $f: \dN \to X$ is called a \emph{(fast) Cauchy sequence} if 
\[ \mathrm{is\_fast\_cauchy}\  f \equivto \lall{n, m: \dN} \distx{(f\ n)}{(f\ m)} \leq 2^{-n} + 2^{-m}. \]
\item A point $x : X$ is called the \emph{(fast) limit of} $f: \dN \to X$ if 
\[ \mathrm{is\_fast\_limit}\ x\ f \equivto \lall{n : \dN} \distx{(f\ n)}{x} \leq 2^{-n}. \]
\item A subset $A :\subseteq X$ is \emph{complete} if it contains all limits of fast Cauchy sequences, i.e., 
\[ \iscomplete\ A \equivto \all{f : \dN \to A} \mathrm{is\_fast\_cauchy}\ f \rightarrow \some{x \in A}  \mathrm{is\_fast\_limit}\ x\ f.
\]
\end{enumerate}
\end{defi}
\noindent
We say that $X$ is a \inlinedef{complete metric space} if $\iscomplete\ X$ holds.
It is easy to see that if the limit exists, it is unique (\coqref{metric\_limit\_unique}{/formalization/Hyperspace/MetricSubsets.v}{431}).
Thus, we typically write $\lim f = x$ instead of $\mathrm{is\_fast\_limit}\ x\ f$ and for a complete metric space we can define $\lim: (\dN \to X) \pto X$ as the partial operator mapping fast Cauchy sequences to their limits.

Again, we assume that there is an implicit coercion from the type of dependent pairs of a metric space and its completeness, to the underlying metric spaces.

As working with partial functions often leads to difficulties, it is useful to extend the limit operation to a total (albeit nondeterministic) function.
\begin{lem}[\coqref{extended\_limit}{/formalization/Hyperspace/MetricSubsets.v}{461}]\label{lem:elim}
Let $X$ be a complete metric space, then there exists a total nondeterministic extension of partial $\lim$.
Formally,
\[
\all{f : \dN \to X} \mval \some{x : X}  \mathrm{is\_fast\_cauchy}\ f \rightarrow \mathrm{is\_fast\_limit}\ x\ f.
\]
\end{lem}
\begin{proof}
The main idea is that given any sequence $f: \dN \to X$, we can nondeterministically extend it to a converging sequence (\coqref{to\_convergent}{/formalization/Hyperspace/MetricSubsets.v}{242}).
That is, we want to show
\[
\begin{array}{rl}
\all{f : \dN \to X} &\mval \some{g : \dN \to X}\\
&\quad(\mathrm{is\_fast\_cauchy}\  g)\  \land \\
&\quad(\lall{x : X} \mathrm{is\_fast\_limit}\ x\ f \rightarrow  \mathrm{is\_fast\_limit}\ x\ g).
\end{array}
\]
As $g$ is a fast Cauchy sequence, its limit exists and thus choosing $x = \lim g$ then proves the lemma.

Let us thus show how to define $g$.
To simplify the proof, we first speed up $f$, i.e., we replace $f$ by a function $f'$ such that
\[
\mathrm{is\_fast\_cauchy}\ f \rightarrow \left(\lall{n : \dN} \distx{(f'\ n)}{(f'\ (n+1))} < 2^{-(n+2)}\right) \land \lim\ f = \lim\ f'
\]
which can be done by explicitly setting $f' = \lam{n : \dN} f\ (n+3)$.

The slightly stronger condition than just being a fast Cauchy sequence has the advantage that we only need to check neighboring elements instead of arbitrary indices and guarantees that we can turn any initial segment of $f'$ into a fast Cauchy sequence simply by repeating the last element indefinitely.

Consequently, we want to choose $g$ such that it coincides with $f'$ as long as the condition is not violated, and, as soon as it gets violated, continue with the last element.
As the condition only requires to compare real numbers, violation of the condition is semi-decidable.
The difficulty, however, is to choose $g\ n$ in finite time in the case that the condition is not violated.
This is where nondeterminism comes in handy.
By the nondeterministic dependent choice axiom, it suffices to define a nondeterministic refinement procedure such that any sequence it defines satisfies the requirements we need.

We use the following procedure.
\begin{enumerate}
\item Start with $g\ 0 = f'\  0$.
\item Assume we are given $(g\  n)$ and we need to choose $(g\ (n+1))$.
Note that both $\distx{(f'\ n)}{(f'\ (n+1))} < 2^{-(n+1)}$ and $\distx{(f'\ n)}{(f'\ (n+1))} > 2^{-(n+2)}$ are semi-decidable and that at least one of the conditions holds.
Thus, we can nondeterministically choose one of them.
In the first case, set $(g\ (n+1)) = (f'\ (n+1))$, and in the second case set $(g\ (n+1)) = (g\ n)$. 
Further, remember which choice we have taken.
\item After choosing $(g\ (n+1)) = g\ n$ in the second step, choose $g\ m = g\ n$ for all following $m : \dN$. 
\end{enumerate}
Let us first prove that any $g$ defined in this way is a fast Cauchy sequence.
It suffices to show that $\lall{n : \dN} \distx{(g\ n)}{(g\ (n+1))} \leq 2^{-(n+1)}$.
By how we defined the refinement procedure, 
we either have $(g\ (n+1)) = (f'\ (n+1))$ or $(g\ (n+1)) = g\ n$.
In the second case, the statement is obviously true.
By steps two and three of the procedure, the first case can only occur if $g\ n = f'\ n$ and $\distx{(f'\ n)}{(f'\ (n+1))} < 2^{-(n+1)}$.

Further, by the way we chose $f'$, choosing the second case in step two of the procedure means that the original $f$ must have violated the condition of being a fast Cauchy sequence.
Thus, in case that $f$ is a fast Cauchy sequence, we always choose $g\ n = f'\ n$ and thus $g = f'$, which by construction has the same limit as $f$.
\end{proof}
\begin{rem}
As can be seen in the proof, the nondeterminism in \autoref{lem:elim} stems from the fact that whenever the input is not a fast Cauchy sequence, we cannot deterministically choose the index where the sequence is made constant precisely, thus allowing slightly different values for the extended limit.
\end{rem}

\noindent
As with the limit operator, let us further introduce the notation $\elim: (\dN \to X) \to \mval X$ for the total extension of the limit function.

Before proceeding, let us note that any space $X$ in our setting already carries an intrinsic topology via the notion of open set from \autoref{def:open-closed}.
On the other hand, for a metric space $X$, the natural topology is the usual topology generated by open balls in the metric. As there is no mention of topology in our definition of a metric space, there is a priori no obvious reason why these two topologies should be related in any way. 
It may therefore be surprising at first glance that we later, in \autoref{sec:metric-continuity}, show that they indeed coincide for our definition.
The key ingredient is the continuity of the limit operator, which ties continuous functions to metric approximations and ultimately forces the intrinsic and metric topologies to agree.
For now, however, all mentions of open, closed, etc. should be understood as the intrinsic notions from \autoref{def:open-closed}, and not as anything related to the metric structure.

The extended limit operator turns out to be useful in proving the following standard result from topology.
\begin{lem}[\coqref{closed\_complete}{/formalization/Hyperspace/MetricSubsets.v}{1563}]\label{lem:closed_complete}
Every closed subset of a complete metric space is complete.
\end{lem}
\begin{proof}
Assume $X$ is complete, $A :\subseteq X$ is closed and $a : \dN \to X$ is a fast Cauchy sequence with values in $A$.
As in the proof of \autoref{lem:elim} we can assume that $a$ converges fast enough to make sure that every sequence we get from extending an initial segment of $a$ to an infinite sequence by infinitely repeating the last element, the new sequence will also be a fast Cauchy sequence.

As $X$ is complete we can find the limit $x = \lim a$.
We need to show that $x \in A$.
The proof is by contradiction.
Since $A$ is closed there is a function $f : X \to \dS$ such that $\defineds{(f\ x)} \liff x \notin A$.
Thus, by composing $f$ with the extended limit $\elim$, we get
\[
\all{b : \dN \to X}  \mval \some{s : \dS} \mathrm{is\_fast\_cauchy}\ b \rightarrow  (\defineds{s} \liff \lim b \notin A).
\]
Further, as the property of not being a fast Cauchy sequence is semi-decidable,
\[
\all{b : \dN \to X} \some{t : \dS} \defineds{t} \liff \neg \mathrm{is\_fast\_cauchy}\ b.
\]
Combining the two semi-decidable procedures by disjunction, we get for any $b : \dN \to X$
\[
 \mval \some{s : \dS} (\mathrm{is\_fast\_cauchy}\ b \rightarrow  (\defineds{s} \liff \lim b \notin A)) \land  (\neg \mathrm{is\_fast\_cauchy}\ b \rightarrow \defineds{s}).
\]
That is, by requiring to always return $\top$ when the input is not a fast Cauchy sequence, we ensure that for any $b$, there is a unique $s : \dS$ with the property.
We can therefore use singleton elimination to remove the nondeterminism, and thus there is a deterministic function $g: (\dN \to X) \to \dS$ such that 
\[
\lall{b : \dN \to X} \mathrm{is\_fast\_cauchy}\ b \rightarrow (\defineds{(g\ b)} \liff \lim b \notin A).
\]

Now assume that $\lim a \notin A$.
Then, since $a$ is a fast Cauchy sequence and $\defineds{(g\ a)}$, by continuity of $g$ we have
\[
\mval \some{m : \dN} \all{b : \dN \to X} \restr{b}{m} = \restr{a}{m} \rightarrow \defineds{(g\ b)}.
\]
Define a sequence $b : \dN \to X$ by $b\ n = a\ n$ for $n < m$ and $b\ n = a\ m$ otherwise.
Then $b$ is a fast Cauchy sequence with limit $a\ m$ and $(g\ b)$ is defined.
But then $(a\ m) \notin A$, contradicting the assumption that $a$ is a sequence taking values in $A$.
\end{proof}
\begin{rem}
Note that the other direction, i.e., every complete set is closed, is not true constructively.
For example, in the space of natural numbers with the discrete topology, every subset is complete, but the Halting set is not closed.
\end{rem}

\noindent
The last property of metric spaces we need is separability, i.e., that there exists a countable dense subset.
\begin{defi}[\coqref{dense}{/formalization/Hyperspace/MetricSubsets.v}{18}]
A subset $A :\subseteq X$ is defined to be \emph{dense} in $X$ by 
\[
\mathrm{is\_dense}\ A \equivto \all{ U : \sopen\ X} U \neq \emptyset \rightarrow  A \cap U \neq \emptyset
\]
In general, we call a type $Y$ \emph{enumerable} if there exists a map $f: \dN \to Y$ which is classically surjective, i.e., $\lall{y : Y} \csome{n : \dN} f\ n = y$.
Thus, we define the space to be \emph{separable} by 
\[
\mathrm{is\_separable}\ X \equivto \some{e : \dN \to X} \mathrm{is\_dense}\ \lam{x : X} \csome{n : \dN} (e\ n) = x.
\]
\end{defi}

\noindent
We call a complete, separable metric space a \inlinedef{Polish space}.
Given a Polish space $X$, we simply write $X$ to denote also the underlying metric space, when there is no possible confusion.

Although some statements can be proven in a more general setting, for the remainder of the paper we will for simplicity typically assume that $X$ is a Polish space. 
The interested reader can also compare the statements in the Coq formalization, which sometimes do not require all the assumptions. 

\subsection{Open Balls and Metric Continuity}\label{sec:metric-continuity}
For $c : X$ and $n : \dN$, we denote by $\ball(c,n)$ the open ball with center $c$ and radius $2^{-n}$, i.e., the classical subset defined by 
\[x \in \ball(c,n) \liff (\distx{x}{c})< 2^{-n}. \] 
It is easy to see that balls defined in this way are actually open:
\begin{lem}[\coqref{metric\_open}{/formalization/Hyperspace/MetricSubsets.v}{118}]\label{lem:open-ball}
For all $c : X$ and $n : \dN$, the subsets $\ball(c,n)$ are open in the sense of \autoref{def:open-closed}.
\end{lem}

\noindent
For a Polish space we can use separability to choose a point from an open set.
\begin{lem}[\coqref{separable\_open\_choice}{/formalization/Hyperspace/MetricSubsets.v}{79}]\label{lem:open-choice}
Let $X$ be a Polish space and suppose $U :\subseteq X$ is open and (classically) nonempty. 
Then we can nondeterministically choose an element from the dense subset that is contained in $U$.
That is,
\[
\all{U : \sopen\ X} U \neq \emptyset \rightarrow \mval\some{n : \dN} e_n \in U.
\]
where $e$ is the countable dense sequence of $X$. 
\end{lem}
\begin{proof}
Since $e$ is a dense sequence, $\csome{n : \dN} e_n \in U$ holds.
That is, there classically exists an element from the dense subset contained in $U$.
As $U$ is open $\all{n : \dN} \some{s : \dS} \defineds{s} \liff e_n \in U$.
Thus, we can define a function $f: \dN \to \dS$ such that $(f\ n)$ is defined iff $e_n \in U$.
Thus applying the nondeterministic countable selection principle (\autoref{ax:cchoice}) on $f$ completes the proof.
\end{proof}
\begin{cor}[\coqref{separable\_open\_overt}{/formalization/Hyperspace/MetricSubsets.v}{90}]
Let $X$ be a Polish space and $U :\subseteq X$ be open. 
Then $U$ is overt.
\end{cor}
\begin{proof}
For any open $U' :\subseteq X$, $U \cap U'$ is open.
Thus $U \cap U' \neq \emptyset$ iff we can find an $i : \dN$ with $e_i \in U \cap U'$.
\end{proof}
\noindent Examples of overt sets that are not open are  singleton sets in the space of real numbers.

Next we show that we can approximate any $x : X$ arbitrarily well with elements of the dense subset, and additionally if $x$ is contained in some open set $U :\subseteq X$, we can choose the approximation to be also contained in $U$.
\begin{cor}[\coqref{separable\_metric\_approx\_inside}{/formalization/Hyperspace/MetricSubsets.v}{127}]\label{cor:approx}
Let $X$ be a Polish space  and suppose $U :\subseteq X$ is open.
Then,
\[
\all{x \in U} \all{n : \dN} \mval \some{i : \dN} e_i \in U \  \land \  \distx{x}{e_i} \leq 2^{-n}.
\]
\end{cor}

\noindent
It is a well-known fact from classical topology that every separable metric space is second countable, i.e., every open set can be written as the union of elements from a countable collection of open subsets, called a base for the space.

A possible choice for such a base is the collection of balls $\ball(e_n, m)$ with $n,m : \dN$ centered at points from the dense subset.
This gives an alternative representation for open subsets of computable metric spaces (sometimes called c.e. open in computable analysis) by effectively exhausting nonempty open sets with base elements.
We prove that this is equivalent to our generic definition of open sets using Sierpinski space.
Let us first formally define a base in our setting.
\begin{defi}
For $b :\subseteq X$, we define
\[
\mathrm{is\_base_X}\  b \equivto (\csome{i\ n : \dN} b = \ball(e_i, n)) \lor b = \emptyset.
\]
\end{defi}
\noindent
That is, a base element is either a ball centered at some point of the dense subset or it is empty.
We choose some arbitrary enumeration of $\dN \times \dN$ and denote the sequence of base elements by $(B_i)_{i : \dN}$. 
We further denote the type of base elements by \[\mathrm{base}\ X \equivto 
\some{b :\subseteq X}\mathrm{is\_base_X}\ b.\]

Our goal is to prove that every Polish space is second countable.
Before we can prove the theorem, we first consider the special case of the Baire space, and later reduce the general case to this case.
As in \autoref{sec:continuity}, we identify the nonempty base elements of $\dN \to \dN$ with finite lists of natural numbers.
To simplify notation we state the theorem only for nonempty subsets here.
\begin{lem}[\coqref{enumerate\_baire\_open}{/formalization/Hyperspace/MetricSubsets.v}{845}]\label{lem:baire-second-countable}
Assume $U :\subseteq (\dN \to \dN)$ is an open, nonempty subset of Baire space. 
Then
\[ \mval \some{F : \dN \to [\dN]} U = \bigcup_{n : \dN} [F\  n] \]
where for a list $l : [\dN]$, by $[l] :\subseteq (\dN \to \dN)$ we denote the subset of sequences extending $l$.
\end{lem}
\begin{proof}
From continuity of the characteristic function we get
\[
\all{x \in U}  \mval\some{m : \dN} [\restr{x}{m}] \subseteq U.
\]
Applying partial Baire choice (\autoref{lem:partial-baire-choice}) we (nondeterministically) get a partial modulus of continuity function, i.e., 
\[ \mval \some{\mu: (\dN \to \dN) \pto \dN} \lall{x \in U} [\restr{x}{(\mu\ x)}] \subseteq 
U.\]
For a list $l$, we define $\lext{l} : \dN \to \dN $ as the infinite sequence defined by starting with the initial segment given by $l$ and then continuing with the last element of the list indefinitely.

As $U$ is open and nonempty we can enumerate all finite lists $l : [\dN]$ such that $\lext{l} \in U$.
Denote the $n$-th element in the enumeration by $l_n$ and define $F \equivto \lam{n : \dN} \lext{(l_n)}_{(\mu\ \lext{l_n})}$.
We claim $U = \bigcup_{n : \dN} [F\ n]$.
First, $\bigcup_{n : \dN} [F\ n] \subseteq U$ obviously holds by the way we defined $F$.

To show that $U \subseteq \bigcup_{n : \dN} [F\ n]$, let $x \in U$.
It suffices to show that there is some finite list $p : [\dN]$ such that $x \in [\lext{p}_{(\mu\ \lext{p})}]$. 
By \autoref{lem:partial-continuity},  $\mu$ is continuous, i.e.,
\[
\mval \some{m : \dN} \lall{y : \dN \to \dN} \restr{x}{m} = \restr{y}{m} \rightarrow (\mu\ x) = (\mu\ y).
\]
Note that classical existence of $p$ suffices and thus the nondeterminism is not an issue here.
Now assume we are given $m : \dN$ as above and let $k \geq \max((\mu\ x),m)$ and $p \equivto \restr{x}{k}$.
By definition of $\mu$, $\lext{p} \in U$ and $\mu\ \lext{p} = \mu\ x$.
Thus, $\mu\ \lext{p} \leq k$ and therefore by definition of $p$, $x$ and $\lext{p}$ coincide on the initial segment of length $(\mu\ \lext{p})$.
It follows  $x \in [\lext{p}_{(\mu\ \lext{p})}]$ as claimed.
\end{proof}
\begin{rem}
In the Coq formalization, we prove the more general fact that any open set (i.e. including the empty set) can be exhausted with base elements consisting either of sets of the form $[l]$ for some finite list of natural numbers $l$ or the empty set.
The statement we prove in the paper then follows by the fact that for nonempty $U$ there has to be at least one nonempty base element in the enumeration and we can replace each occurrence of the empty set by the first such element.
\end{rem}
\noindent
We can now prove that every Polish space is second countable.
\begin{thm}[\coqref{separable\_metric\_second\_countable}{/formalization/Hyperspace/MetricSubsets.v}{1200}]\label{thm:polish-second-countable}
 A subset $U :\subseteq X$ is open iff we can (nondeterministically) find a sequence of base elements exhausting $U$.
Formally, for any $U :\subseteq X$,
\[
   \isopen\  U  \liff   \mval \some{F : \dN \to \mathrm{base\ X}} U = \bigcup_{n : \dN} (F\ n).
\]
\end{thm}
\begin{proof}
By \autoref{lem:open-ball} each of the base elements is open, and thus by \autoref{lem:open-facts} the countable union of base elements is open thus we get $ ( \mval \some{F : \dN \to \mathrm{base}\ X)} U = \bigcup_{n : \dN} (F\ n)  ) \rightarrow \sopen\  U$.

For the other direction, let us assume that $U$ is open in the sense of the previous section, that is, there is a function $f: X \to \dS$ such that $\defineds{(f\ x)}$ iff $x \in U$.
The idea of the proof is to enumerate all $e_i \in E$ and find some $n$ such that $\ball(e_i, n) \subseteq U$ and $n$ large enough to guarantee that the union will cover all $x \in U$.
Note that continuity of $f$ alone is not sufficient for the latter.

To apply \autoref{lem:continuous-modulus}, we therefore reduce the statement to a statement on Baire space.
Define a function $\delta : (\dN \to \dN) \to \mval \dX$ by 
$ \lam{x : \dN \to \dN} \elim\  e \circ x$.
That is, we encode elements $x : X$ by infinite sequences of elements of the dense subset converging to $x$, similar to the standard Cauchy representation for computable metric spaces.
Here, we use the extended limit operator to guarantee that $\delta$ is total.
By \autoref{ax:baire-choice} we can nondeterministically choose a single-valued function instead, i.e., we show
\[
\mval \some{\delta : (\dN \to \dN) \to X } \lall{x : \dN \to \dN} \mathrm{fast\_cauchy}\  (e \circ x) \rightarrow \mathrm{is\_fast\_limit}\  (\delta\ x) \  (e \circ x). 
\]
For any such $\delta$ we can define an open subset $\mathcal{U}$ of the Baire space $\dN \to \dN$ as the preimage of $U$ under $\delta$.
More precisely, we show
\begin{multline}\label{eq:baire-U-def}
    \mval \some{\mathcal{U} : \sopen\ (\dN \to \dN)} \lall{x : \dN \to \dN}  \mathrm{is\_fast\_cauchy}\  (e \circ x) \\
\rightarrow (x \in \mathcal{U} \liff  \lim\  (e \circ x) \in U). 
\end{multline}

\noindent
To see that the chosen $\mathcal{U}$ is open, we go back to the original definition of $\delta$ and note that $\elim (e \circ x) \in U$ can be checked as $U$ is open.

By lifting the nondeterminism, it suffices to show how to define the function $F : \dN \to \mathrm{base}\ X$ from any such $\mathcal{U}$.

To simplify the notation, let us again w.l.o.g. assume that $\mathcal{U}$ is nonempty.
Then, by \autoref{lem:baire-second-countable} there exists a sequence of lists $(l_i)_{i : \dN}$ such that
\begin{equation}\label{eq:baire-U}
\mathcal{U} = \bigcup_{i : \dN} [l_i].
\end{equation}
Note that the case that $\mathcal{U}$ might be empty can be handled analogously,  by allowing the empty set to appear in this enumeration.

Next, we want to remove all lists from the enumeration that cannot be extended to a fast Cauchy sequence when interpreted as enumeration of points in $E$.
First note that as inequality of reals is semi-decidable, we can define a function $t: [\dN] \to \dS$ such that
\[
\lall{l : [\dN]} \defineds{(t\ l)} \liff \lall{n < \abs{l}-1} (\distx{e_{l[n]}}{e_{l[n+1]}}) < 2^{-(n+1)}.
\]
That is, for any list it is semi-decidable if all neighboring elements $e_{l[n]}$ and $e_{l[n+1]}$ are $2^{-(n+1)}$ close to each other.

Using this fact, we can (nondeterministically) choose a subsequence $\{l'_i \mid i : \dN \} \subseteq \{l_i \mid i : \dN \} $ consisting of exactly those $l_i$ with $\defineds{(t\ l_i)}$.
Note that we again w.l.o.g. assumed that there is at least one list with the property to avoid including the empty set in the enumeration and that the order of the enumeration can be different from the original enumeration.

As in the proof of the previous lemma, let $\lext{l}$ denote the extension of the list $l$ to an infinite sequence by repeating the last element indefinitely.
By the way we choose the subsequence, for any $l'_i$, $e \circ \lext{l'_i}$ is a fast Cauchy sequence in $X$.
Finally, we define the function $F$ by $F\ i \equivto \ball((\delta\  \lext{l'_i}), \abs{l'_{i}}+1)$.

We claim that $U = \bigcup_{i : \dN} (F\ i)$.
Let us first show that $\bigcup_{i : \dN} (F\ i) \subseteq U$.
Assume $x \in \bigcup_{i : \dN} (F\ i)$, i.e., there is some $i : \dN$ with $x \in \ball((\delta\  \lext{l'_i}), \abs{l'_{i}}+1)$. 
Let $k : \dN$ be the last element of $l'_i$ and let $j = \abs{l'_i}$, then $\ball((\delta\  \lext{l'_i}), \abs{l'_{i}}+1) = \ball(e_k,j+1)$, i.e., $(\distx{x}{e_k}) < 2^{-(j+1)}$.

Using \autoref{cor:approx} we can show that there is a function $\varphi: \dN \to \dN$ such that $\lall{n : \dN} (\distx{x}{e_{(\varphi\ n)}}) < 2^{-(n+1)}$.
Define $\psi : \dN \to \dN$ by $\psi\  n \equivto l'_i[n]$ for $n < j$ and $\psi\ n \equivto \varphi\ n$ otherwise, i.e., we replace the first elements of $\varphi$ by the list $l'_i$.
As $(\distx{e_k}{e_{(\varphi\ j)}}) < 2^{-j}$ it follows that $\psi$ is a fast Cauchy sequence converging to $x$. 
Further, $\psi$ extends $l'_i$ and therefore by \autoref{eq:baire-U}, $\psi \in \mathcal{U}$.
By \autoref{eq:baire-U-def} it then follows  $x = (\lim e \circ \psi) \in U$ as claimed.

Let us now show the other direction, i.e., that  $U \subseteq \bigcup_{i : \dN} (F\ i)$.
Assume $x \in U$.
By \autoref{cor:approx} we can find a sequence $\varphi : \dN \to \dN$ such that
\begin{equation}\label{eq:varphi-close}
\lall{n : \dN} (\distx{e_{(\varphi\ n)}}{x}) < 2^{-(n+2)}
\end{equation}
$\varphi$ is a fast Cauchy sequence converging to $x \in U$, and therefore by \autoref{eq:baire-U-def} it follows that $\varphi \in \mathcal{U}$.
By \autoref{eq:baire-U}, there is some $i : \dN$ such that $\varphi \in [l_i]$, i.e., $l_i$ is an initial segment of $\varphi$.
By \autoref{eq:varphi-close}, $e \circ \lext{l_i}$ is a fast Cauchy sequence with $\lim\  (e \circ \lext{l_i}) = e_k$.
Thus, $(\delta\ \lext{l_i}) = e_k \in U$ and by \autoref{eq:varphi-close} it follows that $x \in \ball((\delta\ \lext{l_i}), \abs{l_i}+1)$.
Further, by the choice of $\varphi$ it follows that neighboring elements in $l_i$ are close enough to satisfy $\defineds{(t\ l)}$, i.e., $l_i$ is contained in the subsequence $(l'_i)_{i : \dN}$ and thus $x \in \bigcup_{i : \dN} (F\ i)$ as claimed.
\end{proof}

\begin{cor}[\coqref{separable\_metric\_continuous}{/formalization/Hyperspace/MetricSubsets.v}{924}]\label{cor:metric-continuity}
    Every function $f: X \to \dS$ is continuous with respect to the topology induced by the metric $\distxop$.
    That is,
    \[
    \all{f : X \to \dS} \all {x : X} \defineds{(f\ x)} \rightarrow \mval \some{m : \dN} \lall{y \in \ball(x,m)} \defineds{(f\ y)}.
    \]
\end{cor}
\begin{proof}
Let $U :\subseteq X$ be the open set defined by $x \in U \liff \defineds{(f\ x)}$.
By \autoref{thm:polish-second-countable}, we can write $U = \bigcup_{i : \dN} (F\ i)$ for some $F : \dN \to \mathrm{base\ X}$.

Let $x: X$ with $\defineds{(f\ x)}$, i.e., $x \in U$.
By \autoref{ax:cchoice} we get  $\mval \some{i : \dN} x \in (F\ i)$.
As obviously $(F\ i) \neq \emptyset$, there are some $j, n : \dN$ such that $(F\ i)= \ball(e_j,n)$.
As $2^{-n} > \distx{e_j}{x}$, we can further find some $m : \dN$ such that $2^{-m} < 2^{-n}-(\distx{e_j}{x})$.
Then $\ball(x,m) \subseteq \ball(e_j, n) \subseteq U$ which proves the statement.
\end{proof}

\subsection{Bishop-compactness, Locatedness and Compact-Overt Sets}
Subsets that we are particularly interested in are those that can be characterized in terms of `drawings' with arbitrary precision, meaning
that for each $n$ we can generate a picture of $A$ in terms of `pixels' (i.e.\ small boxes or closed balls) of size $2^{-n}$.
As we will see in this section, it turns out that subsets for which such drawings exist are precisely those that are both compact and overt.

Such sets (and slight variations) have been considered in constructive mathematics and computable analysis under different names, sometimes using the same terminology for slightly different concepts.
In Bishop's constructive mathematics, a set is \emph{compact} if it is \emph{complete} and \emph{totally bounded} \cite{bishop1967foundations}.
This definition of compactness is therefore known as \inlinedef{Bishop-compactness}.
Here, totally bounded means that it can be covered by finitely many subsets of fixed size.
More precisely, we use the following definition.
\begin{defi}[\coqref{totally\_bounded}{/formalization/Hyperspace/MetricSubsets.v}{1555}]
For a subset $A :\subseteq X$ of a Polish space we define 
\[
\begin{array}{rll}
\istotallybounded\  A \equivto & \all{n : \dN} \some{L : [X]} \\[1ex]
        & \quad \lall{c \in L} \ball(c,n) \cap A \neq \emptyset\\[1ex]
        & \quad \land \quad  \lall{x \in A} \csome{c \in L} x \in \ball(c,n). 
   \end{array}
\]
\end{defi}
\noindent
Thus, $A$ is totally bounded if for each $n$ we get a finite covering of $A$ by open balls, each with radius at most $2^{-n}$ and such that each ball intersects $A$.
We call the set defined by the union of the elements of the $n$-th list of the sequence the $n$-th \inlinedef{covering} of $A$ and denote it by $A_n$.

In \cite{DBLP:conf/mfcs/Konecny0T23} we considered coverings by closed balls instead of open balls.
It is not difficult to see that they are equivalent.
As some proofs turn out to be simpler when using open balls, and others are simpler using closed balls, we add a proof of the statement, so that in the remainder of the paper we can use both open and closed ball coverings, depending on which makes the proof easier.
\begin{lem}[\coqref{totally\_bounded\_whichever}{/formalization/Hyperspace/MetricSubsets.v}{1775}]
\[
\begin{array}{rll}
\istotallybounded\  A \liff & \all{n : \dN} \some{L : [X]} \\[1ex]
        & \quad \lall{c \in L} \overline{\ball(c,n)} \cap A \neq \emptyset
         \\[1ex]
        &\quad \land \quad \lall{x \in A} \csome{c \in L} x \in \overline{\ball(c,n)}. \\
   \end{array}
\]
\end{lem}
\begin{proof}
If $A$ is totally bounded, then replacing the balls in the covering by their closure suffices.
If we have a covering of $A$ with closed balls as above, we can show that $A$ is totally bounded by defining the balls for the $n$-th covering by taking the $(n+1)$-st closed covering and use open balls with the same centers and radius $2^{-n}$.
\end{proof}

\noindent
Further, sometimes it is required that in the definition of a set being totally bounded, the centers of the coverings need to be contained in the set itself.
For a complete subset, the notions are again equivalent.
\begin{lem}[\coqref{bishop\_compact\_centered}{/formalization/Hyperspace/MetricSubsets.v}{2749}]\label{lem:tot-bounded-strong}
Let $A :\subseteq \dX$ be totally bounded and complete.
Then,
\[ \mval \lall{n : \dN} \some{L : [A]} \lall{x \in A} \csome{c \in L} x \in \ball(c,n). \]
\end{lem}
\begin{proof}
By applying the nondeterministic countable choice, it suffices to show that for each $n : \dN$, we can nondeterministically find an $n$-th covering of $A$ with centers in $A$.
Choose the covering $A_{n+2}$.
For each center $c : \dX$ in the covering, by \autoref{lem:located-choice} we can find a $c' \in A$ with $\distx{c'}{c} < 2^{-(n+1)}$.
Thus $\ball(c,n+1) \subseteq \ball(c',n)$ and replacing the center with $c'$ suffices for the $n$-th covering.
\end{proof}
\noindent
Another related notion that plays a central role in constructive mathematics is locatedness.
A subset of a Polish space is \inlinedef{located} if the distance between the set and a point is computable.
Naturally, there is a tight connection between locatedness and overtness (cf. \cite{spitters2010locatedness}), however, we currently do not explore this completely in our formalization and only show a few facts that are useful for our applications.

For $A :\subseteq X$ and $x \in X$, let us write $d(x,A)$ for $\inf_{y \in A} \distx{x}{y}$.
We define a classical proposition, stating that a real number $r: \dR$ is the distance of $A$ and $x$ as follows:
\[
\mathrm{is\_dist}\ A\ x\ r \equivto \left(\lall{y \in A} \distx{x}{y} \geq r \right) \  \land \  \left( \lall{s : \dR}(\lall{y \in A} \distx{x}{y} \geq s) \rightarrow s \leq r \right) .
\]
By applying classical completeness, it is easy to show that the distance to any nonempty set exists classically and is unique.
\begin{lem}[\coqref{classical\_dist\_exists}{/formalization/Hyperspace/MetricSubsets.v}{2019}]\label{lem:classical-dist-exists}
Let $A :\subseteq X$ and $A \neq \emptyset$, then
\[ \lall{x : X} \usome{r : \dR} \mathrm{is\_dist}\ A\ x\ r \]
where $\usome{x : X}P\ x$ is defined naturally for classical unique existence in $\Prop$.
\end{lem}
\noindent
We formally define locatedness as follows.
\begin{defi}[\coqref{located}{/formalization/Hyperspace/MetricSubsets.v}{1933}]
For a subset $A :\subseteq X$ of a Polish space we define 
\[
  \islocated\  A \equivto  \all{x : X} \some{r : \dR} \mathrm{is\_dist}\ A\ x\ r.  
\]
\end{defi}
\noindent
For a located set $A$, we  write $d_A : X \to \dR$ for the distance function $x \mapsto r$ witnessing the locatedness of $A$.
Note that the empty set is not located according to our definition.

It is possible to choose a point from a complete, located set:
\begin{lem}[\coqref{located\_choice}{/formalization/Hyperspace/MetricSubsets.v}{2717}]\label{lem:located-choice}
Let $A :\subseteq X$ be complete and located, then $\mval \some{x : X} x \in A$. 
\end{lem}
\begin{proof}
Since $A$ is located, it has a distance function
$d_A:X\to\dR$. Consider the open set
\[
U_2:=\{x:X\mid d_A(x)<2^{-2}\}.
\]
It is semidecidable, and it is nonempty because $A\subseteq U_2$.
Hence, by \autoref{lem:open-choice}, we can nondeterministically choose
$c_0\in U_2$.

We construct a sequence $(x_n)_{n\in\dN}$ inductively starting with $x_0=x_1=x_2=c_0$. Suppose that $n\geq 2$ and that $x_n$ has been chosen with
$d_A(x_n)<2^{-n}$. Consider the open set
\[
V_n
 :=
 \{y:X\mid d_A(y)<2^{-(n+1)}\}
 \cap
 B(x_n,n-1).
\]
 It is nonempty as there exists $a\in A$ such that
\[
d_X(x_n,a)<d_A(x_n)+2^{-n}< 2^{-(n-1)}
\]
and since $d_A(a)=0<2^{-(n+1)}$, $a\in V_n$. Therefore,
\autoref{lem:open-choice} yields a point $x_{n+1}\in V_n$. In
particular,
$d_A(x_{n+1})<2^{-(n+1)}$ and $d_X(x_n,x_{n+1})<2^{-(n-1)}$.
Using nondeterministic dependent choice, these successive choices can
be combined into an infinite sequence.

Define $f_n:=x_{n+2}$ then $f$ is a fast Cauchy sequence and thus has a limit $x=\lim f$. 
It remains to show that
$x\in A$. 
For each $n$, there is point $a_n\in A$ satisfying
$
d_X(f_n,a_n) < 2^{-(n+1)}.
$
Consequently, after shifting the sequence $(a_n)_n$ by a fixed number
of indices, it is a fast Cauchy sequence in $A$ with limit $x$.
Completeness of $A$ therefore implies $x\in A$, as required.
\end{proof}

\begin{lem}[\coqref{totally\_bounded\_located}{/formalization/Hyperspace/MetricSubsets.v}{2509}]\label{lem:bc-located}
Assume $A$ is nonempty and totally bounded. Then $A$ is located.
\end{lem}
\begin{proof}
We need to show that for any nonempty totally bounded set we can compute the distance function.
That is, we want to show
\[
 \istotallybounded \ A \rightarrow A \neq \emptyset \rightarrow \all{x : X} \some{r : \dR}\mathrm{is\_dist}\ A\ x\ r.
\]
The idea is to show that for any $n : \dN$ we can get a $2^{-n}$ approximation of the distance (which classically exists according to \autoref{lem:classical-dist-exists}).
We can then use the limit operator to compute the limit.
The proof consists of two parts:
\begin{enumerate}
\item We show that the distance to the $n$-th cover approximates the distance well. More precisely, we show the following statement:
\[
\mathrm{is\_dist}\ A\ x\ r \rightarrow \mathrm{is\_dist}\ A_n\ x\ r' \rightarrow \abs{r-r'} \leq 3 \cdot 2^{-n}.
\]
\item We show that we can approximate the distance to the $n$-th approximation well enough:
\[
\all{n : \dN} \some{r : \dR} \all{r' : \dR} \mathrm{is\_dist}\ A_n\ x\ r' \rightarrow \abs{r-r'} \leq 2 \cdot 2^{-n}.
\]
\end{enumerate}
Then obviously approximating the distance to the $(n+3)$-rd covering gives a $2^{-n}$ approximation to the distance of $x$ and $A$.
For 1., it suffices to show that $A \subseteq A_n$ and that for any $x \in A_n$ there is a $y \in A$ with $\distx{x}{y} \leq 2 \cdot 2^{-n}$.
For 2., we show that $d(x,A \cup B) = \min(d(x,A),d(x,B))$.
Since $A_n$ is a finite union of balls, we can thus approximate the distance to $A_n$ by approximating the distance to each ball.
Now, $\distx{c}{x}$ is a good enough approximation to $d(x,\ball(c,n))$.
\end{proof}

\begin{lem}[\coqref{compact\_overt\_totally\_bounded}{/formalization/Hyperspace/MetricSubsets.v}{1680}, \coqref{compact\_complete}{/formalization/Hyperspace/MetricSubsets.v}{1914}]
Let $X$ be Polish and $K :\subseteq X$ be compact and overt, then $K$ is complete and totally bounded.
\end{lem}
\begin{proof}
Inequality on metric spaces is semi-decidable, thus by \autoref{lem:compact-closed} every compact set is closed and by \autoref{lem:closed_complete} complete.
It remains to show that $K$ is totally bounded.
As the exact choice of coverings can be nondeterministic, we formally show the statement
\[
 \iscompact\ K \rightarrow \isovert\ K \rightarrow \mval\ ( \istotallybounded\ K).
\]
For this, it suffices that for any $n : \dN$, we can nondeterministically find a list $L : [X]$ such that the balls of radius $2^{-n}$ centered at the elements of $L$ cover $K$ and each of the balls intersects $K$.
First note that for a set that is both compact and overt, it is decidable if the set is empty or not.
In the case that the set is empty we can thus simply return the empty list.
Now assume that $K$ is nonempty.
By compactness of $K$ and \autoref{lem:compact-fin-cover} it suffices to find a countable cover with the desired property.
Since $K$ is overt, we test if the open ball of radius $2^{-n}$ centered at $e_i$ intersects $K$ for any element $e_i : X$ of the dense subset of $X$.
More precisely, we can define a sequence $f: \dN \to \dS$ by $\defineds{(f\ i)} \liff \ball(e_i,n) \cap K \neq \emptyset$.
As there exists an $x \in K$ and some element of the dense subset is $2^{-n}$ close to that $x$, there is at least one $i : \dN$ such that $\defineds{(f\ i)}$.
Thus, we can (nondeterministically) choose a sequence $g: \dN \to \dN$ such that the subsequence defined by $g$ contains exactly those $e_j$ for which $\ball(e_j, n)$ intersects $K$. 
It remains to show that this defines a cover of $K$.
For $x \in K$, we can find an $i : \dN$ such that $\distx{x}{e_i} < 2^{-n}$.
Thus $x \in \ball(e_i, n)$ and therefore $\ball(e_i,n) \cap K \neq \emptyset$ which means that $i$ is contained in the subsequence. 
\end{proof}

\noindent
Before we show the other direction, we need the following lemma.
\begin{lem}[\coqref{totally\_bounded\_id\_intersection}{/formalization/Hyperspace/MetricSubsets.v}{1840}]\label{lem:K-covering-intersection}
Let $K$ be totally bounded and complete.
Then $K = \bigcap_{n : \dN} K_n$ where $K_n$ denotes the $n$-th covering of $K$.
\end{lem}
\begin{proof}
Assume $K$ is totally bounded and complete and let $x \in K$.
Then by definition $x$ is contained in one of the balls of the covering $K_n$.
Now assume $x \in K_n$ for all $n : \dN$.
We can define a sequence $x_n$ with $x_n \in K$ and $\lim x_n = x$ since by definition there is a ball $B$ of radius $2^{-(n+1)}$ with $x \in B$ and $K \cap B \neq \emptyset$.
Thus, any $x_n \in B \cap K$ has distance at most $2^{-n}$ to $x$.
Since $K$ is complete, the limit of the sequence is contained in $K$, which shows that $x \in K$.
\end{proof}
\noindent
The next lemma shows that Bishop-compactness is equivalent to being compact-overt.
\begin{lem}[\coqref{bishop\_compact\_compact}{/formalization/Hyperspace/MetricSubsets.v}{3418}, \coqref{bishop\_compact\_overt}{/formalization/Hyperspace/MetricSubsets.v}{3539}]
Assume $K$ is totally bounded and complete, then $K$ is compact and overt.
\end{lem}
\begin{proof}
Note that for totally bounded $K$, $K = \emptyset$ is decidable and in the case that $K = \emptyset$ both the compact and overt information is trivial.
We thus can assume that $K \neq \emptyset$.

We first show that $K$ is compact.
It suffices to show that $K$ is classically compact, i.e.,
\[
\all{U : \dN \to \sopen \  X} K \subseteq \bigcup_{n : \dN} U_n \rightarrow \csome{k : \dN} K \subseteq \bigcup_{n = 0}^k U_n  
\]
and that for any cover $U$ and $k : \dN$, we can semi-decide whether $K \subseteq \bigcup_{n=0}^k U_n$, as this allows us to search through all finite covers in parallel until we find one that is large enough.
Further, using the fact that the space is second countable, it suffices to use open balls for the open sets $U_i$ in the coverings.

The first part is very similar to the standard textbook proof of the equivalence of totally bounded and complete and compactness.
We omit it here and refer to the Coq formalization (\coqref{bishop\_compact\_classical\_fincover\_compact}{/formalization/Hyperspace/MetricSubsets.v}{3359}).

For the second part, we further use the fact that a complete and totally bounded set  is classically sequentially compact, i.e.,
\[
\all{f : \dN \to K} \csome{g : \dN \to \dN} (\all{n : \dN} (g\  (n+1)) > (g\ n)\ ) \land \ \csome{y \in K} y = \lim f\circ g. 
\]
Again, the classical proof is similar to the textbook proof and omitted from the paper (\coqref{bishop\_compact\_classically\_seqcompact}{/formalization/Hyperspace/MetricSubsets.v}{3170}).

Now, assume that $B = \bigcup_{i= 0}^N \ball(c_i, n_i)$ is a finite union of balls.
We show that $K \subseteq B$ is semi-decidable. 
Since $K$ is totally bounded, we can find coverings $K_n = \bigcup_{i=0}^{N_i} K_{n,i}$ of $K$ with arbitrarily small balls.
The idea is that for large enough $n$, each $K_{n,i}$ will be completely contained in at least one of the balls.
To this end, for two balls $B_1 = \ball(x, n)$ and $B_2 = \ball(y, m)$, let us write $B_1 \ll B_2$ if $(\distx{x}{y}) < 2^{-m} - 2^{-n}$. 
Obviously, $B_1 \ll B_2$ implies that $B_1 \subseteq B_2$ and $B_1 \ll B_2$ is semi-decidable as it reduces to comparisons of reals.
Let us further write $L_n : [X]$ for the list of centers of the $n$-th covering.
According to \autoref{lem:tot-bounded-strong} we may assume that all $x \in L_n$ are contained in $K$.
Let \[
P\,k:=
\lall{x\in L_k}\,
\csome{0\le i\le N}\,
\ball(x,k)\ll\ball(c_i,n_i).
\]
For each $k : \dN$, $P\, k$ is again semi-decidable, as it is a finite conjunction of finite disjunctions of semi-decidable statements.
Thus $\csome{k : \dN}P\ k$ is also semi-decidable.

Thus, it remains to show that $K \subseteq B \liff \csome{k : \dN} P\, k$.
The implication $\csome{k : \dN} P\,k \to K \subseteq B$ is immediate, as
if $P\,k$ holds, then every ball $\ball(x,k)$ with $x \in L_k$ is contained in
one of the balls $\ball(c_i,n_i)$ and hence $K \subseteq \bigcup_{x \in L_k}\ball(x,k) \subseteq B$. 

For the converse implication, assume that $K \subseteq B$ and towards a contradiction assume that $\lall{k : \dN} \neg P(k)$, then unfolding the definitions and using $L_k \subseteq K$, we get
\[
\lall{k : \dN} \csome{x \in K}  \lall{0 \leq i \leq N} \distx{x}{c_i} \geq 2^{-{n_i}} - 2^{-k}.
\]
Thus, by classical countable choice, there (again classically) exists a sequence $(x_k)_{k:\dN} \subseteq K$ such that $\lall{k : \dN} \lall{0 \leq i \leq N} \distx{x_k}{c_i} \geq 2^{-{n_i}} - 2^{-k}$.
Using classical sequential compactness, there further is a convergent subsequence
$(x_{k_j})_{j : \dN} \subseteq K$ with $\lim_j x_{k_j} \in K$.
However, for all $0 \leq i \leq N$ and all $j : \dN$, $\distx{x_{k_j}}{c_i} \geq 2^{-n_i}-2^{-k_j}$, and thus the limit is contained in none of the balls $\ball(c_i,n_i)$, i.e.\ $\lim_j x_{k_j} \notin B$, contradicting the assumption that $K \subseteq B$.

Finally, we show that $K$ is also overt.
We first show that it is possible to check that $\ball(x,n) \cap K \neq \emptyset$ for any ball $\ball(x,n)$.
Since we have already shown that $K$ is compact, $\ball(x,n) \cap K \neq \emptyset$ iff $d_K(x) < 2^{-n}$.
Thus, we can use the fact that $K$ is located by \autoref{lem:bc-located} and that inequality on real numbers is semi-decidable.

Now let $U :\subseteq X$ be an arbitrary open set. 
By \autoref{thm:polish-second-countable}, we can exhaust $U$ with base elements, i.e., $U = \bigcup_{i : \dN} U_i$ where $U_i$ are open balls.
For all $i : \dN$, let $s_i : \dS$ be defined by $\defineds{s_i} \liff U_i \cap K \neq \emptyset$.
Then $(U \cap K \neq \emptyset)$  iff $\exists (i:N). s_i \downarrow$.
\end{proof}

\noindent
Combining the above lemmas we get:
\begin{thm}
A subset $K :\subseteq X$ of a Polish space is compact and overt iff it is totally bounded and complete.
\end{thm}

\noindent
The totally bounded information used to produce coverings yields a representation that is easy to manipulate computationally and efficient if our goal is to draw the whole set or do other operations that require global knowledge.
However, naturally the size of the list of coverings grows exponentially when increasing the precision and one might argue that in practice high precision approximations are more useful locally, i.e.\ instead of getting a picture of the whole set one might want to zoom into a small part of the set and draw only this sector with high precision.

Locatedness comes in handy here, as instead of outputting the whole set, the list can be replaced by a nondeterministic test function, that tests whether a ball is close enough to the set or not:
\[
\mathrm{M\_test}\ A \equivto \all{x : X} \all{n : \dN} \mval\ \big( (d_A(x) < 2 \cdot 2^{-n}) + (d_A(x) > 2^{-n})\big).  
\]
Recall that here $P+Q$ denotes the sum type in $\Type$, not a proposition:
an element of $P+Q$ gives constructive evidence of which of $P$ or $Q$ holds.
Thus, the operation is a nondeterministic near/far test that distinguishes points that are definitely far from the set from those that are sufficiently close, and is allowed to be inconclusive near the boundary.
Similar representations have been studied in complexity theory (e.g. \cite{braverman2005complexity,braverman2009computability}) as they give a reasonable notion of polynomial time complexity for compact subsets.
It is possible to replace the totally bounded information with the $\mathrm{M\_test}$ information to get a similar representation for compact-overt sets.
This might allow the definition of more efficient algorithms in some cases, but we have not pursued this approach in detail thus far.

\subsection{Operations on Compact-Overt Subsets}
For two nonempty sets, the classical Hausdorff distance is defined by
 \[
       d_{\mathrm H}(S,T) = \max\left\{\,\sup_{x \in S} d(x,T),\, \sup_{y \in T} d(S,y) \,\right\}.
\]
For any metric space the Hausdorff distance defines a metric on the space of nonempty compact subsets of the metric space.

To formally define the Hausdorff distance, we define the \inlinedef{fattening} $U_\varepsilon$ of a subset $U :\subseteq X$ by 
\[
U_\varepsilon \equivto \lam{x : X} \csome{y \in U} (\distx{x}{y}) \leq \varepsilon.
\]
We can then define a classical proposition on classical subsets $S, T :\subseteq X$ and a real number $r : \dR$ by
\[
\mathrm{dist_H\_of}\ S\ T\ r\equivto \mathrm{is\_inf}\  (\lam{\varepsilon : \dR} \varepsilon > 0 \land S \subseteq T_\varepsilon \land T \subseteq S_\varepsilon)\ r
\]
for $r$ being the Hausdorff distance between $S$ and $T$ 
where $\mathrm{is\_inf}$ is also defined classically.
Given the compact-overt information of a set, we can make this constructive.
In fact, it suffices that the sets are totally bounded.
\begin{lem}[\coqref{Hausdorff\_dist\_exists}{/formalization/Hyperspace/MetricSubsets.v}{4324}]
Let $K_1, K_2 :\subseteq X$ be totally bounded and nonempty.
Then, $ 
\some{r : \dR} (\mathrm{dist_H\_of}\ K_1\ K_2\ r)$.
\end{lem}

\begin{proof}
By applying the limit operator, it suffices to show that for each $n : \dN$, we can find a $2^{-n}$ approximation of $r$.
For any totally bounded set $K$, the centers of the balls in the $n$-th covering of $K$ have distance less than $2^{-n}$ to some point in the set.
Thus, for the set of finite points $K^{(n)}$ defined by the centers of the $n$-th covering it holds $d_H\;K\;K^{(n)} \leq 2^{-n}$.
By the triangle inequality it suffices that we can compute the Hausdorff distance for (nonempty) finite point sets.

We show this for the one-sided Hausdorff distance $d_{\vec{H}}(S,T)\ \equivto \mathrm{inf}_{\varepsilon \geq 0} \ S \subseteq T_\varepsilon$ from which the claim follows by taking the maximum of $d_{\vec{H}}(S,T)$ and $d_{\vec{H}}(T,S)$.
The proof is by induction using the easy to verify facts that
$d_{\vec{H}}(A  \cup \{x\}, B) = \max(d_{\vec{H}}(A,B), (d_B\;x))$ and
$d_{\vec{H}}(\{x\}, B \cup \{y\}) = \min (d_{\vec{H}}(\{x\},B), (\distx{x}{y}))$\\
(\coqref{Hausdorff\_dist\_os\_add\_pt}{/formalization/Hyperspace/MetricSubsets.v}{3736} and \coqref{Hausdorff\_dist\_os\_pt\_extend}{/formalization/Hyperspace/MetricSubsets.v}{3821}).\qedhere
\end{proof}

\noindent
The next lemma shows that we can construct the limits of compact-overt sets.

\begin{lem}[\coqref{hausdorff\_fast\_cauchy\_limit\_exists}{/formalization/Hyperspace/MetricSubsets.v}{4633}, \coqref{totally\_bounded\_lim}{/formalization/Hyperspace/MetricSubsets.v}{4420}]
For any compact-overt sets $K, K'$ and  $x: \dR$, define $(d_{\mathrm H}\ K\ K' \leq x) :\Prop$ classically using $\mathrm{dist_H\_of}$.
Let $(K_i)_{i : \dN} :\subseteq X$ be a sequence of nonempty compact-overt sets such that for all $n,m : \dN$, $d_{\mathrm H}\;K_{n}\;K_{m} \leq 2^{-n}+2^{-m}$ holds.
Then there is a unique compact-overt set $K :\subseteq X$ such that
$\lall{n : \dN} d_{\mathrm H}\;K_{n}\;K \leq 2^{-n}$. 
\end{lem}
\begin{proof}
Each $K_i$ is totally bounded and we denote  its $j$-th covering by $K_{i,j}$.
We define a totally bounded set $K$ by setting the $n$-th covering of $K$ to be $K_{n+1, n+1}$ and define the limit as the unique complete set defined by $\bigcap_{n : \dN} K_n$.
\end{proof}

\begin{rem}
    Putting the results together, we can consider the space of (classically) compact subsets of a Polish space $X$ with the compact-overt representation as a complete metric space with respect to the Hausdorff distance.
    Further, as we can approximate any totally bounded $K :\subseteq X$ arbitrarily well by finite sets, the set of nonempty finite sets forms a countable dense subset.
    Thus, we can consider the space of classically compact subsets itself as a Polish space.
\end{rem}

\noindent
While the previous sections already show that we can compute many operations on compact-overt subsets, the algorithms we get from the proofs are extremely inefficient and thus not very useful.
On the other hand, using the totally bounded information, many operations can be directly applied on the coverings, yielding more efficient algorithms.
We therefore specialize the theorems for many of the standard operations.
The extracted programs can be used as a small calculus to combine subsets using set operations.

For any set $A :\subseteq X$ and function $f: X \to Y$ let us define the \inlinedef{image} $f[A] :\subseteq Y$ by  \[y \in f[A] :\liff \csome{x \in A} f(x) = y.\]
We further say a function $f: X \to Y$ is \inlinedef{uniformly continuous} on some set $A$ if \[\all{n : \dN} \some{m : \dN} \lall{x\ y \in A} \distx{x}{y} < 2^{-m} \rightarrow \mathrm{d}_Y \; (f\;x)\; (f\;y) < 2^{-n}.\]
Using the continuity principle it is possible to show that every function is uniformly continuous on compact-overt sets.
However, in many cases it is more practical to provide the uniform modulus of continuity manually and thus avoid the inefficient computation.
We therefore only show the following lemma.
\begin{lem}[\coqref{image\_totally\_bounded}{/formalization/Hyperspace/MetricSubsets.v}{4683}]
Let $f: X \to Y$ be uniformly continuous, then the image of compact-overt sets is again compact-overt.
\end{lem}
\begin{proof}
Let $A$ be a compact-overt set. Using the uniform continuity we can find an $m : \dN$ such that the image of the $m$-th covering of $A$ is an $n$-th covering of $f[A]$.
\end{proof}

\noindent
The properties stated below can mostly be shown by directly applying the operations on the coverings.
In the formal development, we currently only prove them for Euclidean space, but the proofs can be easily generalized to any space where the operations are defined.
\begin{lem}[\coqref{tbounded\_union}{/formalization/Hyperspace/EuclideanSubsets.v}{1004},\coqref{tbounded\_scale\_up}{/formalization/Hyperspace/EuclideanSubsets.v}{1129},\coqref{tbounded\_scale\_down}{/formalization/Hyperspace/EuclideanSubsets.v}{1081}, \coqref{tbounded\_translation}{/formalization/Hyperspace/EuclideanSubsets.v}{1025}]
Let $X$ be a Polish space with addition $+ : X \to X \to X$ and scalar multiplication $\cdot : \dR \to X \to X$ defined and compatible with the metric.
Then the following hold.
\begin{enumerate}
    \item For compact-overt sets $A, B :\subseteq X$ their union $A \cup B$ is compact-overt.
    \item For a compact-overt set $A :\subseteq X$ and any $c > 0$, the scaled set $c A \equivto \lam{x:X}
    \csome{y \in A}x = c \cdot y$ is compact-overt.
    \item For a compact-overt set $A :\subseteq X$ and any $c : X$, the set $A + c \equivto 
    \lam{x :X} \csome{y \in A} x = y + c$
    is compact-overt.
\end{enumerate}
\end{lem}
\noindent
As mentioned above, instead of working on coverings, another efficient way to compute with sets is to work on the $\mathrm{M\_test}$ representation.
The above operations can equivalently be defined using the information given by the $\mathrm{M\_test}$ operator, but we omit the details.

\section{Realizing Axioms in Code Extraction}
\label{s:extraction}
To extract  programs from proofs in the Coq formalization, we extended Coq's program extraction mechanism to handle the newly added axioms, including the continuity principle, providing a certified module in \softwarename{AERN} for defining and manipulating higher-order objects. 
Our choice of Haskell as the extraction language is mostly pragmatic, as the \softwarename{AERN} library, which provides the exact-real arithmetic used by the extracted programs, had already been developed in Haskell. At the same time, Haskell’s built-in support for lazy evaluation makes it particularly well suited for implementing exact-real computation and the constructions used in our development. Nevertheless, there should be no fundamental obstacle to defining a similar extraction mechanism for another language such as OCaml, provided that one also builds suitable support for basic exact-real computation.

The types $\dK$ and $\dR$ are extracted to  \softwarename{AERN} types $\ckleenean$ and $\creal$, which encapsulate convergent sequences of lazy Booleans and dyadic intervals, respectively.
Thus, each value in $\dK$ or $\dR$ corresponds to multiple values (realizers) in $\ckleenean$ or $\creal$, respectively.
The nondeterminism monad $\mval$ becomes the identity monad in Haskell, as the nondeterminism permitted in $\mval$ 
disappears at the level of $\ckleenean$ and $\creal$ values (realizers).
As discussed before \autoref{ax:baire-choice}, the existence of multiple realizers for values of types such as $\dK$ or $\dR$ is the only source of nondeterminism permitted in $\mval$.

\begin{itemize}
    \item \autoref{ax:disj} (Countable disjunction) extracts to 
    the \softwarename{AERN} library's built-in function \texttt{or\_countable}, which dovetails the evaluation of the infinite sequence of Kleeneans, each entry of which is itself an infinite sequence of lazy Booleans.

\item \autoref{ax:baire-choice} (Baire choice) extracts to
the Haskell identity function since, after removing the nondeterminism monad $\mval$, the axiom has essentially no computational content.

\item \autoref{ax:cchoice} (Nondeterministic countable selection) extracts to
the AERN function \\\texttt{selectCountable}
which searches for \texttt{True} within an infinite sequence of infinite sequences and, if it finds one, returns the number of the sequence in which it found it.
It traverses the sequences in such an order that each of the elements will eventually be checked.

\item \autoref{ax:k-to-seq} extracts to  the expression 
\begin{lstlisting}[language=Haskell]
\k n -> if isCertainlyTrue (k?(bits n)) then 1 else 0
\end{lstlisting}
translating the given convergent Kleenean sequence \texttt{k} into a function sequence using the \texttt{?} operator that evaluates the sequence with a certain effort, which determines the precision of any numerical computation involved.

\item \autoref{ax:continuity} extracts to
the AERN function \texttt{maxIntParamUsed} which takes the two parameters $f$ and $x$ and applies $f(x)$ while spying on how the function $f$ uses the sequence $x$, keeping track of the largest accessed index.
If $f(x)$ terminates and returns True, the function returns the value of the largest parameter used.
The spying is achieved by instrumenting the sequence $x$ with imperative code that updates a mutable variable.
The imperative code is hidden inside the pure code using \texttt{unsafePerformIO}.

\end{itemize}

\section{Examples in Euclidean Space}\label{sec:examples}
Let us consider a few larger examples for the concrete space $\dR^n$ and the programs that we can extract from the proofs.
We use the subset of dyadic rational numbers as the countable dense subset.  
The dyadic numbers are defined as pairs $\dD \equivto \dZ \times \dN$ and we identify the pair $(z,n)$ with the real $z \cdot 2^{-n}$.
Similarly, we write $\dD^m$ for the type of length-$m$ lists over $\dD$ when $m : \dN$, and
suppose that the trivial coercion from $\dD^m$ to $\dR^m$ is taken implicitly.

We use dyadic numbers to approximate real numbers.

\begin{lem}
For any $m : \dN$, the space $\dR^m$ with the distance induced by the maximum norm $\| \cdot \|$, the limit operator from \autoref{sec:preliminaries} and the dyadic rational numbers as countable dense subset is a Polish space.
\end{lem}
\begin{proof}
 In \cite{konevcny2022extracting} it is already shown that for any real number we can nondeterministically find a dyadic rational approximation up to any prescribed error, i.e., the following statement holds in the Euclidean space:
\[
\all{x : \dR^m} \all{p : \dN} \mval \some{d : \dD^m} \dist{x}{d} \leq 2^{-p}.
\]
It follows that $\dD^m$ is dense in $\dR^m$.
The other properties of a Polish space are easy to show.
\end{proof}
\noindent
From a proof that a set is compact-overt in our Coq implementation, we can extract an \softwarename{AERN} program that computes the coverings.
We implemented the examples below in Coq and extracted programs to generate for each $n$ a finite list of real numbers, encoding the center and radius of each ball in the $n$-th covering. 
We can then use \softwarename{AERN} to give us arbitrarily exact rational approximations of these real numbers, which we in turn can output and visualize by drawing the approximate boxes.
Let us start with a simple example.
We define a triangle $T \subseteq \IR^2$ by
\[
  (x,y) \in T \liff x \geq 0 \land y \geq 0 \land x + y \leq 1.
\]
To show that the set is totally bounded, let
$b_{i,j,n} \equivto \ball((\frac{2i+1}{2^{n+1}}, \frac{2j+1}{2^{n+1}}), n)$.
We define the $n$-th covering as the list $L$ containing all $b_{i,j,n}$  with $i+j < 2^n$.
\autoref{fig:triangle} shows the coverings defined by this procedure.
All of the balls have radius $2^{-n}$ by definition and each ball intersects the triangle as $i+j < 2^n$ implies
$\frac{2i+1 + 2j + 1}{2^{n+1}} \leq 1$.
Each point of the triangle is contained in one of the $b_{i,j,n}$ as for any $x \in [0,1]^2$ and $n \in \IN$ we can find some $k \in \IN$ such that $ 2^{-n} k \leq x  < 2^{-n}(k+1)$.
We can thus find such coordinates $(i,j)$ for a point $(x,y) \in T$ and as $x+y \leq 1$ it follows that $i+j < 2^{n}$ as claimed.
\begin{figure}
\centering
\begin{subfigure}{.23\textwidth}
\includegraphics[width=0.9\textwidth]{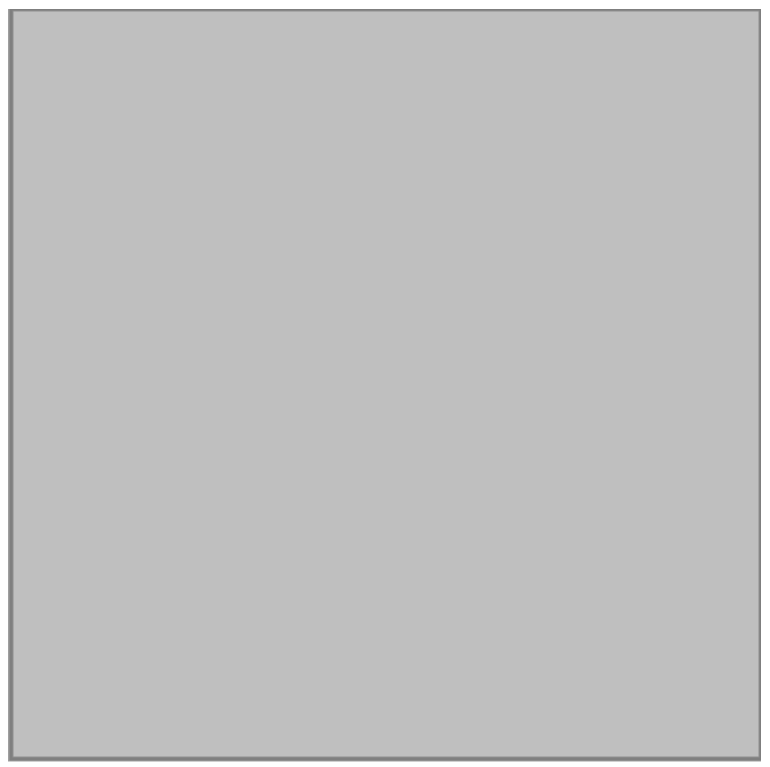}
\end{subfigure}
\begin{subfigure}{.24\textwidth}
\includegraphics[width=0.9\textwidth]{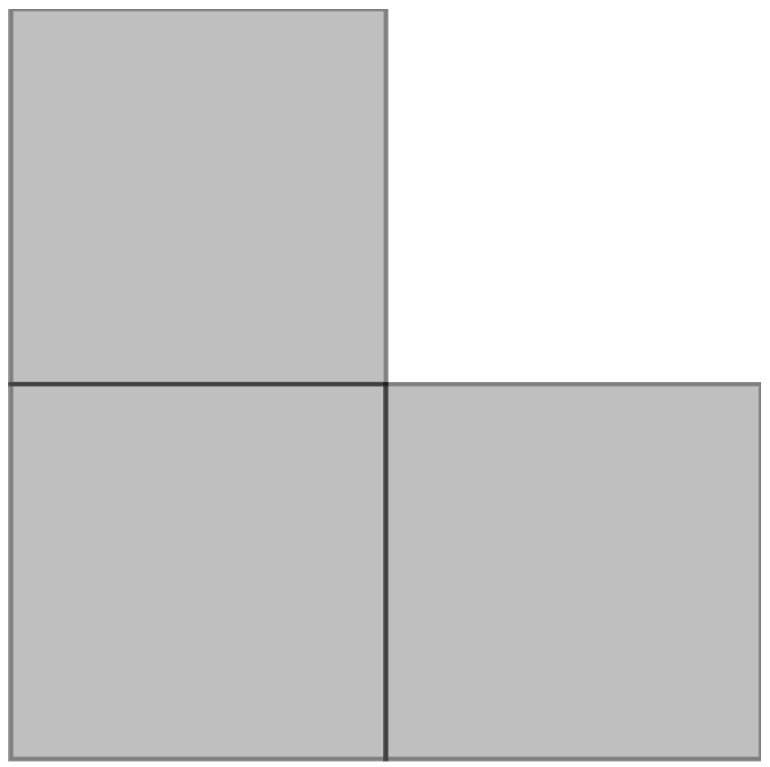}
\end{subfigure}
\begin{subfigure}{.24\textwidth}
\includegraphics[width=0.9\textwidth]{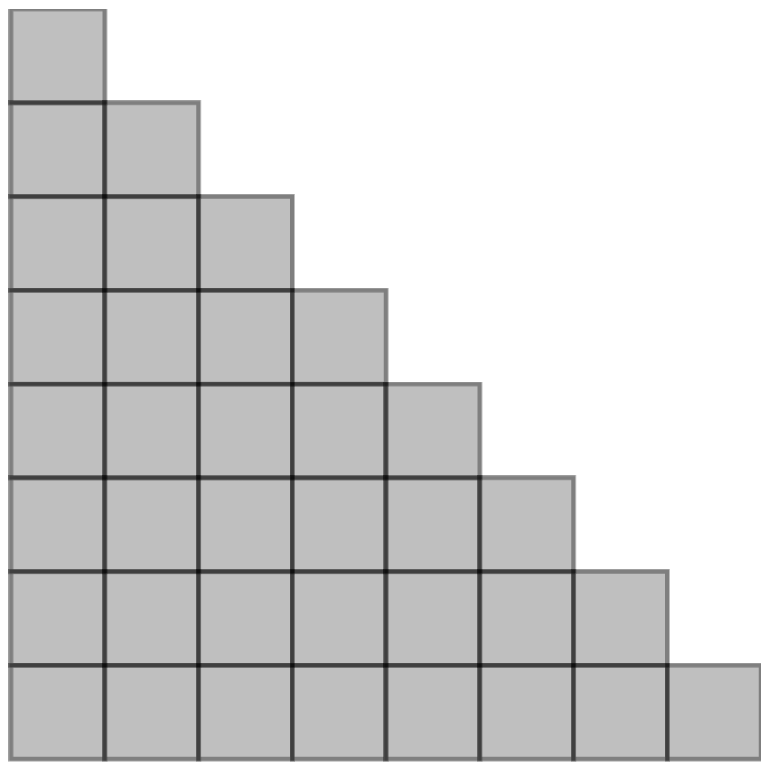}
\end{subfigure}
\begin{subfigure}{.24\textwidth}
\includegraphics[width=0.9\textwidth]{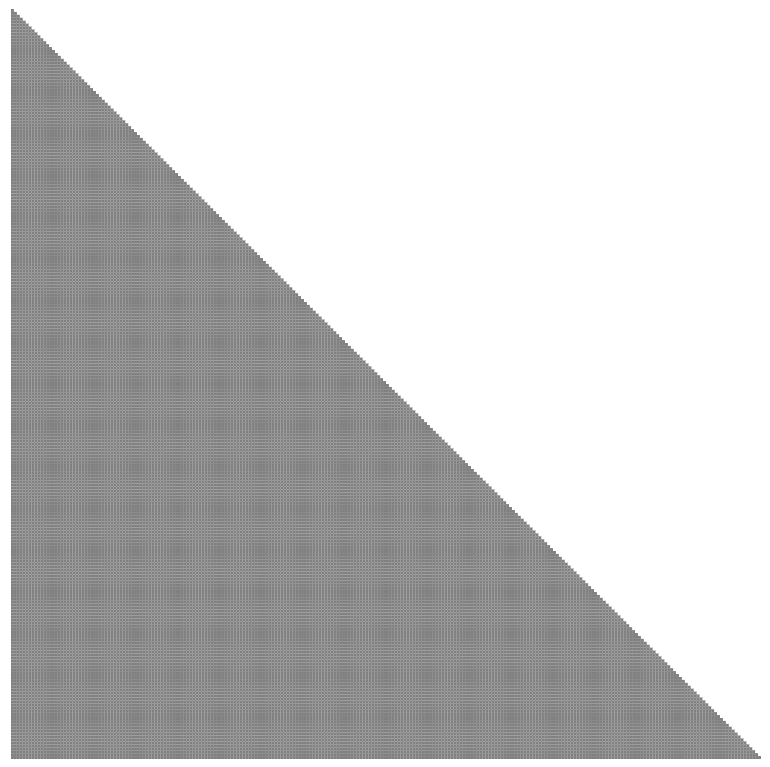}
\end{subfigure}
\caption{Approximations of the triangle}
\label{fig:triangle}
\end{figure}

\subsection{Drawing Fractals}
As a more interesting application, let us look at a procedure to generate certified drawings of fractals.
We consider simple self-similar fractals generated by iterated
function systems (IFS) without rotation. 
For simplicity we only consider fractals that are contained in the unit cube $[-1,1]^m$.
For a finite set of points $D \equivto (d_1, d_2, \dots, d_k) \subseteq [-1,1]^m$ we define a classical subset $\mathcal{I}(D) \subseteq \dR^m$ as the unique compact subset such that

\begin{enumerate}
\item $D \subseteq \mathcal{I}(D)$, and
\item  $\mathcal{I}(D) = \bigcup_{d \in D} \{ \frac{x+d}{2} \,|\, x \in  \mathcal{I}(D)\}$
\end{enumerate}
We show that any set defined in this way is totally bounded.
To do so, we need to define a sequence of coverings of the set.
We start with the first covering $L_0$ being the list containing only the unit cube.
As each $d \in D$ is contained in the unit disc, the intersection and covering properties hold.
We then recursively define the covering $L_{n+1}$ from the previous covering $L_n$ by following the construction rule, i.e.,  by making for each $d \in D$ a copy of the previous covering where the center of each ball is moved halfway towards $d$ and has half the radius.

Obviously, each of the balls in $L_{n+1}$ has radius $2^{-(n+1)}$. 
To show that this procedure preserves the intersection property,
let $b \in L_{n+1}$. 
Then there is some $b' \in L_n$ and $d \in D$ such that $x \in b$ iff there is an $x' \in b'$ with $x = \frac{x'+d}{2}$.
As $L_n$ has the intersection property, there is some $x' \in b'$ with $x' \in \mathcal{I}(D)$.
By definition of $\mathcal{I}(D)$, then $\frac{x'+d}{2} \in \mathcal{I}(D)$, i.e., $b \cap \mathcal{I}(D) \neq \emptyset$.
Similarly, to show that the covering property is preserved, 
assume $x \in \mathcal{I}(D)$.
Then by definition either $x = d$ for some $d \in D$ or there is some $x' \in \mathcal{I}(D)$ and some $d \in D$ such that $x = \frac{x'+d}{2}$. 
Any $d \in D$ is contained in $L_{n+1}$ as $d \in L_{n}$ and $d = \frac{d+d}{2}$.
In the other case, there is some $b \in L_n$ such that $x' \in b$ as $L_n$ is a covering.
But then by definition of the procedure, there is some $b' \in L_{n+1}$ that contains $x$.

A more elegant way to prove that $\mathcal{I}(D)$ is totally bounded is to use the limit operation.
We can define a sequence of compact-overt sets $(T_i)_{i \in \IN}$ by letting $T_0$ be the unit disc and then applying the iteration given by the IFS.
It is not hard to show that each $T_i$ has Hausdorff distance at most $2^{-i}$ to the fractal and thus applying the limit operation suffices. 
However the extracted program is slightly less efficient than the direct encoding as the limit uses $T_{n+1}$ to get the $n$-th approximations and increases the radius leading to coverings where balls overlap.

\subsection{Examples and Evaluation}
In the Coq formalization, we have proved total boundedness of the
following sets:
\begin{enumerate}
\item the right triangle $T$ with vertices $(0,0)$, $(1,0)$, $(0,1)$
  (\coqref{T\_tbounded}{/formalization/Hyperspace/Simpletriangle.v}{413});
\item the Sierpinski triangle with $\mathcal{D} = \{(0,0), (0,1), (1,0)\}$
  (\coqref{STR\_tbounded}{/formalization/Hyperspace/SierpinskiTriangle.v}{766});
\item the equilateral Sierpinski triangle with
  \[\mathcal{D} = \{(-1,-1), (1,-1), (0, \sqrt{3}-1)\}\]
  (\coqref{STE\_tbounded}{/formalization/Hyperspace/SierpinskiTriangle.v}{1039});
\item the equilateral Sierpinski triangle with an additional central map,

  \[\mathcal{D} = \{(-1,-1), (1,-1), (0, \sqrt{3}-1), (0, \sqrt{3}/3-1)\}\]
  (\coqref{STE4\_tbounded}{/formalization/Hyperspace/SierpinskiTriangle.v}{1101}).
\end{enumerate}
All four proofs follow the construction outlined above. For the Sierpinski
triangle of Item~(2) we additionally provide a second proof, in which total
boundedness is established indirectly through the limit operation
(\coqref{tbounded\_sierpinski}{/formalization/Hyperspace/SierpinskiTriangleLimit.v}{602}). 

The file \texttt{Extract.v} in the repository contains the extraction
directives for the listed lemmas; the resulting Haskell files, together with
the post-processing described in the repository's \texttt{README}, can be found
in the \texttt{src/} directory. 

We evaluated the extracted programs on a MacBook Pro 16-inch (2021; Apple M1
Pro, 16\,GB). Each extracted function takes two parameters $n$ and $p$ and
returns a list of two-dimensional balls, each with a dyadic radius and a center
given as a $2^{-p}$-approximation of a real number, whose union approximates
the set within Hausdorff distance $2^{-n}$. For each function we increased $n$
in steps of one, recording the number of balls and the runtime in seconds.
Each runtime measurement was repeated three times, and the reported value is the average of the three runs. 
We performed these measurements at two center precisions, $p = 53$ (double precision) and $p = 106$ (double-double precision), in order to measure the secondary effect of $p$. 
We increased $n$ until a single run exceeded $300$ seconds.

The results are plotted in \autoref{fig:result}; \autoref{fig:sierpinski-triangle} shows the shapes drawn by the extracted programs. 
Both the evaluator and the raw data are available in the repository, so the results can be reproduced. 
As $n$ increases the number of balls in the covering grows exponentially, and so does the runtime. 
Furthermore, as \autoref{fig:result}(b) shows, the runtime is essentially linear in the number of balls.
Moreover, most of the runtime is spent constructing the balls, so the precision used to compute their centers has little effect on the overall runtime (\autoref{fig:result}(c)).

\definecolor{cTn}{RGB}{31,119,180}
\definecolor{cSTR}{RGB}{255,127,14}
\definecolor{cSTE}{RGB}{44,160,44}
\definecolor{cSTE4}{RGB}{214,39,40}
\definecolor{cLim}{RGB}{148,103,189}

\pgfplotsset{
  styTn/.style  ={cTn,   mark=*},
  stySTR/.style ={cSTR,  mark=square*},
  stySTE/.style ={cSTE,  mark=triangle*},
  stySTE4/.style={cSTE4, mark=diamond*},
  styLim/.style ={cLim,  mark=pentagon*},
}

\pgfplotsset{
  discard if not/.style 2 args={
    x filter/.append code={
      \edef\tempa{\thisrow{#1}}%
      \edef\tempb{#2}%
      \ifx\tempa\tempb\else\def\pgfmathresult{nan}\fi
    }
  }
}

\begin{figure}
\centering
\begin{tikzpicture}
\begin{groupplot}[
    group style={group size=2 by 2, horizontal sep=3cm, vertical sep=2cm},
    width=0.4\textwidth, height=0.4\textwidth,
    every axis plot/.append style={mark size=1.5pt, line width=0.7pt},
    grid=both, grid style={gray!25, very thin},
    tick label style={font=\scriptsize},
    label style={font=\small}, title style={font=\small},
  ]

  \nextgroupplot[
      ymode=log, xmin=0.5, xmax=15.5, xtick={1,2,4,...,14},
      xlabel={precision $n$}, ylabel={runtime (s)},
      title={(a)},
      legend cell align=left, legend columns=3, legend to name=timinglegend,
      legend style={font=\scriptsize, draw=none,
                    /tikz/every even column/.append style={column sep=5pt}}]
    \addplot[styTn]   table[col sep=comma, x=n, y=secs,
        discard if not={example}{Tn},     discard if not={bits}{53}]{timing_plot.csv};
  \addlegendentry{triangle $T$}
    \addplot[stySTR]  table[col sep=comma, x=n, y=secs,
        discard if not={example}{STRn},   discard if not={bits}{53}]{timing_plot.csv};
      \addlegendentry{Sierpinski (right)}
    \addplot[stySTE]  table[col sep=comma, x=n, y=secs,
        discard if not={example}{STEn},   discard if not={bits}{53}]{timing_plot.csv};
      \addlegendentry{Sierpinski (equil.)}
    \addplot[stySTE4] table[col sep=comma, x=n, y=secs,
        discard if not={example}{STE4n},  discard if not={bits}{53}]{timing_plot.csv};
      \addlegendentry{Sierpinski (equil., 4-map)}
    \addplot[styLim]  table[col sep=comma, x=n, y=secs,
        discard if not={example}{STRLim}, discard if not={bits}{53}]{timing_plot.csv};
      \addlegendentry{Sierpinski (right, limit)}

  \nextgroupplot[
      xmode=log, ymode=log,
      xlabel={number of balls}, ylabel={runtime (s)},
      title={(b)}]
    \addplot[styTn]   table[col sep=comma, x=balls, y=secs,
        discard if not={example}{Tn},     discard if not={bits}{53}]{timing_plot.csv};
    \addplot[stySTR]  table[col sep=comma, x=balls, y=secs,
        discard if not={example}{STRn},   discard if not={bits}{53}]{timing_plot.csv};
    \addplot[stySTE]  table[col sep=comma, x=balls, y=secs,
        discard if not={example}{STEn},   discard if not={bits}{53}]{timing_plot.csv};
    \addplot[stySTE4] table[col sep=comma, x=balls, y=secs,
        discard if not={example}{STE4n},  discard if not={bits}{53}]{timing_plot.csv};
    \addplot[styLim]  table[col sep=comma, x=balls, y=secs,
        discard if not={example}{STRLim}, discard if not={bits}{53}]{timing_plot.csv};
    \addplot[dashed, gray, no marks] coordinates {(5, 9.055e-5) (1.5e7, 271.6)};

  \nextgroupplot[
      xmin=0.5, xmax=15.5, ymin=0.5, ymax=1.50, xtick={1,2,4,...,14},
      xlabel={precision $n$}, ylabel={ratio},
      title={(c)}]
    \addplot[dashed, black, no marks] coordinates {(5.5,1) (15.5,1)};
    \addplot[styTn]   table[col sep=comma, x=n, y=ratio,
        discard if not={example}{Tn}]{ratios.csv};
    \addplot[stySTR]  table[col sep=comma, x=n, y=ratio,
        discard if not={example}{STRn}]{ratios.csv};
    \addplot[stySTE]  table[col sep=comma, x=n, y=ratio,
        discard if not={example}{STEn}]{ratios.csv};
    \addplot[stySTE4] table[col sep=comma, x=n, y=ratio,
        discard if not={example}{STE4n}]{ratios.csv};
    \addplot[styLim]  table[col sep=comma, x=n, y=ratio,
        discard if not={example}{STRLim}]{ratios.csv};

\end{groupplot}
\end{tikzpicture}\\[2pt]
\ref{timinglegend}
\caption{Runtime of the extracted programs, averaged over three runs. \textbf{(a)}~wall-clock time against the precision parameter $n$, with real centers computed to double precision. \textbf{(b)}~the same data against the number of balls in the covering. \textbf{(c)}~ratio $t_{106}/t_{53}$ of runtime with double-double centers $t_{106}$ to double centers $t_{53}$.}
\label{fig:result}
\end{figure}

\begin{figure}
\centering
\begin{subfigure}{\textwidth}
  \centering
  \includegraphics[width=0.23\textwidth]{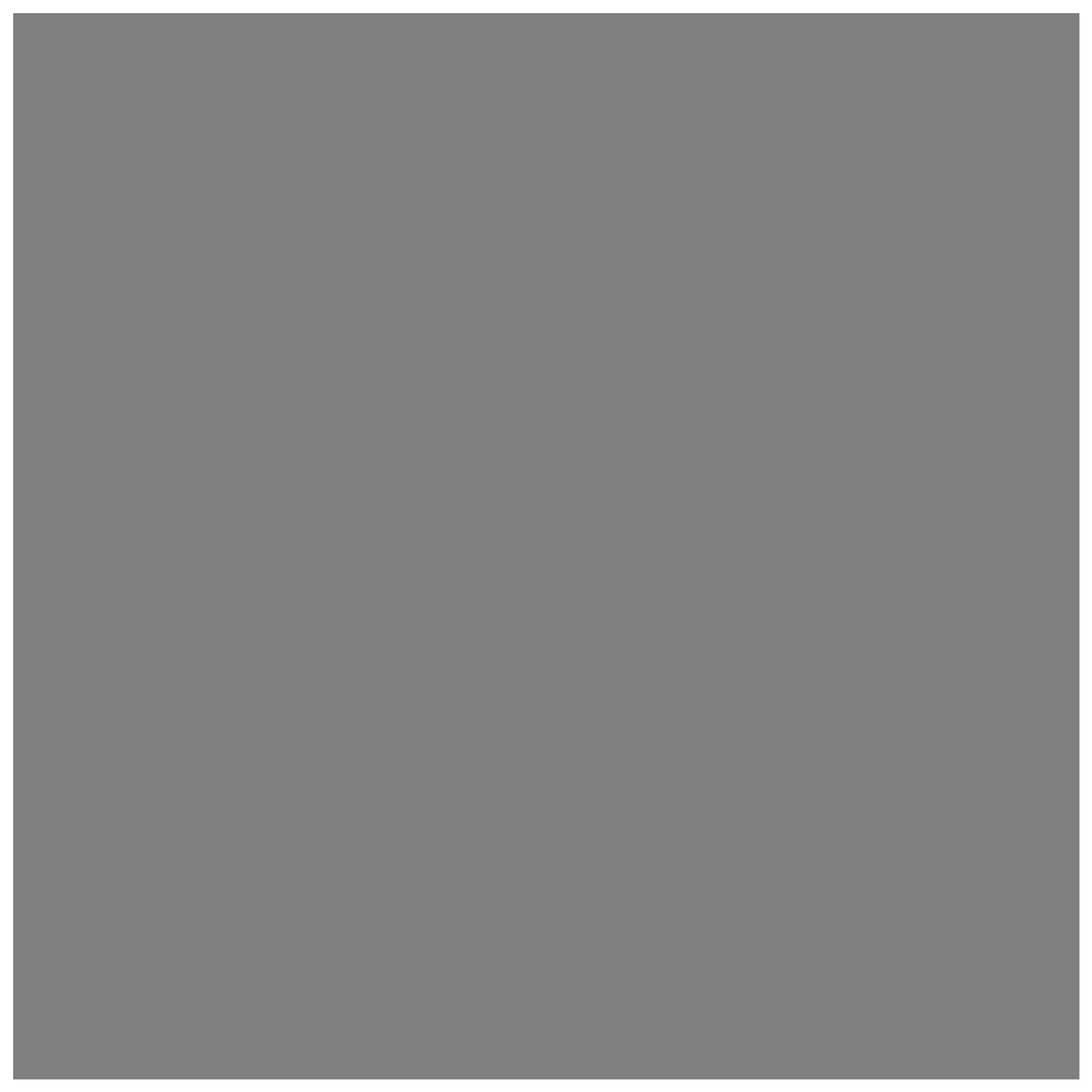}\hfill
  \includegraphics[width=0.23\textwidth]{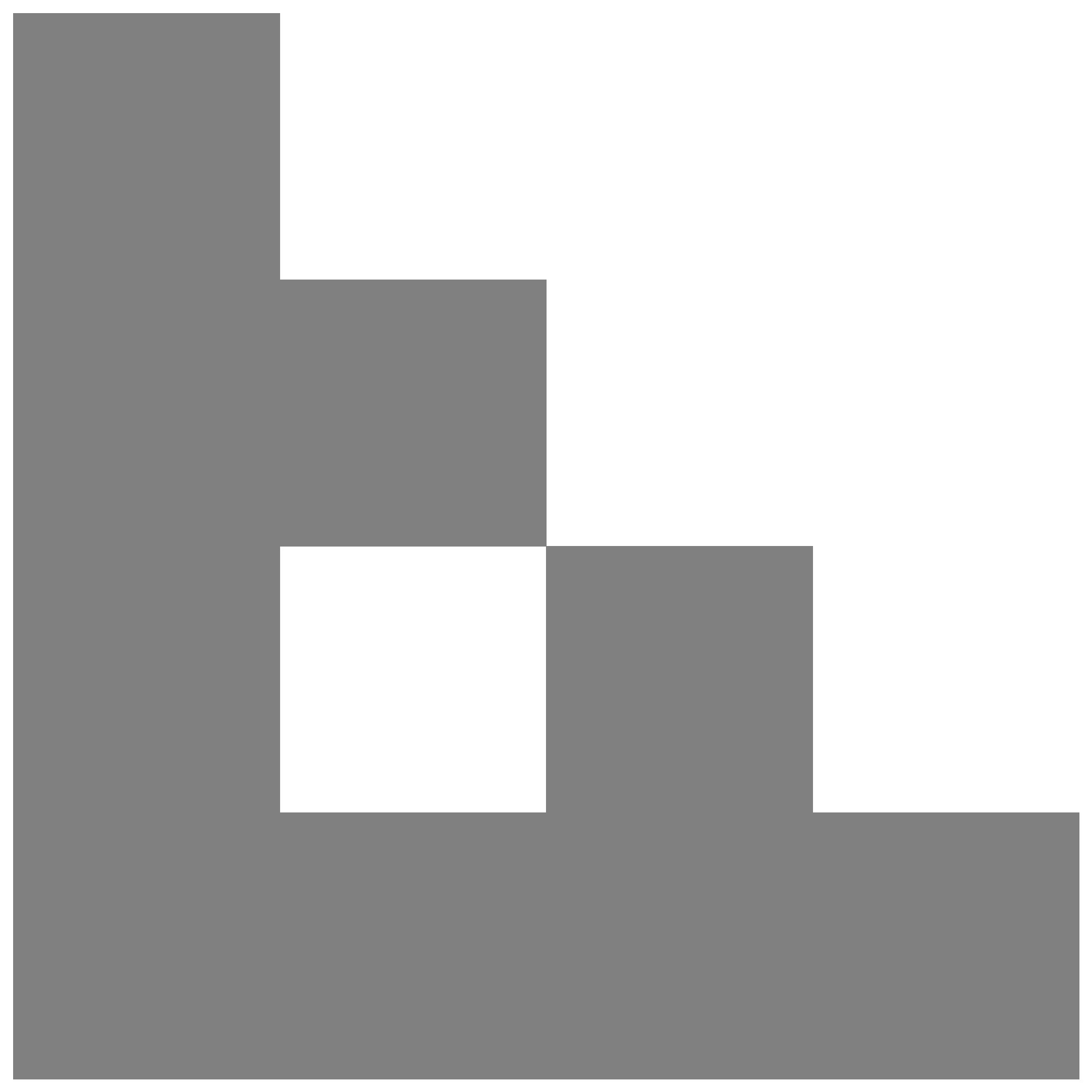}\hfill
  \includegraphics[width=0.23\textwidth]{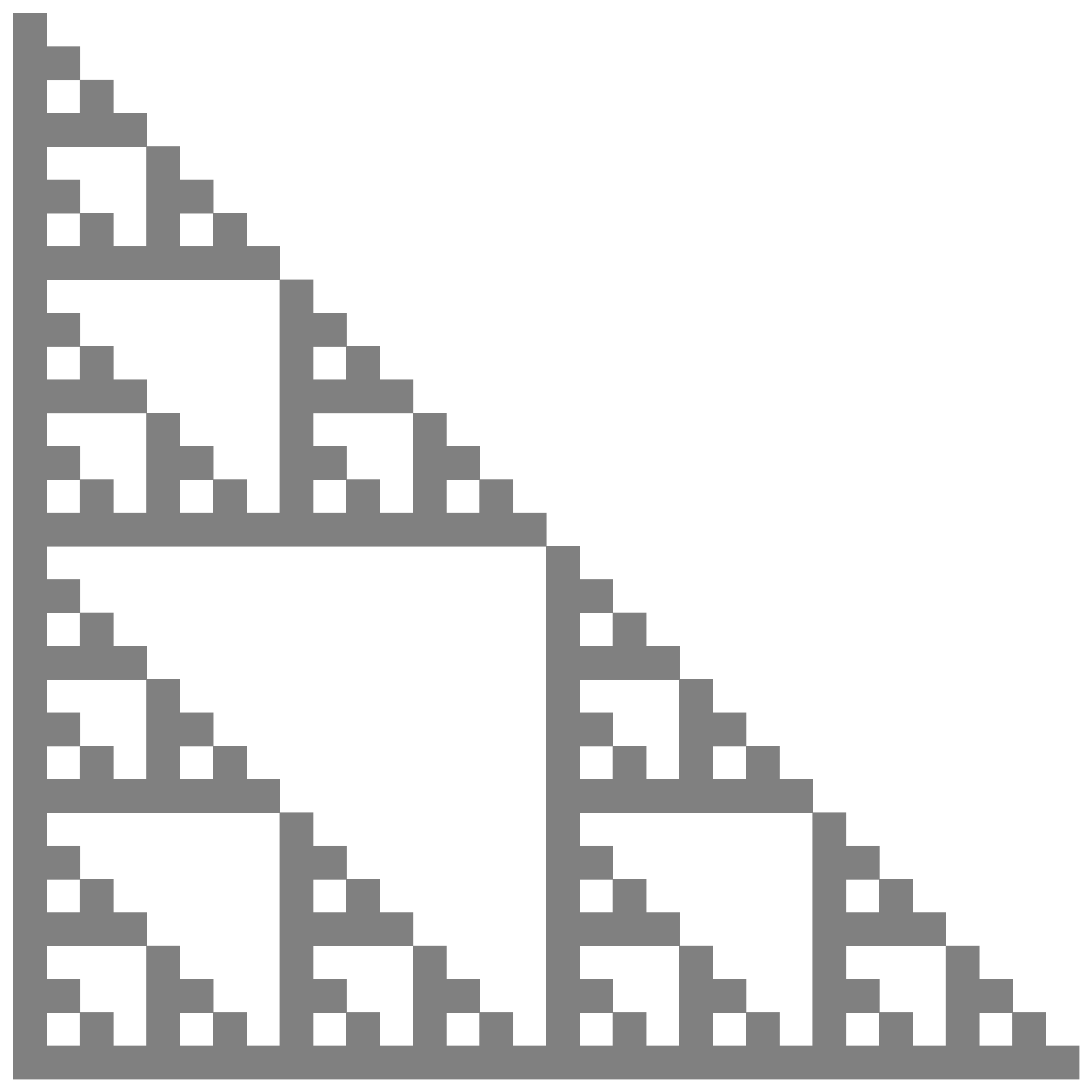}\hfill
  \includegraphics[width=0.23\textwidth]{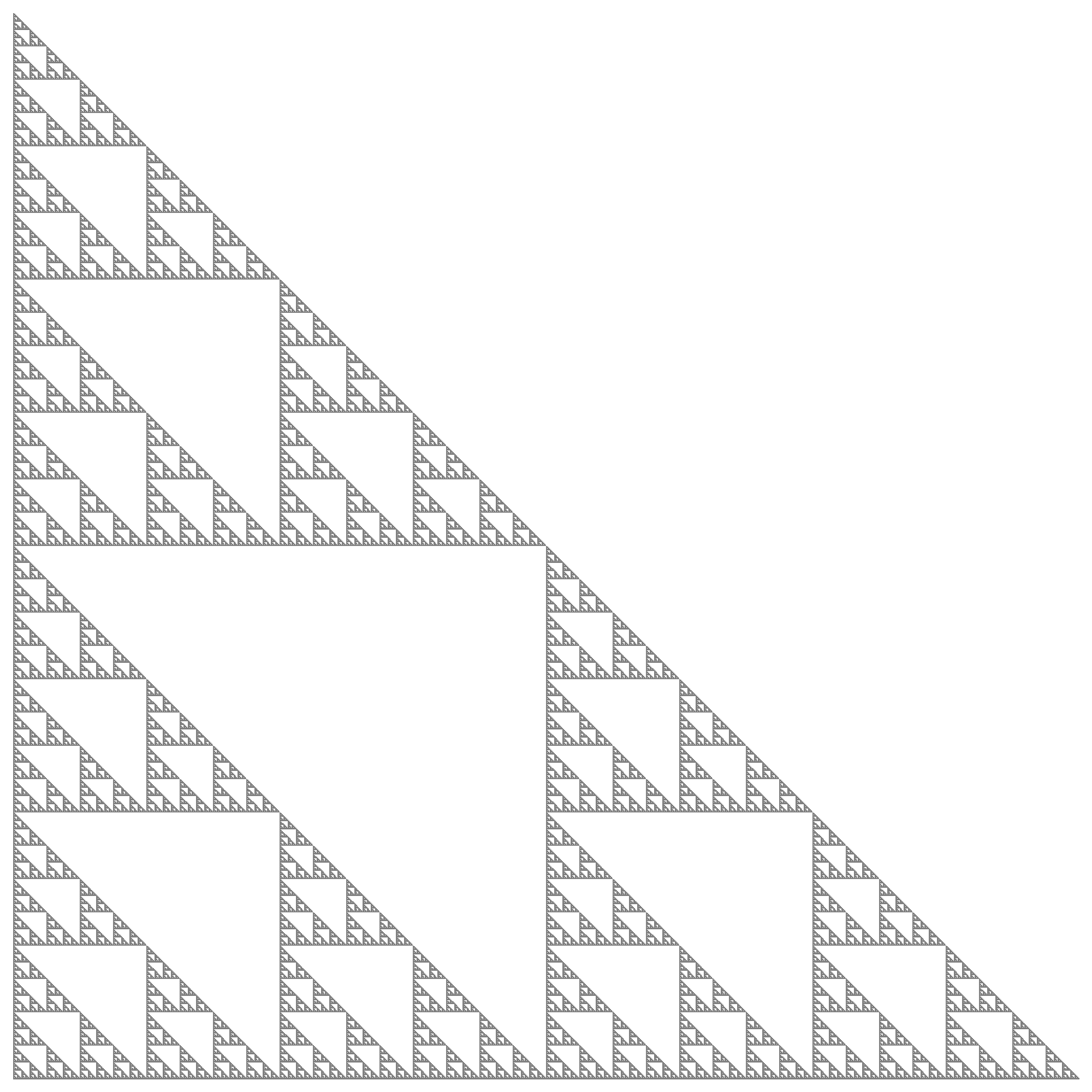}
  \caption{Right-angled Sierpinski triangle.}
\end{subfigure}

\begin{subfigure}{\textwidth}
  \centering
  \includegraphics[width=0.23\textwidth]{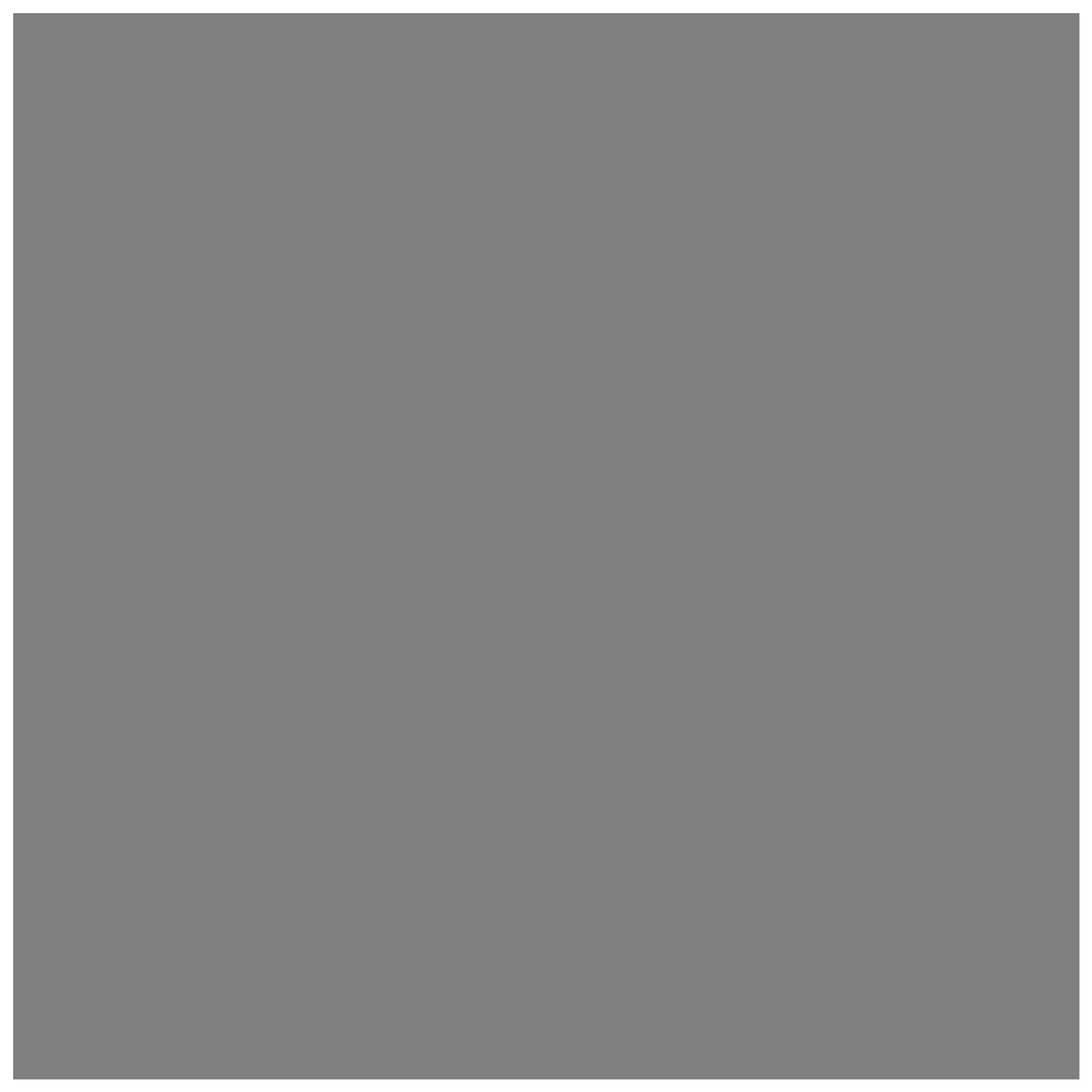}\hfill
  \includegraphics[width=0.23\textwidth]{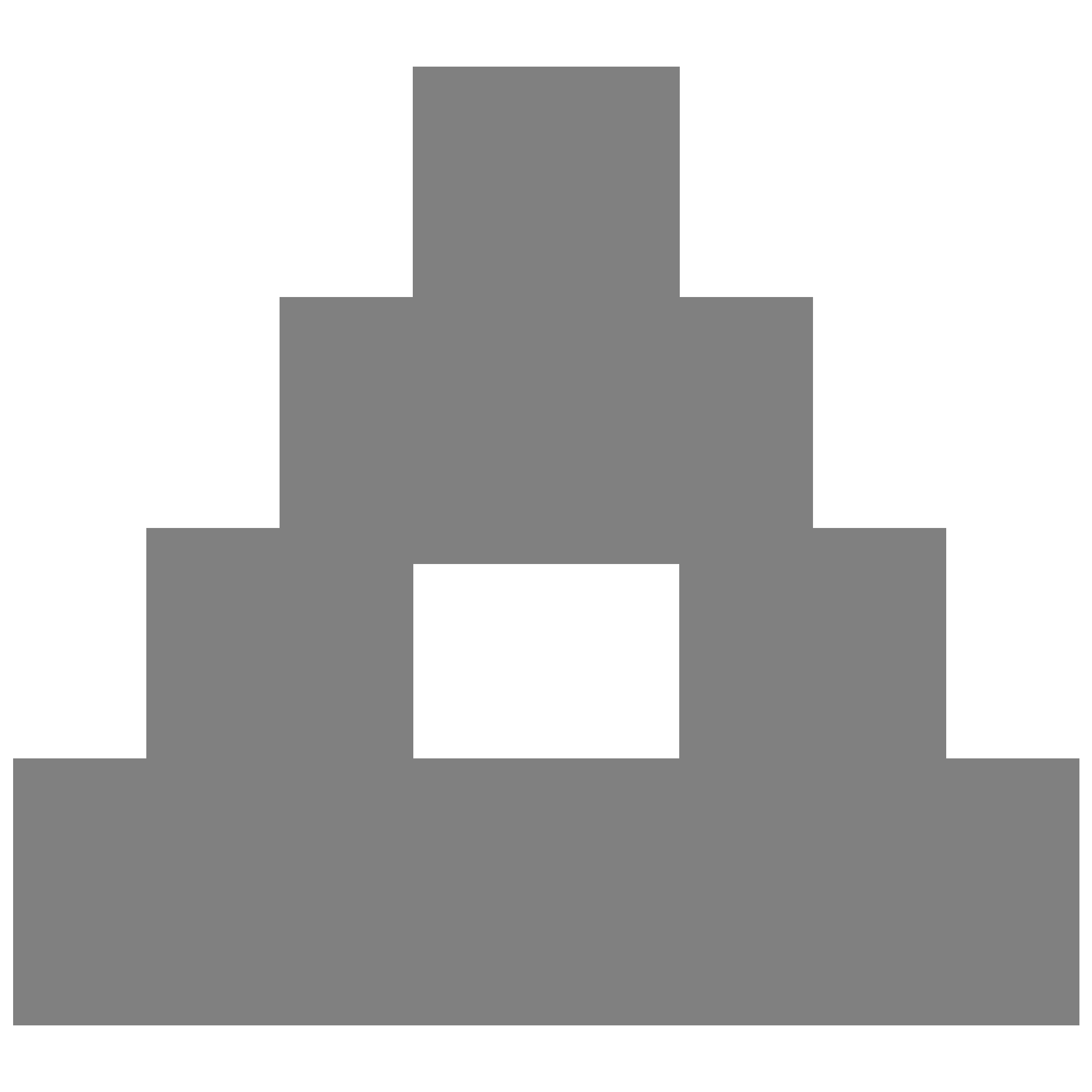}\hfill
  \includegraphics[width=0.23\textwidth]{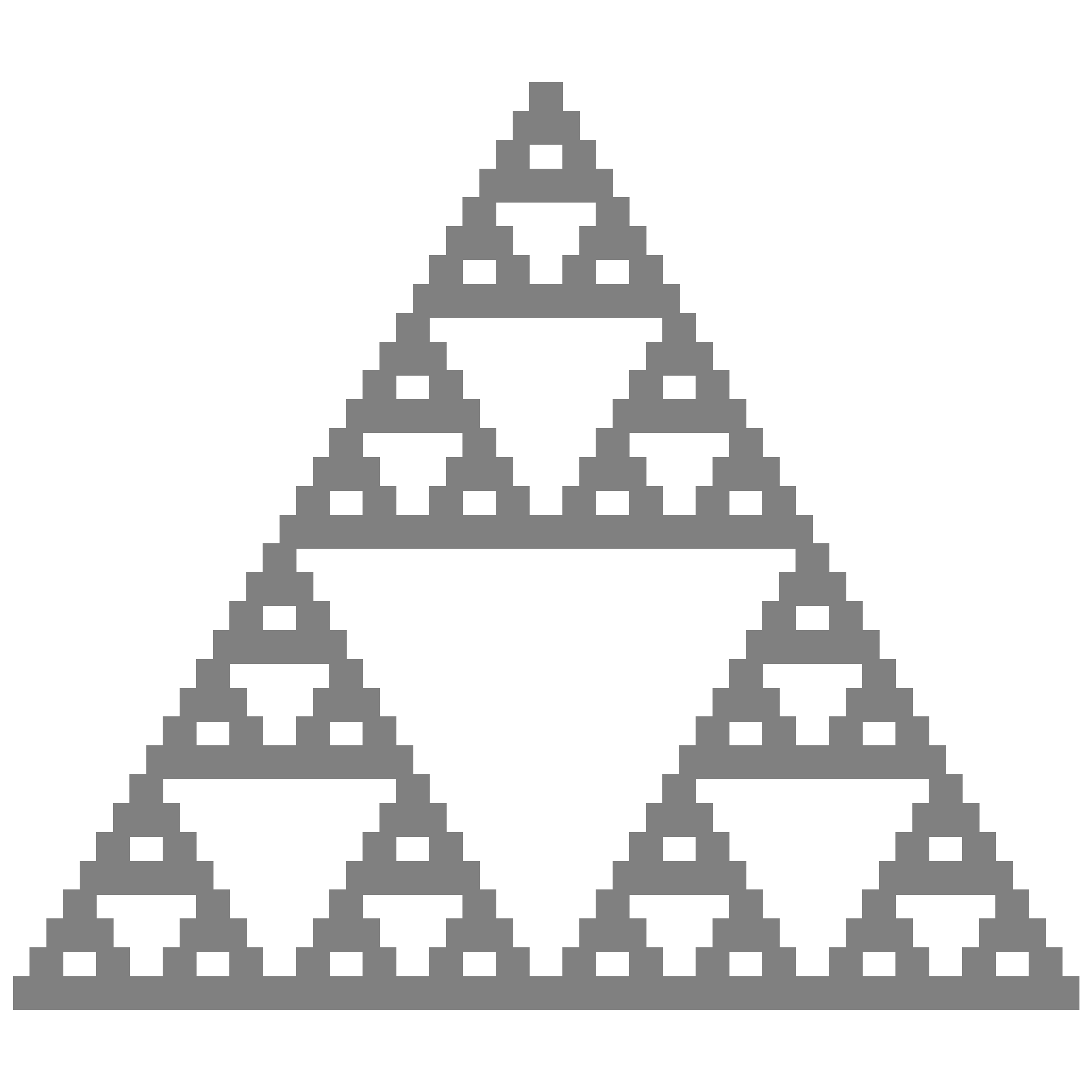}\hfill
  \includegraphics[width=0.23\textwidth]{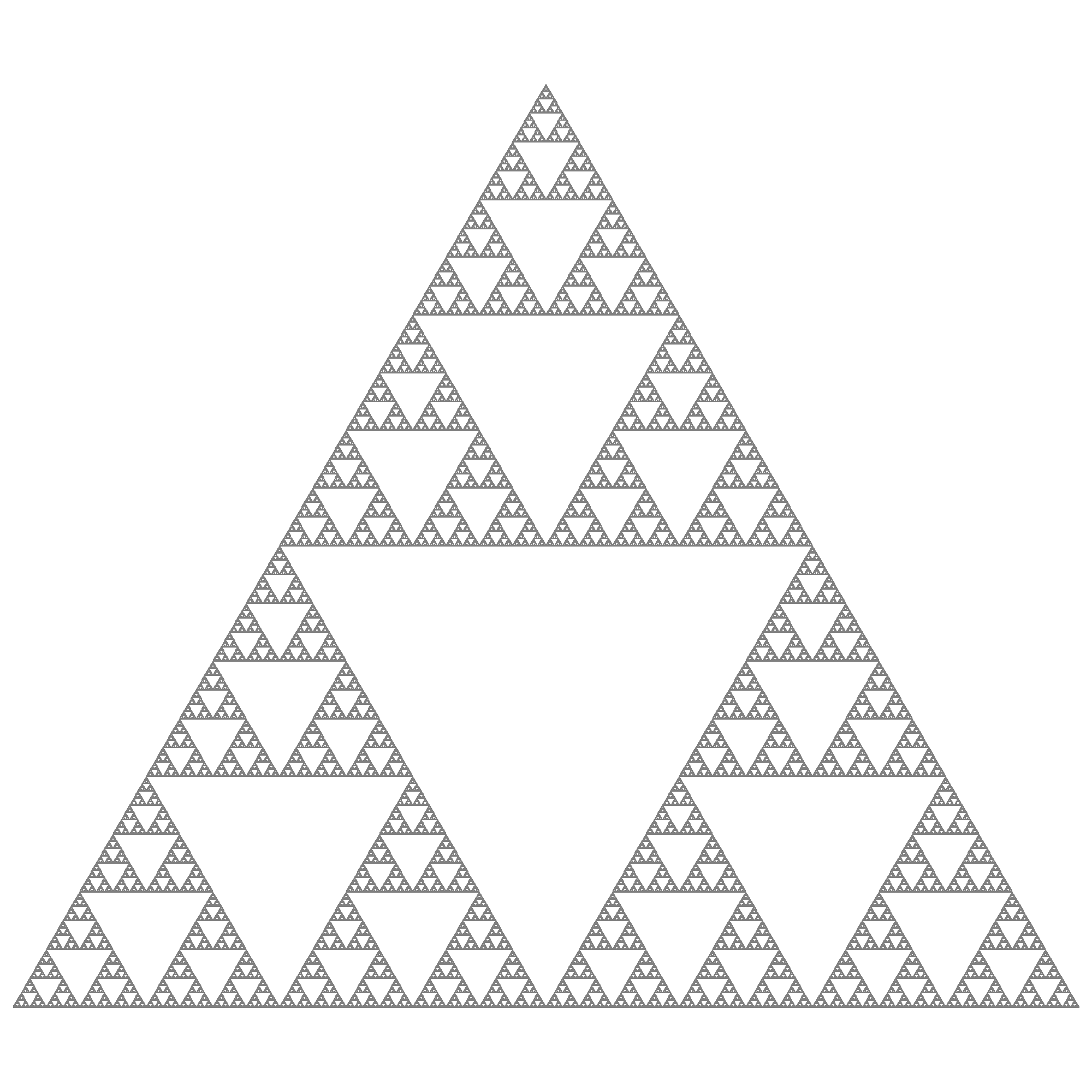}
  \caption{Equilateral Sierpinski triangle.}
\end{subfigure}

\begin{subfigure}{\textwidth}
  \centering
  \includegraphics[width=0.23\textwidth]{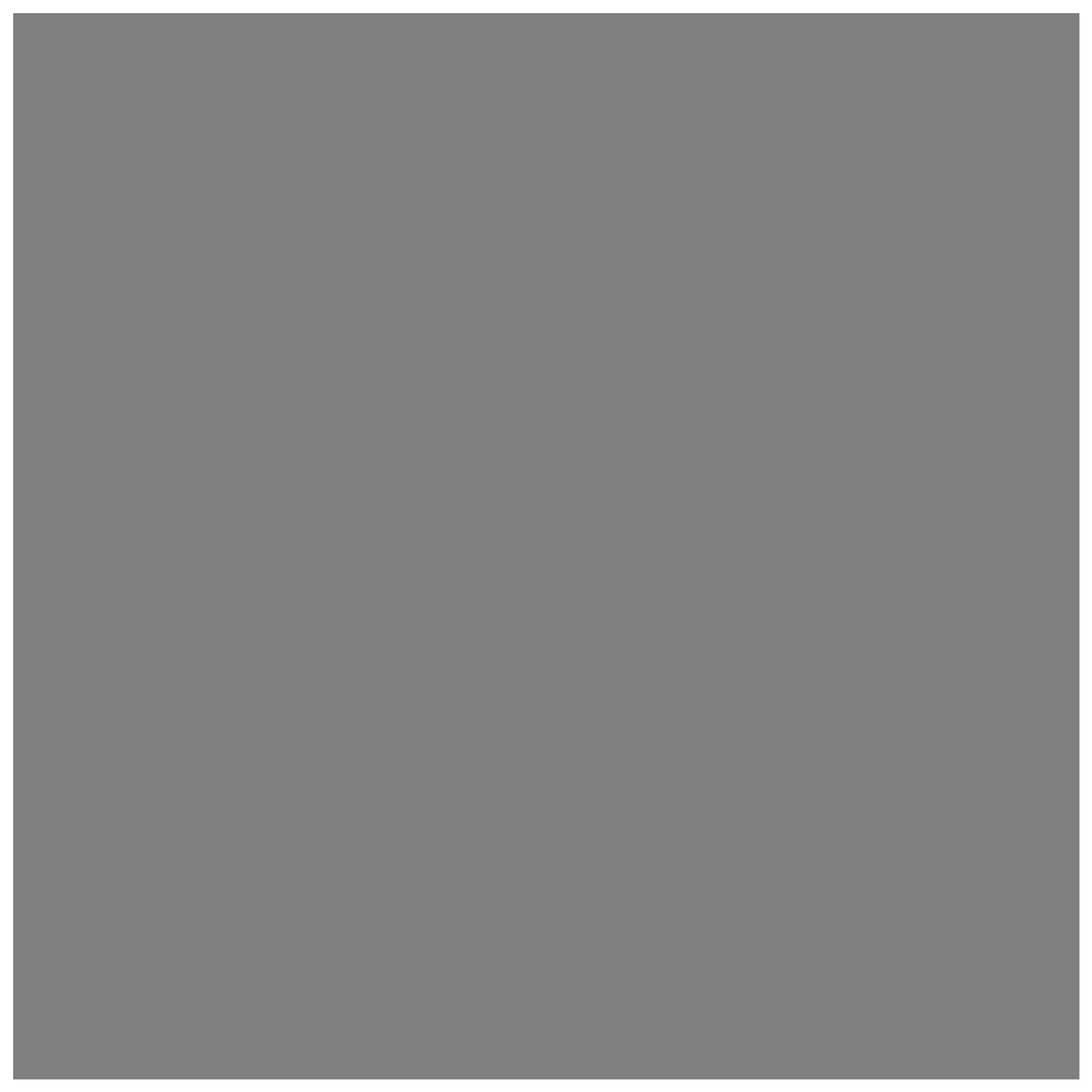}\hfill
  \includegraphics[width=0.23\textwidth]{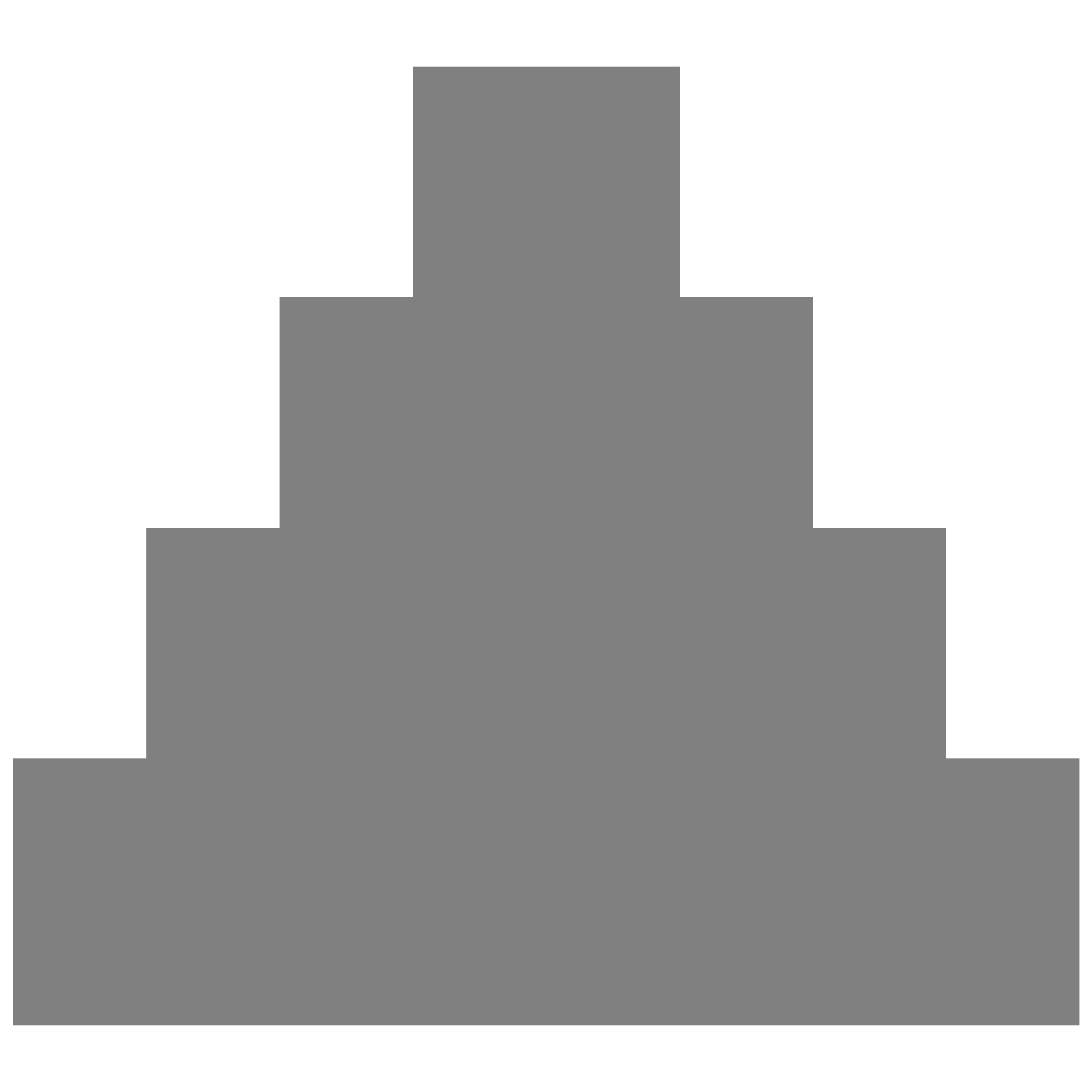}\hfill
  \includegraphics[width=0.23\textwidth]{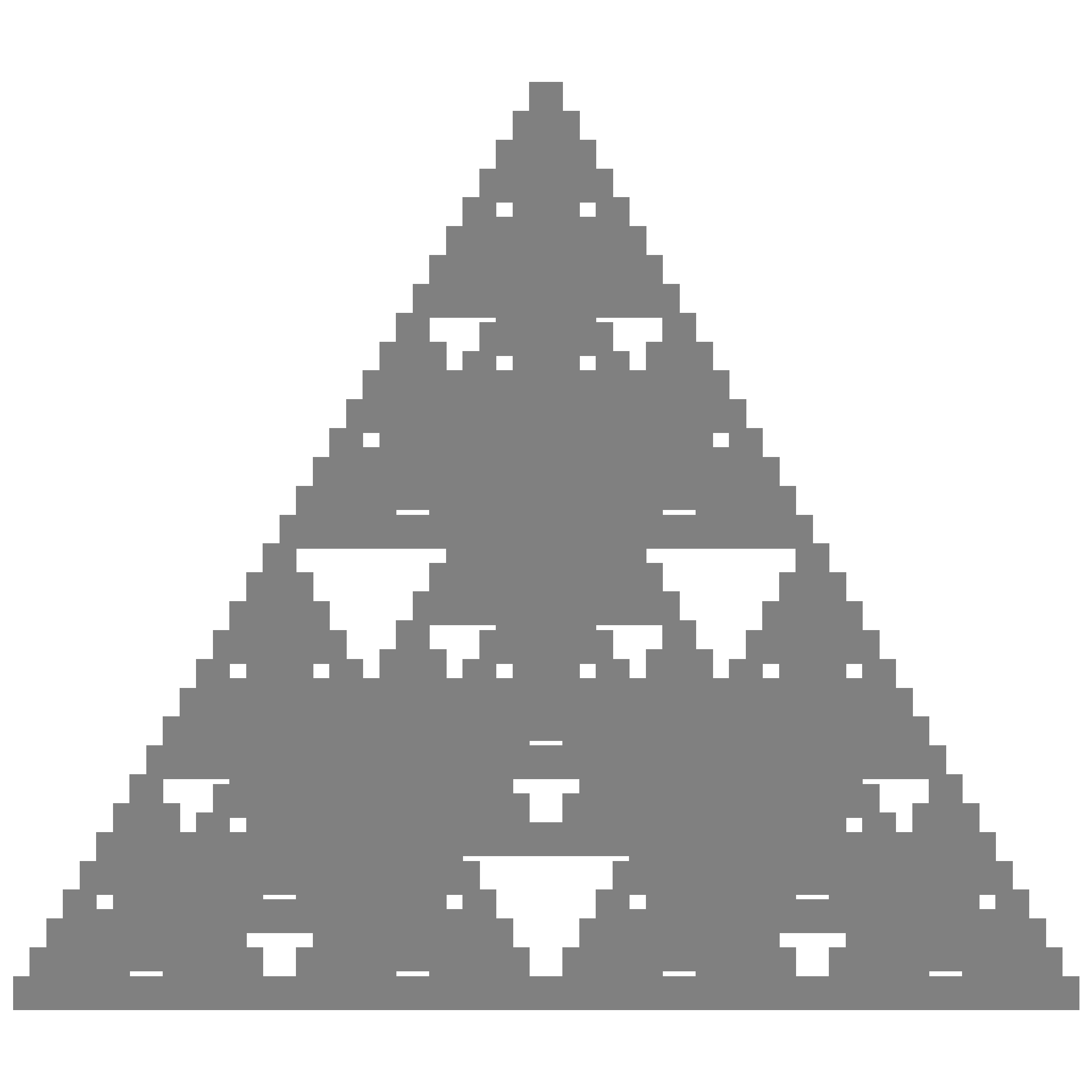}\hfill
  \includegraphics[width=0.23\textwidth]{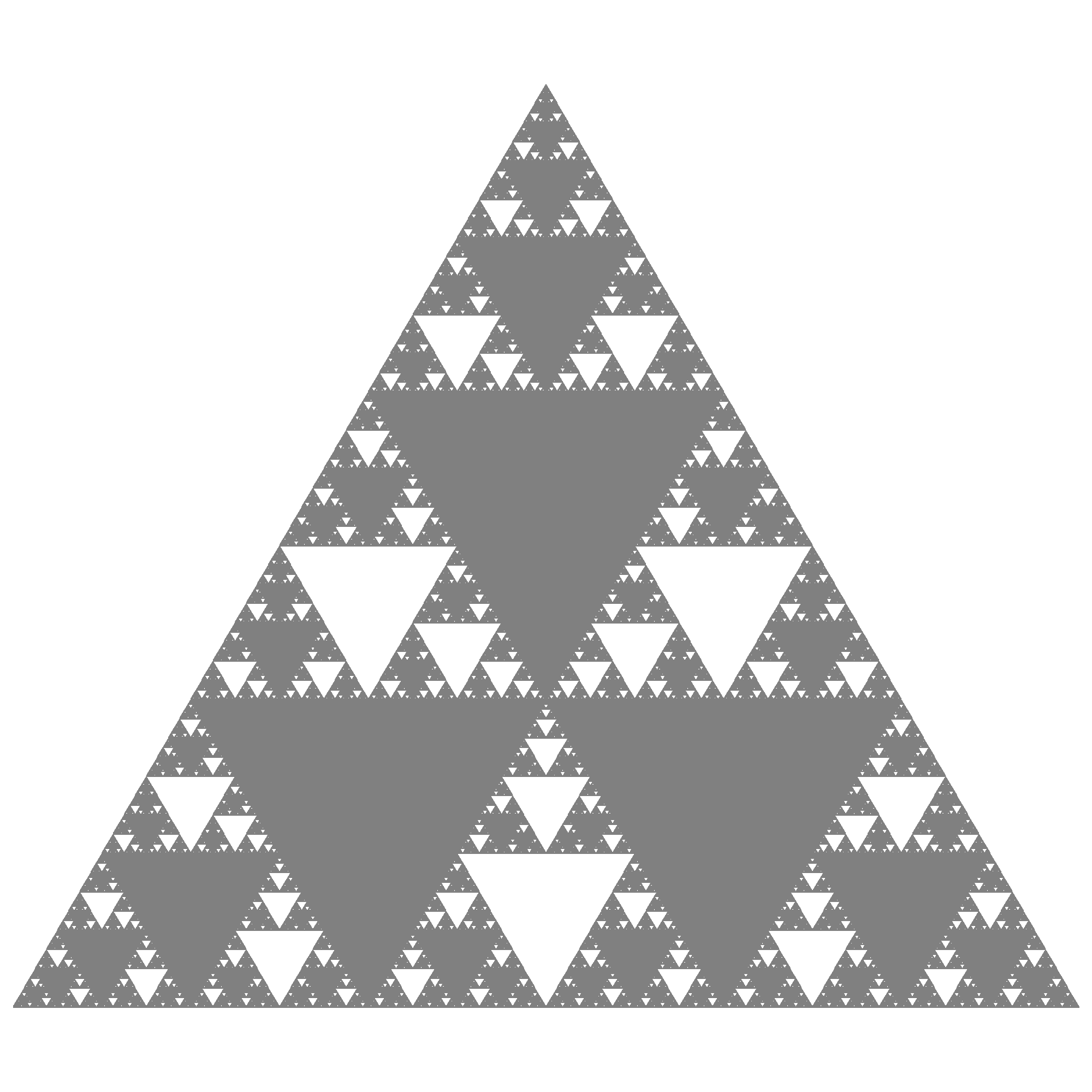}
  \caption{Equilateral Sierpinski triangle with a fourth map.}
\end{subfigure}
\caption{Approximations of three Sierpinski triangles for $n = 0, 2, 5, 10$,
  drawn with balls of radius $2^{-n-1}$ whose centers are approximated to
  within $2^{-53}$.}
  \label{fig:sierpinski-triangle}
\end{figure}

\section{Conclusion and Future Work}
We considered procedures to work with different classes of subsets in a computational setting.
The main focus of this work was on subsets of Polish spaces, although we also provide a framework for subsets of generic represented spaces.
Formal proofs in our system allow us to extract programs to compute several operations on subsets.
We consider simple operations on subsets to build new subsets and extract programs for certified drawings of compact-overt subsets up to any desired resolution using the fact that those subsets are totally bounded.

Our main example in this work is Euclidean space, for which we show how to draw simple fractals with any desired precision.
However, the general framework of Polish spaces allows for several other applications such as computation over function spaces, Hilbert spaces and so on, which we plan to address in future work.
This would also allow for extensions to more advanced operations like integration and solution operators for ordinary differential equations \cite{kawamura2018parameterized,nedialkov2006interval,thies2026ode}.

Another direction for future work is to address more interesting operations on Euclidean space by computing images of subsets under some class of functions.
While in the current work, we already showed how to compute images under arbitrary functions, the process itself is not very efficient and  thus unlikely to be used for actual computation.
Thus, considering more efficient algorithms for computing images of, e.g.,   polynomials would provide many more interesting applications.
One possibility that comes to mind is using polynomial models such as Taylor models \cite{makino2003taylor} to represent subsets of Euclidean space.
In follow-up work \cite{DBLP:conf/itp/0001T24} we already provide a simple implementation of Taylor models in our framework, albeit with another application in mind.
In future work we plan to combine the approach with the work presented in this paper.

\section*{Acknowledgment}
  \noindent 
  The authors thank the anonymous reviewers for their helpful comments and suggestions. They also thank Ulrich Berger, Pieter Collins, Arno Pauly and Hideki Tsuiki for interesting discussions.

\noindent
Sewon Park was supported by JSPS KAKENHI (Grant-in-Aid for JSPS Fellows) JP22F22071, by JST, CREST Grant Number JPMJCR21M3, and by the Slovenian Research and Innovation Agency award P1-0294.
Holger Thies was supported by JSPS KAKENHI Grant Numbers JP20K19744, JP23K28036, and JP24K20735, and by JSPS Bilateral Program Number JPJSBP120265001.
This project has received funding from the European Union's Horizon 2020 research and innovation programme under the Marie Sk{\l}odowska-Curie grant agreement No 731143.
This work is part of the Marie Sk{\l}odowska-Curie Staff Exchanges project NFFC (grant agreement No.~101236610).
Funded by the European Union \hbox{\includegraphics[height=1em]{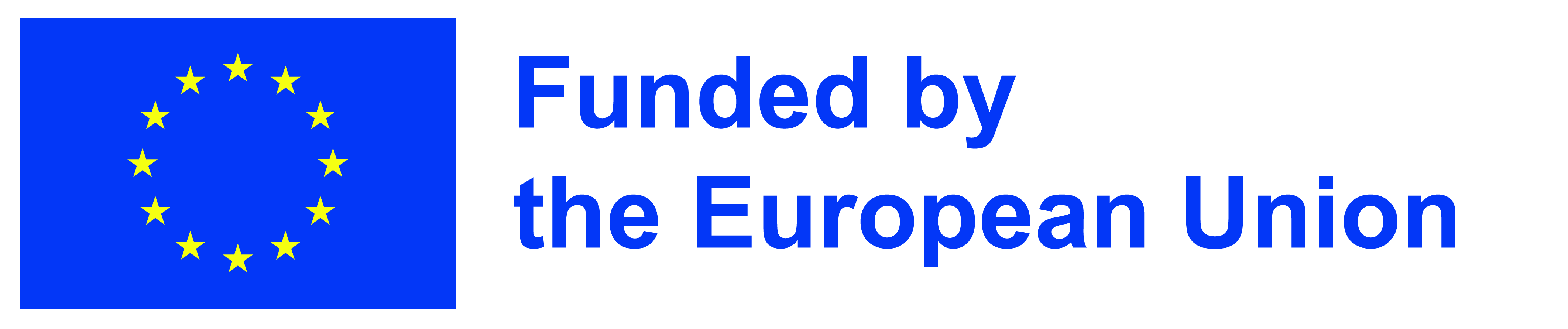}}.
Views and opinions expressed are however those of the author(s) only and do not necessarily reflect those of the European Union or the European Research Executive Agency.
Neither the European Union nor the granting authority can be held responsible for them.
  
\paragraph{Use of AI tools}
  \noindent 
Claude Opus 4.8 was used to assist in preparing the Python script used for the experimental evaluation. The authors reviewed and verified all generated code. No AI tools were used in the development of the Coq formalization.

\bibliographystyle{alphaurl}
\bibliography{refs}

\newcommand{\etalchar}[1]{$^{#1}$}
\begin{thebibliography}{CFGW04}

\bibitem[BCC{\etalchar{+}}06]{Ariadne}
A.~Balluchi, A.~Casagrande, P.~Collins, A.~Ferrari, T.~Villa, and A.L. Sangiovanni-Vincentelli.
\newblock {Ariadne: a Framework for Reachability Analysis of Hybrid Automata}.
\newblock In {\em Proc. 17th Int. Symp. on Mathematical Theory of Networks and Systems}, Kyoto, 2006.

\bibitem[BCZ22]{10.1007/978-3-031-14788-3_5}
Franz Brau{\ss}e, Pieter Collins, and Martin Ziegler.
\newblock Computer science for continuous data: Survey, vision, theory, and practice of a computer analysis system.
\newblock In Fran{\c{c}}ois Boulier, Matthew England, Timur~M. Sadykov, and Evgenii~V. Vorozhtsov, editors, {\em Computer Algebra in Scientific Computing}, volume 13366 of {\em Lecture Notes in Computer Science}, pages 62--82, Cham, 2022. Springer.
\newblock \href {https://doi.org/10.1007/978-3-031-14788-3_5} {\path{doi:10.1007/978-3-031-14788-3_5}}.

\bibitem[BH98]{BRATTKA1998490}
Vasco Brattka and Peter Hertling.
\newblock Feasible real random access machines.
\newblock {\em Journal of Complexity}, 14(4):490--526, 1998.
\newblock \href {https://doi.org/10.1006/jcom.1998.0488} {\path{doi:10.1006/jcom.1998.0488}}.

\bibitem[Bir99]{BIRKEDAL20002}
Lars Birkedal.
\newblock {\em Developing Theories of Types and Computability via Realizability}.
\newblock PhD thesis, Carnegie Mellon University, Pittsburgh, PA 15213, USA, December 1999.

\bibitem[Bis67]{bishop1967foundations}
Errett~Albert Bishop.
\newblock {\em Foundations of constructive analysis}.
\newblock McGraw-Hill, New York, 1967.

\bibitem[BIS99]{bridges1999sequential}
Douglas Bridges, Hajime Ishihara, and Peter Schuster.
\newblock Sequential compactness in constructive analysis.
\newblock {\em {\"O}sterreich. Akad. Wiss. Math.-Natur. Kl. Sitzungsber. II}, 208:159--163, 1999.

\bibitem[BIS02]{10.1007/3-540-45793-3_7}
Douglas Bridges, Hajime Ishihara, and Peter Schuster.
\newblock Compactness and continuity, constructively revisited.
\newblock In Julian Bradfield, editor, {\em Computer Science Logic}, pages 89--102, Berlin, Heidelberg, 2002. Springer Berlin Heidelberg.

\bibitem[BKM15]{DBLP:conf/ershov/BrausseKM15}
Franz Brau{\ss}e, Margarita~V. Korovina, and Norbert~Th. M{\"u}ller.
\newblock Towards using exact real arithmetic for initial value problems.
\newblock In Manuel Mazzara and Andrei Voronkov, editors, {\em Perspectives of System Informatics}, volume 9609 of {\em Lecture Notes in Computer Science}, pages 61--74, Cham, 2015. Springer.
\newblock \href {https://doi.org/10.1007/978-3-319-41579-6_6} {\path{doi:10.1007/978-3-319-41579-6_6}}.

\bibitem[BLM16]{BLM16}
Sylvie Boldo, Catherine Lelay, and Guillaume Melquiond.
\newblock Formalization of real analysis: A survey of proof assistants and libraries.
\newblock {\em Mathematical Structures in Computer Science}, 26(7):1196--1233, 2016.
\newblock \href {https://doi.org/10.1017/S0960129514000437} {\path{doi:10.1017/S0960129514000437}}.

\bibitem[BP03]{BRATTKA200343}
Vasco Brattka and Gero Presser.
\newblock Computability on subsets of metric spaces.
\newblock {\em Theoretical Computer Science}, 305(1):43--76, 2003.
\newblock Topology in Computer Science.
\newblock \href {https://doi.org/10.1016/S0304-3975(02)00693-X} {\path{doi:10.1016/S0304-3975(02)00693-X}}.

\bibitem[BP24]{clerical_ocaml}
Andrej Bauer and Sewon Park.
\newblock An implementation of {Clerical} in {OCaml}.
\newblock Available at \url{https://github.com/andrejbauer/clerical}, 2024.

\bibitem[BPS24]{clerical}
Andrej Bauer, Sewon Park, and Alex Simpson.
\newblock An imperative language for verified exact real-number computation, 2024.
\newblock \href {https://arxiv.org/abs/2409.11946} {\path{arXiv:2409.11946}}.

\bibitem[Bra05]{braverman2005complexity}
Mark Braverman.
\newblock On the complexity of real functions.
\newblock In {\em 46th Annual IEEE Symposium on Foundations of Computer Science (FOCS'05)}, pages 155--164. IEEE, 2005.

\bibitem[BT21]{DBLP:journals/apal/BergerT21}
Ulrich Berger and Hideki Tsuiki.
\newblock Intuitionistic fixed point logic.
\newblock {\em Annals of Pure and Applied Logic}, 172(3):102903, 2021.
\newblock \href {https://doi.org/10.1016/j.apal.2020.102903} {\path{doi:10.1016/j.apal.2020.102903}}.

\bibitem[BW99]{BRATTKA199965}
Vasco Brattka and Klaus Weihrauch.
\newblock Computability on subsets of {Euclidean} space {I}: closed and compact subsets.
\newblock {\em Theoretical Computer Science}, 219(1):65--93, 1999.
\newblock \href {https://doi.org/10.1016/S0304-3975(98)00284-9} {\path{doi:10.1016/S0304-3975(98)00284-9}}.

\bibitem[BY09]{braverman2009computability}
Mark Braverman and Michael Yampolsky.
\newblock {\em Computability of {Julia} sets}, volume~23 of {\em Algorithms and Computation in Mathematics}.
\newblock Springer, Berlin, Heidelberg, 2009.
\newblock \href {https://doi.org/10.1007/978-3-540-68547-0} {\path{doi:10.1007/978-3-540-68547-0}}.

\bibitem[CFGW04]{cruz2004c}
Lu{\'i}s Cruz-Filipe, Herman Geuvers, and Freek Wiedijk.
\newblock {C-CoRN}, the constructive {Coq} repository at {Nijmegen}.
\newblock In Andrea Asperti, Grzegorz Bancerek, and Andrzej Trybulec, editors, {\em Mathematical Knowledge Management}, volume 3119 of {\em Lecture Notes in Computer Science}, pages 88--103, Berlin, Heidelberg, 2004. Springer.
\newblock \href {https://doi.org/10.1007/978-3-540-27818-4_7} {\path{doi:10.1007/978-3-540-27818-4_7}}.

\bibitem[Col20]{collins2020}
Pieter Collins.
\newblock Computable analysis with applications to dynamic systems.
\newblock {\em Mathematical Structures in Computer Science}, 30(2):173–233, 2020.
\newblock \href {https://doi.org/10.1017/S096012952000002X} {\path{doi:10.1017/S096012952000002X}}.

\bibitem[CS05]{DBLP:journals/jucs/CoquandS05}
Thierry Coquand and Bas Spitters.
\newblock Formal topology and constructive mathematics: the {Gelfand} and {Stone-Yosida} representation theorems.
\newblock {\em Journal of Universal Computer Science}, 11(12):1932--1944, 2005.
\newblock \href {https://doi.org/10.3217/JUCS-011-12-1932} {\path{doi:10.3217/JUCS-011-12-1932}}.

\bibitem[CS10]{DBLP:journals/jla/CoquandS10}
Thierry Coquand and Bas Spitters.
\newblock Constructive theory of {Banach} algebras.
\newblock {\em Journal of Logic and Analysis}, 2, 2010.
\newblock \href {https://doi.org/10.4115/jla.2010.2.11} {\path{doi:10.4115/jla.2010.2.11}}.

\bibitem[Die08]{diener2008compactness}
Hannes Diener.
\newblock {\em Compactness under constructive scrutiny}.
\newblock PhD thesis, University of Canterbury. Mathematics and Statistics, 2008.

\bibitem[ES01]{escardoSimpson2001}
Martín~Hötzel Escardó and Alex Simpson.
\newblock A universal characterization of the closed {Euclidean} interval.
\newblock In {\em Proceedings 16th Annual IEEE Symposium on Logic in Computer Science}, pages 115--125, Boston, MA, USA, 2001. IEEE Computer Society.
\newblock \href {https://doi.org/10.1109/LICS.2001.932488} {\path{doi:10.1109/LICS.2001.932488}}.

\bibitem[Esc04]{escardo2004synthetic}
Mart{\'\i}n Escard{\'o}.
\newblock Synthetic topology: of data types and classical spaces.
\newblock {\em Electronic Notes in Theoretical Computer Science}, 87:21--156, 2004.
\newblock \href {https://doi.org/10.1016/j.entcs.2004.09.017} {\path{doi:10.1016/j.entcs.2004.09.017}}.

\bibitem[FG82]{FOURMAN1982107}
Michael~P Fourman and Robin~J Grayson.
\newblock Formal spaces.
\newblock In A.S. Troelstra and D.~{van Dalen}, editors, {\em The L. E. J. Brouwer Centenary Symposium}, volume 110 of {\em Studies in Logic and the Foundations of Mathematics}, pages 107--122. Elsevier, 1982.

\bibitem[GPS21]{DBLP:conf/csl/GassnerP021}
Christine Ga{\ss}ner, Arno Pauly, and Florian Steinberg.
\newblock Computing measure as a primitive operation in real number computation.
\newblock In Christel Baier and Jean Goubault{-}Larrecq, editors, {\em 29th {EACSL} Annual Conference on Computer Science Logic ({CSL} 2021)}, volume 183 of {\em LIPIcs}, pages 22:1--22:22. Schloss Dagstuhl -- Leibniz-Zentrum f{\"u}r Informatik, 2021.
\newblock \href {https://doi.org/10.4230/LIPIcs.CSL.2021.22} {\path{doi:10.4230/LIPIcs.CSL.2021.22}}.

\bibitem[Her99]{DBLP:journals/mlq/Hertling99}
Peter Hertling.
\newblock A real number structure that is effectively categorical.
\newblock {\em Math. Log. Q.}, 45:147--182, 1999.
\newblock \href {https://doi.org/10.1002/malq.19990450202} {\path{doi:10.1002/malq.19990450202}}.

\bibitem[Hof95]{10.1007/BFb0022273}
Martin Hofmann.
\newblock On the interpretation of type theory in locally cartesian closed categories.
\newblock In Leszek Pacholski and Jerzy Tiuryn, editors, {\em Computer Science Logic}, volume 933 of {\em Lecture Notes in Computer Science}, pages 427--441, Berlin, Heidelberg, 1995. Springer Berlin Heidelberg.
\newblock \href {https://doi.org/10.1007/BFb0022273} {\path{doi:10.1007/BFb0022273}}.

\bibitem[IJ24]{DBLP:journals/aml/IljazovicJ24}
Zvonko Iljazovi{\'{c}} and Matea Jeli{\'{c}}.
\newblock Computable approximations of a chainable continuum with a computable endpoint.
\newblock {\em Archive for Mathematical Logic}, 63(1-2):181--201, 2024.
\newblock \href {https://doi.org/10.1007/s00153-023-00891-5} {\path{doi:10.1007/s00153-023-00891-5}}.

\bibitem[IK21]{Iljazovi2021}
Zvonko Iljazovi{\'{c}} and Takayuki Kihara.
\newblock Computability of subsets of metric spaces.
\newblock In Vasco Brattka and Peter Hertling, editors, {\em Handbook of Computability and Complexity in Analysis}, pages 29--69. Springer International Publishing, Cham, 2021.
\newblock \href {https://doi.org/10.1007/978-3-030-59234-9_2} {\path{doi:10.1007/978-3-030-59234-9_2}}.

\bibitem[KLS59]{KreiselLacombeShoenfield1959}
Georg Kreisel, Daniel Lacombe, and Joseph~R. Shoenfield.
\newblock Partial recursive functionals and effective operations.
\newblock In Arend Heyting, editor, {\em Constructivity in Mathematics}, Studies in Logic and the Foundations of Mathematics, pages 290--297. North-Holland, Amsterdam, 1959.

\bibitem[Kon21]{konecny2008AERN}
Michal Konečný.
\newblock {aern2-real: A Haskell library for exact real number computation}.
\newblock \url{https://hackage.haskell.org/package/aern2-real}, 2021.

\bibitem[Kon22]{haskell-reals-comparison-2022}
Michal Konečný.
\newblock {A comparison of Haskell exact real number implementations}, July 2022.
\newblock \href {https://doi.org/10.5281/zenodo.13922094} {\path{doi:10.5281/zenodo.13922094}}.

\bibitem[KPT23]{DBLP:conf/mfcs/Konecny0T23}
Michal Kone{\v{c}}n{\'y}, Sewon Park, and Holger Thies.
\newblock Formalizing hyperspaces for extracting efficient exact real computation.
\newblock In J{\'e}r{\^o}me Leroux, Sylvain Lombardy, and David Peleg, editors, {\em 48th International Symposium on Mathematical Foundations of Computer Science ({MFCS} 2023)}, volume 272 of {\em LIPIcs}, pages 59:1--59:16. Schloss Dagstuhl -- Leibniz-Zentrum f{\"u}r Informatik, 2023.
\newblock \href {https://doi.org/10.4230/LIPIcs.MFCS.2023.59} {\path{doi:10.4230/LIPIcs.MFCS.2023.59}}.

\bibitem[KPT25]{konevcny2022extracting}
Michal Kone{\v{c}}ný, Sewon Park, and Holger Thies.
\newblock Extracting efficient exact real number computation from proofs in constructive type theory.
\newblock {\em Journal of Logic and Computation}, 35(6):exae066, 2025.
\newblock \href {https://doi.org/10.1093/logcom/exae066} {\path{doi:10.1093/logcom/exae066}}.

\bibitem[KST18]{kawamura2018parameterized}
Akitoshi Kawamura, Florian Steinberg, and Holger Thies.
\newblock Parameterized complexity for uniform operators on multidimensional analytic functions and {ODE} solving.
\newblock In {\em International Workshop on Logic, Language, Information, and Computation}, volume 10944 of {\em Lecture Notes in Computer Science}, pages 223--236. Springer, 2018.
\newblock \href {https://doi.org/10.1007/978-3-662-57669-4_13} {\path{doi:10.1007/978-3-662-57669-4_13}}.

\bibitem[KW85]{kreitz1985theory}
Christoph Kreitz and Klaus Weihrauch.
\newblock Theory of representations.
\newblock {\em Theoretical Computer Science}, 38:35--53, 1985.

\bibitem[Luc77]{LUCKHARDT1977321}
Horst Luckhardt.
\newblock A fundamental effect in computations on real numbers.
\newblock {\em Theoretical Computer Science}, 5(3):321 -- 324, 1977.
\newblock \href {https://doi.org/10.1016/0304-3975(77)90048-2} {\path{doi:10.1016/0304-3975(77)90048-2}}.

\bibitem[MB03]{makino2003taylor}
Kyoko Makino and Martin Berz.
\newblock {Taylor} models and other validated functional inclusion methods.
\newblock {\em International Journal of Pure and Applied Mathematics}, 6:239--316, 2003.

\bibitem[MM05]{MR2137733}
Val\'{e}rie M\'{e}nissier-Morain.
\newblock Arbitrary precision real arithmetic: design and algorithms.
\newblock {\em Journal of Logic and Algebraic Programming}, 64(1):13--39, 2005.
\newblock \href {https://doi.org/10.1016/j.jlap.2004.07.003} {\path{doi:10.1016/j.jlap.2004.07.003}}.

\bibitem[M{\"u}l01]{irram}
Norbert~Th. M{\"u}ller.
\newblock {The iRRAM: Exact arithmetic in C++}.
\newblock In Jens Blanck, Vasco Brattka, and Peter Hertling, editors, {\em Computability and Complexity in Analysis}, volume 2064 of {\em Lecture Notes in Computer Science}, pages 222--252, Berlin, Heidelberg, 2001. Springer.
\newblock Proceedings of CCA 2000.
\newblock \href {https://doi.org/10.1007/3-540-45335-0_14} {\path{doi:10.1007/3-540-45335-0_14}}.

\bibitem[Ned06]{nedialkov2006interval}
Nedialko~S Nedialkov.
\newblock Interval tools for {ODEs} and {DAEs}.
\newblock In {\em 12th GAMM-IMACS International Symposium on Scientific Computing, Computer Arithmetic and Validated Numerics (SCAN 2006)}, pages 4--15. IEEE, 2006.
\newblock \href {https://doi.org/10.1109/SCAN.2006.28} {\path{doi:10.1109/SCAN.2006.28}}.

\bibitem[O'C08]{DBLP:journals/corr/abs-0806-3209}
Russell O'Connor.
\newblock A computer verified theory of compact sets.
\newblock {\em CoRR}, abs/0806.3209, 2008.
\newblock \href {https://arxiv.org/abs/0806.3209} {\path{arXiv:0806.3209}}.

\bibitem[O'C09]{conorPhD}
Russell O'Connor.
\newblock {\em Incompleteness \& Completeness: Formalizing Logic and Analysis in Type Theory}.
\newblock PhD thesis, Radboud Universiteit Nijmegen, 2009.

\bibitem[Pau16]{Pauly2016OnTT}
Arno Pauly.
\newblock On the topological aspects of the theory of represented spaces.
\newblock {\em Computability}, 5(2):159--180, 2016.
\newblock \href {https://doi.org/10.3233/COM-150049} {\path{doi:10.3233/COM-150049}}.

\bibitem[PBC{\etalchar{+}}24]{park2016foundation}
Sewon Park, Franz Brauße, Pieter Collins, SunYoung Kim, Michal Konečný, Gyesik Lee, Norbert Müller, Eike Neumann, Norbert Preining, and Martin Ziegler.
\newblock {Semantics, Specification Logic, and Hoare Logic of Exact Real Computation}.
\newblock {\em {Logical Methods in Computer Science}}, {Volume 20, Issue 2}, June 2024.
\newblock \href {https://doi.org/10.46298/lmcs-20(2:17)2024} {\path{doi:10.46298/lmcs-20(2:17)2024}}.

\bibitem[PT24]{DBLP:conf/itp/0001T24}
Sewon Park and Holger Thies.
\newblock A {Coq} formalization of {Taylor} models and power series for solving ordinary differential equations.
\newblock In Yves Bertot, Temur Kutsia, and Michael Norrish, editors, {\em 15th International Conference on Interactive Theorem Proving ({ITP} 2024)}, volume 309 of {\em LIPIcs}, pages 30:1--30:19. Schloss Dagstuhl -- Leibniz-Zentrum f{\"u}r Informatik, 2024.
\newblock \href {https://doi.org/10.4230/LIPIcs.ITP.2024.30} {\path{doi:10.4230/LIPIcs.ITP.2024.30}}.

\bibitem[Sam87]{Sambin1987}
Giovanni Sambin.
\newblock Intuitionistic formal spaces --- a first communication.
\newblock In {\em Mathematical Logic and Its Applications}, pages 187--204. Springer US, Boston, MA, 1987.

\bibitem[SB23]{lmcs:11550}
Dieter Spreen and Ulrich Berger.
\newblock {Computing with Infinite Objects: the Gray Code Case}.
\newblock {\em {Logical Methods in Computer Science}}, {Volume 19, Issue 3}, July 2023.
\newblock \href {https://doi.org/10.46298/lmcs-19(3:1)2023} {\path{doi:10.46298/lmcs-19(3:1)2023}}.

\bibitem[Sch03]{DBLP:phd/dnb/Schroder03a}
Matthias Schr{\"{o}}der.
\newblock {\em Admissible representations for continuous computations}.
\newblock PhD thesis, University of Hagen, Germany, 2003.
\newblock URL: \url{https://d-nb.info/967433002}.

\bibitem[Spi10]{spitters2010locatedness}
Bas Spitters.
\newblock Locatedness and overt sublocales.
\newblock {\em Annals of Pure and Applied Logic}, 162(1):36--54, 2010.

\bibitem[SSTZ21]{selivanova2021exact}
Svetlana Selivanova, Florian Steinberg, Holger Thies, and Martin Ziegler.
\newblock Exact real computation of solution operators for linear analytic systems of partial differential equations.
\newblock In {\em Computer Algebra in Scientific Computing ({CASC} 2021)}, pages 370--390. Springer, 2021.
\newblock \href {https://doi.org/10.1007/978-3-030-85165-1_21} {\path{doi:10.1007/978-3-030-85165-1_21}}.

\bibitem[STT21]{steinberg2019quantitative}
Florian Steinberg, Laurent Thery, and Holger Thies.
\newblock {Computable analysis and notions of continuity in Coq}.
\newblock {\em {Logical Methods in Computer Science}}, {Volume 17, Issue 2}, May 2021.
\newblock \href {https://doi.org/10.23638/LMCS-17(2:16)2021} {\path{doi:10.23638/LMCS-17(2:16)2021}}.

\bibitem[Tay10]{TaylorP:lamcra}
Paul Taylor.
\newblock A lambda calculus for real analysis.
\newblock {\em Journal of Logic and Analysis}, 2(5):1--115, August 2010.
\newblock \href {https://doi.org/10.4115/jla.2010.2.5} {\path{doi:10.4115/jla.2010.2.5}}.

\bibitem[Tay11]{TaylorP:fofct}
Paul Taylor.
\newblock Foundations for computable topology.
\newblock In Giovanni Sommaruga, editor, {\em Foundational Theories of Classical and Constructive Mathematics}, number~76 in Western Ontario Series in Philosophy of Science. Springer-Verlag, January 2011.
\newblock URL: \url{https://paultaylor.eu/ASD/foufct}.

\bibitem[Thi26]{thies2026ode}
Holger Thies.
\newblock Computing solutions for systems of multivariate ordinary differential equations in {Rocq}.
\newblock In {\em Proceedings of the 15th ACM SIGPLAN International Conference on Certified Programs and Proofs (CPP 2026)}, pages 29--44. ACM, 2026.
\newblock \href {https://doi.org/10.1145/3779031.3779097} {\path{doi:10.1145/3779031.3779097}}.

\bibitem[TKP26]{caern}
Holger Thies, Michal Kone{\v{c}}n{\'y}, and Sewon Park.
\newblock The c{AERN} library, July 2026.
\newblock \href {https://doi.org/10.5281/zenodo.21349522} {\path{doi:10.5281/zenodo.21349522}}.

\bibitem[Vie22]{vietoris1922bereiche}
Leopold Vietoris.
\newblock Bereiche zweiter {O}rdnung.
\newblock {\em Monatshefte f{\"u}r Mathematik und Physik}, 32:258--280, 1922.

\bibitem[VO08]{van2008realizability}
Jaap Van~Oosten.
\newblock {\em Realizability: an introduction to its categorical side}.
\newblock Elsevier, 2008.

\bibitem[Wei00]{w00}
Klaus Weihrauch.
\newblock {\em Computable Analysis: An Introduction}.
\newblock Texts in Theoretical Computer Science. An EATCS Series. Springer, Berlin, Heidelberg, 2000.
\newblock \href {https://doi.org/10.1007/978-3-642-56999-9} {\path{doi:10.1007/978-3-642-56999-9}}.

\bibitem[Xu15]{xuthesis}
Chuangjie Xu.
\newblock {\em A continuous computational interpretation of type theories}.
\newblock PhD thesis, University of Birmingham, 2015.
\newblock URL: \url{https://etheses.bham.ac.uk/id/eprint/5967/}.

\end{thebibliography}

\end{document}